\documentclass[11pt,oldfontcommands,a4paper,twoside,openright]{memoir}

\settrimmedsize{297mm}{210mm}{*}
 \settypeblocksize{634pt}{448.13pt}{*}
 \setulmargins{4cm}{*}{*}
 \setlrmargins{*}{*}{1.5}
 \setmarginnotes{17pt}{51pt}{\onelineskip}
 \setheadfoot{\onelineskip}{2\onelineskip}
 \setheaderspaces{*}{2\onelineskip}{*}
 \checkandfixthelayout
 
 \usepackage[colorlinks=true,linkcolor=blue,citecolor=red,pdfpagelabels]{hyperref}

\usepackage{multibib}

\usepackage[stable]{footmisc}
\usepackage{verbatim}

\usepackage{longtable}
\usepackage{multirow}
\usepackage{colortbl}

\usepackage{booktabs}
\usepackage{braket}
\usepackage{amsmath}
\usepackage{bbm}
\usepackage{amsthm}

\usepackage[utf8]{inputenc}

\definecolor{chaptercolor}{gray}{0.8}

\usepackage[adobe-utopia]{mathdesign}
\usepackage[T1]{fontenc}

\setsecnumdepth{subsubsection}

\chapterstyle{pedersen}

\definecolor{tud1b}{RGB}{0,90,169}
\definecolor{tud9b}{RGB}{230,0,26}
\definecolor{tud6a}{RGB}{255,224,92}

\newcites{me}{Included Papers}

\usepackage{tikz}
\usepackage{wasysym}
\usetikzlibrary{backgrounds,fit,decorations.pathreplacing,positioning,calc}

\newtheorem{theorem}{Theorem}
\newtheorem{lemma}{Lemma}
\newtheorem{corollary}{Corollary}

\newcommand{\bx}{\ensuremath{\mathbf{x}}}
\newcommand{\bz}{\ensuremath{\mathbf{z}}}
\newcommand{\td}[2]{\ensuremath{\tfrac12\Big\|#1-#2\Big\|_1}}
\newcommand{\id}{\ensuremath{\mathbbm{1}}}
\newcommand{\idb}{\makebox[8pt][c]{\id}}
\newcommand{\idx}{\makebox[8pt][c]{\id}}

\newcommand{\etal}{~\emph{et al.}\ }
\newcommand{\etalsp}{~\emph{et al.}~}

\hypersetup{%
  pdftitle={The Physics of Quantum Information},
  pdfauthor={Joseph M. Renes},
  pdfsubject={quantum information},
  pdfview=FitH,
  pdfstartview=FitV
}

\newlength{\longtablewidth}
\title{\Huge The Physics of Quantum Information\\
{\LARGE Complementarity, Uncertainty, and Entanglement}}
\author{Joseph M.\ Renes}
\date{12 December 2012}

\begin{document}

\frontmatter

\setcounter{page}{1}
\pagenumbering{roman}

\begin{titlingpage}
 \aliaspagestyle{titlingpage}{empty}
\setlength{\droptitle}{70pt}
\maketitle

\vspace{8cm}
\centering
Habilitation\hspace{1cm} TU Darmstadt, Fachbereich Physik\hspace{1cm} 2011
\end{titlingpage}

\cleardoublepage

\tableofcontents*

\cleardoublepage

\setcounter{page}{1}
\pagenumbering{arabic}

\mainmatter

\chapter*{Preface}
\addcontentsline{toc}{chapter}{Preface}
\markboth{Preface}{Preface}

Complementarity is one of the central mysteries of quantum mechanics. First put forth by Bohr~\cite{bohr_quantenpostulat_1928,bohr_quantum_1928,bohr_discussion_1949}, complementarity holds that the attributes of a physical system familiar from classical mechanics do not all simultaneously exist and are not entirely independent of how they are measured. Famously, if the momentum of a particle is known then its position must be unknown, and \emph{vice versa}, a fact encapsulated in Heisenberg's uncertainty relation $\Delta x\Delta p\geq\hbar/2$~\cite{heisenberg_ueber_1927}. Even more dramatic is the wave-particle duality encountered in Young's double-slit experiment, which illustrates the important role of observation. Light passing through the double slit setup produces an interference pattern on a screen beyond the slits, as would be characteristic of a wave. But a closer examination reveals that light arrives in particle-like ``packets'' at the screen, and the interference pattern only arises as a statistical average of these particle arrival events. This particle picture tempts us to observe which slit the light went through, which we find destroys the interference pattern! Feynman regarded this bizarre phenomena as characteristic of all the seemingly-paradoxical quantum behavior, claiming that the double-slit experiment is ``impossible, \emph{absolutely impossible} to describe classically, [and which] has in it the heart of quantum mechanics'', and that ``in reality, it contains the \emph{only} mystery'' (emphasis original)~\cite{feynman_feynman_1970}.

The overarching goal of this thesis is to demonstrate that complementarity is also at the heart of quantum information theory, that it allows us to make (some) sense of just what information "quantum information" refers to, and that it is useful in understanding and constructing quantum information processing protocols. The detailed research results which form the basis of these claims are to be found in the included papers, and the aim here is to present an overview comprehensible to a more general audience.\footnote{The included papers are referenced in alphabetical style, while references to other works are numeric.} 

As we shall see in Chapter~\ref{chap:intro}, quantum information can heuristically be thought of as a kind of combination of two types of normal ``classical'' information, specifically, classical information about the result of measuring one of two complementary observables. Due to the uncertainty principle, we can expect both pieces of information are not simultaneously realizable, and indeed the uncertainty principle will play a central quantitative role throughout this work. Particularly relevant will be the entropic uncertainty relation of \citeme{renes_conjectured_2009} and its generalization in \citeme{berta_uncertainty_2010}, which state that the more that can be known by one party about one observable, the less can be known by another party about a complementary observable.
That complementary observables play an important role in quantum information theory is not new to this thesis, and Chapter~\ref{chap:ill} discusses several fundamental quantum information processing tasks based on their use, such as teleportation and quantum error-correction. This chapter also provides some relevant formal background for the remainder of this work and establishes the notation used herein. 

Chapter~\ref{chap:char} begins the overview of the new results obtained in the included papers. Here we show that information about complementary observables not only plays an important role, but indeed a central one, and that possession of both complementary pieces of classical information is strictly equivalent to the existence of entanglement between the physical system the information pertains to and the system in which the information is stored. Moreover, the uncertainty principle provides a dual characterization, saying that entanglement between these two systems exists when the ``environment'', i.e.\ any and all other degrees of freedom, has no information about either complementary observable. Both characterizations can be modified to describe secret keys useful in cryptography instead of entangled states. Because Chapter~\ref{chap:char} gathers and mixes results from several of the included papers, it is entirely self-contained, whereas subsequent chapters do not go into as much detail. 

In Chapter~\ref{chap:proc} we show that this complementary approach is also useful in constructing quantum information processing protocols and understanding why they work. Especially relevant is the process of entanglement distillation, that is, extracting maximal entanglement from a imperfectly-entangled bipartite resource system. The entanglement distillation process can be built up from two instances, one for each of two complementary observables, of a simpler distillation process for classical information called information reconciliation or data compression with side information. Here partial classical correlation between two systems is refined into maximal correlation, and reconciling classical information about two complementary observables. Protocols for entanglement distillation can then be adapted to a large variety of quantum information processing tasks, such as quantum communication over noisy channels or distillation of secret keys.

Chapter~\ref{chap:duality} extends the duality in characterizing entanglement afforded by the uncertainty principle to two fundamental information processing tasks, the information reconciliation task of establishing  correlations with the first party on the one hand, and the task of \emph{removing} all correlations with the second party on the other. 
The latter is known as privacy amplification, and it turns out that the ability to perform one protocol implies the ability to perform the other in certain circumstances. This duality also implies alternative methods of entanglement distillation, in particular one which proceeds by destroying all classical correlations with the environment that pertain to two complementary observables. We shall also see that information reconciliation and privacy amplification can be combined to enable classical communication over noisy quantum channels. 

Finally, Chapter~\ref{chap:qkd} describes the usefulness of this approach to establishing the security of quantum key distribution (QKD). QKD is perhaps the most natural setting in which the uncertainty principle and corresponding issues of complementarity are immediately relevant, as the goal of this protocol is to establish a secret key between two spatially-separated parties, a shared piece of classical information which no one else should know. Since the uncertainty principle can be understood as a limitation on who can know how much about what sorts of information, we shall see that complementarity-based arguments form the basis for the security of QKD  protocols. These allow us to increase the security threshold, the maximum amount of tolerable noise, of several protocols beyond the previously-known values.   

The following table summarizes which included papers form the basis for the various sections.

\begin{table}[h]
\begin{center}
\begin{tabular}{lcccc}
Chapter & \multicolumn{4}{c}{Section}\\[5mm]
\rowcolor{lightgray!40} 1 Introduction & 1.1 & 1.2 & 1.3 &\\
& & \multicolumn{2}{l}{\citeme{renes_conjectured_2009}, \citeme{berta_uncertainty_2010}}\\[4mm]
\rowcolor{lightgray!40} 2 Illustrations \& Motivations & & --- &&\\
&&&&\\
\rowcolor{lightgray!40} 3 Characterizing Quantum Information & 3.1 & 3.2 & 3.3 & 3.4\\
& \citeme{renes_physical_2008} & \citeme{renes_duality_2010} & \citeme{renes_conjectured_2009} & \citeme{renes_physical_2008}\\[4mm]
\rowcolor{lightgray!40} 4 Processing Quantum Information & 4.1 & 4.2 & 4.3 &\\
& \citeme{renes_physical_2008} & \citeme{boileau_optimal_2009} & \citeme{renes_physical_2008}\\[4mm]
\rowcolor{lightgray!40} 5 Duality of Protocols & 5.1 & 5.2 & 5.3 &\\
& \citeme{renes_duality_2010} & new & \citeme{renes_noisy_2010}\\[4mm]
\rowcolor{lightgray!40} 6 Security of QKD & 6.1 & 6.2 & 6.3 &\\
& --- & \citeme{renes_generalized_2006} & \citeme{renes_noisy_2007}\\
&&&\citeme{smith_structured_2008}\\&&&\citeme{kern_improved_2008}
\end{tabular}
\end{center}
\end{table}

\chapter{Introduction: What is Quantum Information?}
\label{chap:intro}

At a stroke, Shannon's landmark 1948 publication \emph{A Mathematical Theory of Communication}~\cite{shannon_mathematical_1948} established the field of information theory, laying out the fundamental lines of inquiry and answering some of the important basic questions. The fundamental problem, according to Shannon, ``is that of reproducing at one point either exactly or approximately a message selected at another point.'' The different points may be different places, in which case we are interested in transmitting messages from one party to another, such as in a telephone conversation, or they could be different times, and the message should be reliably stored, such as on a sheet of paper. The physical systems used to convey the message carry \emph{information}, which is measured by the entropy in units of bits, short for binary digits.\footnote{Interestingly, Vannevar Bush had already used the phrase of  ``bits of information'' in 1936 to describe information encoded into punchcards~\cite{bush_instrumental_1936}, though his meaning is different from Shannon's.}

The fact that abstract information must always be instantiated in some physical system and that this results in a connection between physics and information theory was stressed by Landauer. He observed this implies that logically irreversible operations, like erasure of information, are therefore physically irreversible and must be driven by a source of energy~\cite{landauer_irreversibility_1961, landauer_information_1991}. This was later used to resolve the paradox of Maxwell's Demon in which an intelligent being can apparently violate the second law by sorting the molecules of a gas into hot (fast) and cold (slow)~\cite{leff_maxwells_2002}. Building on Szil\'ard's simplification of the paradox to a one-atom gas occupying either the left or right side of a divided container~\cite{szilard_ueber_1929}, Bennett showed that the work gained by the demon is precisely balanced by the work needed to reset the demon's memory in a cyclic process~\cite{bennett_thermodynamics_1982}. It should be noted that Szil\'ard's simplification of the problem to a gas occupying one of two nearly anticipates the information-theoretic idea of a bit, also demonstrating the connections between these two fields.

The field of quantum information grew out of this connection by asking the question: What happens to information processing and information theory in general when the information carriers are described by quantum mechanics? 
One immediate implication is the possibility of quantum superpositions of information states of a bit. Instead of just the usual $0$s and $1$s, which might be encoded quantum mechanically as $\ket{0}$ and $\ket{1}$, we can also have states of the form $\alpha\ket{0}+\beta\ket{1}$ for $\alpha,\beta\in\mathbbm{C}$ and $|\alpha|^2+|\beta|^2=1$. This change in structure requires us to reexamine the entirety of Shannon's information theory, rather than being able to only slightly modify the results to account for quantum effects, as pointed out by Ingarden~\cite{ingarden_quantum_1976}: ``The old theory [Shannon's theory] cannot be improved only by inserting into it some quantum formulae.''\footnote{Ingarden also gives a very lucid description of the historical development of quantum information theory for the interested reader.}

By now, a new, explicitly quantum information theory has been constructed by asking many of the same questions as before, but answering them with the tools and methods of quantum mechanics; see for instance the textbook of Nielsen and Chuang~\cite{nielsen_quantum_2000}. It has also been possible to adapt many of the techniques of usual, \emph{classical} information theory to the quantum setting. For instance, Schumacher's result that quantum information emitted from a source can be compressed at a rate equal to the von Neumann entropy of the source follows Shannon's original result quite closely~\cite{schumacher_quantum_1995}. Nevertheless, in contrast to the classical case, we are still left with the question of what quantum information is information about. 


The core theme of this thesis is that quantum information is in a certain sense a combination of two pieces of classical information, information about two physical observables which are \emph{complementary} in the sense first put forth by Bohr~\cite{bohr_quantenpostulat_1928,bohr_quantum_1928,bohr_discussion_1949} and exemplified by the wave-particle duality in the double-slit experiment~\cite{feynman_feynman_1970}. Moreover, this point of view is useful in understanding and constructing protocols in quantum information theory. To appreciate this view of quantum information more clearly, the focus of this chapter, it is useful to first make the notions of classical information concrete in the following exceedingly simple game, the information game.

\section{Understanding Classical Information via the Information Game}

The information game has two players, Alice and Bob, and begins with Bob placing a coin, either heads or tails, in a box, and giving the box to Alice. At some point later she asks Bob whether she will see heads or tails when she opens the box. Bob's goal is to win the game by correctly matching Alice's observation. 

Is there a strategy with which Bob can always win the game? Of course. For instance, Bob could always place the coin heads up in the box and answer ``heads'' whenever Alice comes asking. He could also just randomly place the coin heads up or down in the box, as long as he remembers which it was when Alice asks; this task of remembering is precisely Shannon's fundamental problem. To solve it, Bob could just write down ``heads'' or ``tails'' on a piece of paper and save it for later. In this sense, the paper carries information about the coin, in particular about what Alice will observe when she opens the box. Because there are two equally-likely possibilities, Bob could just as well use one binary digit, a zero or one, to remember the state of the coin. Therefore the paper carries one bit of information. 

Formally, Bob's choice of the state of the coin can be represented as a binary-valued random variable $X$, taking on the values ``heads'' and ``tails'' with whatever probabilities $p_{\rm heads}$ and $p_{\rm tails}=1-p_{\rm heads}$ he decides. The state of the memory system he uses to remember the state of the coin can likewise be represented by a random variable, $M$, and a winning strategy simply has $M=X$ for any choice of $X$.

The amount of information stored the memory can be quantified by the Shannon entropy, defined for an arbitrary random variable $Y$ as 
\begin{align}
H(Y)=-\sum_{y} p_y \log p_y,
\end{align}
using $\log=\log_2$ to measure in bits, a choice we shall make henceforth. The entropy of a random variable $Y$ quantifies its uncertainty and is equal to the expected number of binary (yes/no) questions one would need to ask about $Y$ in order to determine its actual value $y$~\cite{mackay_information_2002}. A more concentrated distribution is less uncertain and makes guessing easier, and therefore has lower entropy, whereas the uniform distribution has maximum entropy and requires the most questions. 

To win the game, the contents of the memory must determine the state of the coin, and thus contain information equal to the entropy of the coin $H(X)$. Thus, for the original winning strategy no information is stored in the memory at all---the memory is not even needed---as the coin \emph{always} shows heads. Correspondingly, the entropy of this distribution is zero. In the second strategy, the memory stores one bit of information, since the coin is placed randomly in the box and $H(X)=1$. For distributions in between these two limiting cases, we can imagine many playing many rounds of the game and the entropy gives the ratio of number of questions needed to number of rounds. For the distribution $p_{\rm heads}=\frac 78$, $p_{\rm tails}=\frac 18$, which has entropy $H(X)=3-\frac{7}{8}\log_2 7\approx 0.54$, only 54 questions would be needed to determine the state of the coin in 100 rounds of play. In this case each memory register stores roughly one-half a bit of information.

On the other hand, given the value stored in the memory, the entropy of the coin random variable $X$ is zero for every winning strategy. Formally, we can describe this using the conditional entropy, defined using the probability of $X=x$ conditional on $M=m$, $p_{x|m}=p_{xm}/p_m$,\footnote{We follow physicists' conventions of naming arguments of functions and expressions, so that e.g.\ $p_{m|x}$ is the probability of $M=m$ given $X=x$, not the probability of $X=m$ given $M=x$.}
\begin{align}
H(X|M)=\sum_m p_m H(X|M=m),\qquad \text{for}\qquad H(X|M=m)=-\sum_x p_{x|m} \log p_{x|m}.
\end{align}
The conditional entropy can also be shown to satisfy $H(X|M)=H(XM)-H(M)$, and we can interpret it as the uncertainty of $X$ given knowledge of $M$.
Since a winning strategy only requires $M=X$, it is easy to work out that $H(X|M)=0$ regardless of Bob's choice of $X$, the probability distribution he uses to decide whether to place the coin heads up or down. 
If the memory is faulty, then the stored value will not precisely match the state of the coin. For instance, if there is one chance in eight of a memory error and the coin was placed randomly in the box, then $p_{\rm heads|heads}=\frac 78$, $p_{\rm tails|heads}=\frac 18$, and similarly for the probability conditioned on tails. Working out the conditional entropy, we find $H(X|M)\approx 0.54$, meaning roughly half the information about the coin has been corrupted!

\section{Complementarity in the Information Game}
What changes if Alice and Bob play the game with the quantum version of coins, \emph{qubits}, instead of classical bits? Qubits are any quantum system with two levels, which we denote $\ket{0}$ and $\ket{1}$, for instance the polarization degree of freedom of a single photon (horizontal versus vertical polarization) or the angular momentum of a spin-$\tfrac 12$ particle (angular momentum aligned or antialigned with a fixed spatial axis). Quantum-mechanical complementarity now comes into play and we can alter the game to illustrate the various effects concretely. Before doing so, let us discuss more precisely what is meant by complementarity, adopting the language of the wave-particle duality simplified to a single photon in a Mach-Zehnder interferometer. 

Thinking of light as a particle, we expect to find the photon in one or the other of the two modes. By placing a photodetector in each arm of the interferometer, we can determine where the photon is by looking to see which of the photodetectors is triggered. Let us call this the amplitude measurement. Associating the states $\ket{0}$ and $\ket{1}$ to the two modes, the amplitude measurement corresponds to a projective measurement in this basis. We may also define the amplitude observable by assigning values to the two possible outcomes. The usual choice comes from thinking of a qubit as a spin-$\tfrac 12$ particle and using the angular momentum, and we define the amplitude observable as $Z=\ket{0}\bra{0}-\ket{1}\bra{1}$. That is, a photon in the first mode takes the value $+1$ and in the second $-1$.

If we instead think of light as a wave, we expect there to be a certain phase relationship between the two arms, and in this case the light can interfere either constructively (in phase, $+$) or destructively (out of phase, $-$). To determine which, we allow the two modes to interfere at a beamsplitter and then check in which mode the photon emerges with a photodetector. Let us call this the phase measurement. Like the amplitude measurement, the phase measurement is a projective measurement, but in the basis $\ket{\pm}=\frac{1}{\sqrt{2}}(\ket{0}\pm\ket{1})$. Again we can define a corresponding observable, the phase observable, which for later convenience is defined exactly as the amplitude observable, but in the new basis: $X=\ket{+}\bra{+}-\ket{-}\bra{-}$. In the original basis this works out to be $X=\ket{1}\bra{0}+\ket{0}\bra{1}$.  

Amplitude and phase are complementary properties precisely as in the double slit setup, in the sense that if the photon is in a definite mode, then the phase relationship is completely undefined, and \emph{vice versa}.
This can be immediately seen from the two sets of basis states, as measurement of either eigenstate of amplitude produces a completely random outcome. 
At the level of observables, we can quantify this by an uncertainty relation. 
The most famous of these is the Heisenberg-Robertson relation relating the variances of the observables to the expectation of their commutator~\cite{heisenberg_ueber_1927,robertson_uncertainty_1929},
\begin{align}
\Delta X\Delta Z\geq \tfrac12\left|\langle[X,Z]\rangle_\psi\right|,
\end{align}
where $\langle\cdots\rangle_\psi$ denotes the expectation value evaluated for the quantum state $\ket{\psi}$ of the system. In this case, however, the bound is trivial. Since the operators $X$ and $Z$ anticommute ($XZ+ZX=0$), the righthand side reduces to $|\langle XZ\rangle_\psi|$. Choosing $\ket{\psi}=\ket{0}$ immediately yields zero, and a simple calculation shows this conclusion holds for \emph{any} possible choice of amplitude and phase observables. 

Fortunately, there exist uncertainty relations for which the bound is state-independent. In particular, a version due to Maassen and Uffink is formulated in terms of entropy~\cite{maassen_generalized_1988},\footnote{Entropic uncertainty relations for position and momentum were first conjectured by Everett~\cite{everett_theory_1957,dewitt_many-worlds_1973} and Hirschmann~\cite{hirschman_note_1957} and proven by Becker~\cite{beckner_inequalities_1975}. Generalizations to arbitrary observables were made by Bialynicki-Birula and Mycielski~\cite{biaynicki-birula_uncertainty_1975} and Deutsch~\cite{deutsch_uncertainty_1983}. Kraus~\cite{kraus_complementary_1987} first conjectured the stronger form (\ref{eq:maassen}).} 
\begin{align}
\label{eq:maassen}
H(X)_\psi+H(Z)_\psi\geq \log \frac 1c.
\end{align} 
The quantity $c$ is related to the commutativity of the observables, $c=\max_{j,k}\left|\braket{\psi_j|\phi_k}\right|^2$ for $\ket{\psi_j}$ the eigenvectors of $X$ and $\ket{\phi_k}$ those of $Z$, while the entropies are independently evaluated for the outcomes of the two observables, respectively, given that the system is originally in the quantum state $\psi$. In addition to the state-independent bound, the values of the observable can take play no role in the measure of uncertainty, only the probabilities of the various values. This makes the entropy a somewhat more natural measure than the variance. In the present case the two observables are complementary, meaning $c$ takes on its maximal value, 1 (for observables on a $d$-level quantum system $c_{\rm max}=\log d$). Thus, the amplitude and phase measurements cannot both be certain, and there must be at least one bit of total entropy. 

Alice and Bob can still play the classical information game with qubits, provided Alice only ever makes, say, the amplitude measurement. Bob is free to prepare amplitude eigenstates at random, just as before. In this sense the formalism of quantum information theory encapsulates classical information theory, as anything we wish to express in the latter can be done by working in a fixed basis in the former.\footnote{Here we consider only finite and not continuous alphabets.} 

Now suppose we alter the game so that Alice is free to make either an amplitude or a phase measurement, but she does not tell Bob which. Bob can prepare arbitrary qubit states, but to win the game he would need to be certain of the outcomes of both possible measurements. According to the Maassen-Uffink relation, Equation (\ref{eq:maassen}), this is impossible. There is no quantum state $\ket{\psi}$ Bob can send to Alice such that $H(X)_\psi$ and $H(Z)_\psi$ are both zero, and therefore he cannot win the game with certainty. A simple calculation shows that the best chance Bob has to win the game is to send Alice a state like $\ket{\psi}=\cos\frac\pi 8\ket{0}+\sin\frac\pi 8\ket{1}$, which is ``in between'' the amplitude and a phase eigenstates $\ket{0}$ and $\ket{+}$ in that $|\braket{0|\psi}|=|\braket{+|\psi}|$. Using $\ket{\psi}$, Bob has a roughly 85\% chance ($\frac 1 2+\frac{1}{2\sqrt{2}}$) of correctly predicting that outcomes of either measurement is $+1$.

\section{Entanglement in the Information Game}
What if, after receiving the qubit from Bob, Alice decides on a measurement at random and only asks for a prediction to this particular measurement? Since Bob does not know in advance which measurement Alice will perform, it would seem that this does not help. After all, he is still faced with the impossible task of preparing a state whose amplitude and phase are both predictable. Surprisingly, however, there does exist a winning strategy! The trick is for Bob to store \emph{quantum information} about the system he sends to Alice.
Note that in the game as played in the previous section, Bob really only makes use of classical memory. He may store information such as how he prepared the state for Alice, but this is effectively a recipe for making the state and there is nothing intrinsically quantum about such a recipe.  

To win this version of the game, Bob should create an \emph{entangled} state of two qubits $A$ and $B$,
\begin{align}
\label{eq:epr}
\ket{\Phi}^{AB}\equiv\tfrac{1}{\sqrt{2}}\left(\ket{0}^A\ket{0}^B+\ket{1}^A\ket{1}^B\right)=\tfrac{1}{\sqrt{2}}\left(\ket{+}^A\ket{+}^B+\ket{-}^A\ket{-}^B\right),
\end{align}
and send the $A$ system to Alice. Such entangled states were first by Einstein, Podolsky, and Rosen (EPR)~\cite{einstein_can_1935} and later translated into this 2-level system language by Bohm~\cite{bohm_quantum_1989}. EPR pointed out the paradoxical property that identical measurements on the two systems always produce identical results---the amplitude of $A$ always matches that of $B$ and likewise for phase---even though amplitude and phase for the individual systems cannot both be simultaneously well-defined.\footnote{The EPR-Bohm states were actually states of two spin-1/2 systems with total angular momentum zero, so that identical measurements are always anticorrelated, but the point is the same.}
In a sense, entangled states display correlations even though there is nothing there to correlate!

However paradoxical, with entanglement Bob can always win the modified game. When Alice asks him to predict a particular measurement, he can simply consult his quantum memory, system $B$, by performing the same measurement Alice will make. Since the results are correlated, $B$ in some sense contains one bit of classical information about both the amplitude and phase of system $A$. However, only one of these can ever be accessed because Bob cannot perform both measurements simultaneously; being able to do so would run afoul of the uncertainty principle. This peculiar combination of classical information about complementary physical properties is the essence of \emph{quantum} information. Demonstrating this more concretely will be the topic of Chapter~\ref{chap:char}. 

At first glance it would seem that this behavior violates the entropic uncertainty relation Equation~(\ref{eq:maassen}). Now, however, Bob makes use of system $B$, so we should consider the entropies of the measurements conditioned on this fact. Thus Equation~(\ref{eq:maassen}) does not apply. Just such a conditional version was conjectured and proven for the particular observables under consideration here in~\citeme{renes_conjectured_2009} and extended to general observables in~\citeme{berta_uncertainty_2010}.\footnote{Uncertainty principles involving conditional entropy were first investigated by Hall~\cite{hall_information_1995} and extended to the case of separate conditional systems by Cerf \emph{et al.}~\cite{cerf_security_2002}. Christandl and Winter~\cite{christandl_uncertainty_2005} gave a version for quantum channels which was the inspiration for the work in~\cite{renes_conjectured_2009}. A much simpler proof of Equation~(\ref{eq:berta}) using the relative entropy was discovered by Coles\etalsp\cite{coles_information_2010}.} It states 
\begin{align}
\label{eq:berta}
H(X^A|B)_\psi+H(Z^A|B)_\psi\geq \log \frac 1c+H(A|B)_\psi,
\end{align}
where now we make use of the quantum conditional entropy, defined using the von Neumann entropy (the Shannon entropy of the eigenvalues of the density matrix) as $H(A|B)_\psi=H(AB)_\psi-H(B)_\psi$. The entropies $H(X^A|B)_\psi$ and $H(Z^A|B)_\psi$ refer to quantum conditional entropies evaluated for the state after the respective observable of system $A$ has been measured. The interpretation of a classical entropy conditioned on a quantum system is not as clear as entropy conditioned on a classical system, but Holevo has shown that it provides a lower bound on the classical conditional entropy of the stated measurement on system $A$ given the result of the optimal measurement on system $B$~\cite{holevo_bounds_1973,holevo_statistical_1973-1}.\footnote{This result was first proven by Forney~\cite{forney_concepts_1963}, who did not make the connection to the conditional entropy.} 

Although the additional term on the righthand side might appear to make the bound tighter, the quantum conditional entropy of $A$ given $B$ can in fact be \emph{negative}.
For example, entangled states such as $\ket{\Phi}$ have $H(A|B)=-1$ since the $AB$ state is pure (whence $H(AB)=0$) but the state of $B$ alone is completely random (whence $H(B)=1$). This reflects another strange nature of the EPR state in that our uncertainty of the whole system $AB$ appears to be \emph{less} than that of one of its parts. In the present context, $H(A|B)=-1$ implies that the righthand side of (\ref{eq:berta}) is zero. Thus, the bound is trivial, and conditioned on the quantum information $B$, the entropy of $X$ and $Z$ \emph{can} both be zero. 

An alternate and fully equivalent form of Equation~(\ref{eq:berta}) ensures that, even if Bob makes use of quantum information in the original version of the game where he has to predict both outcomes, no winning strategy can exist. It now involves three systems: the system to be measured, $A$, and two memory systems $B$ and $C$,
\begin{align}
\label{eq:jcbjmr}
H(X^A|B)_\psi+H(Z^A|C)_\psi\geq \log\frac{1}{c}.
\end{align}
In order to make a prediction of both amplitude and phase, Bob would need two physical systems in which to store this information. Even if he uses systems $B$ and $C$ as quantum memory, Equation~(\ref{eq:jcbjmr}) ensures that amplitude and phase are still not simultaneously predictable. Put differently, although Bob can store classical information about both properties in the EPR state, there is no way to separate the amplitude and phase information without losing some of each in the process.

Note that we were able to define entropy conditioned on a quantum system via the alternate form of the conditional entropy expression, $H(A|B)=H(AB)-H(B)$. In retrospect, it is extremely fortunate that this form exists, because even the very notion of conditioning on quantum information is itself suspect. After all, the very nature of quantum systems is that their physical properties are not well-defined, so it is unclear what one should condition on. For instance, one might also like to define the variance of an observable on system $A$ conditioned on the state of a quantum memory, system $B$. But how can the presence of $B$ be incorporated into a variance calculation? We could stipulate that $B$ is to be measured, calculate the variance of $A$ for each outcome, and take the average, but this leads to unwieldy expressions. In the case of entropy, the formal structure rescues us and allows us to meaningfully speak of uncertainty conditioned on quantum information.

\chapter{Illustrations and Motivations}
\label{chap:ill}

That complementary observables play an important role in quantum information processing is not original to this thesis, though we shall see new, more concrete characterizations of quantum information in terms of classical information pertaining to complementary observables and uses for these characterizations in subsequent chapters. 
In this chapter we recount several protocols in quantum information theory that anticipated and motivated the work presented herein.  Among these are teleportation, where a qubit is sort of transmitted by two classical bits, and superdense coding, where conversely a qubit carries two bits of classical information. An even more concrete prior manifestation comes from quantum error-correction, which is crucial to the possibility of ever constructing a working quantum computer, and its use in protocols for entanglement distillation and quantum key distribution (QKD). This we discuss in more detail, as the structure of error-correcting codes will be useful in later chapters. But first we turn to teleportation and superdense coding. 

\section{Teleportation and Superdense Coding}
Teleportation and superdense coding are two simple quantum information processing protocols which rather dramatically demonstrate how different quantum information is from classical information. They also indicate a connection between quantum information and complementary classical information. Both involve two parties, a sender Alice and a receiver Bob, who share an EPR pair as given in Equation~(\ref{eq:epr}). In superdense coding, Alice would like to transmit classical information to Bob, but using the quantum channel. One  method is for both parties to fix a basis, Alice only sending amplitude eigenstates $\ket{0}$ or $\ket{1}$ and Bob only measuring what he receives in the same basis. This allows them to send one bit of classical information per qubit. 

However, they can do better by making use of their shared entanglement, and Alice can send Bob two classical bits per qubit~\cite{bennett_communication_1992}. The trick is to use the Bell basis,\footnote{So-named as they figure prominently in the study of whether quantum mechanics permits description as a local hidden variable theory by John S.\ Bell~\cite{nielsen_quantum_2000}.} a basis of two maximally-entangled qubit states, defined as follows, 
\begin{align}
\label{eq:bell}
\ket{\beta_{jk}}^{AB}\equiv (X^jZ^k\otimes\id)\ket{\Phi}^{AB},
\end{align}
using the amplitude and phase operators as defined in the previous chapter. For completeness, we again write them here, in the basis $\{\ket{0},\ket{1}\}$,
\begin{align}
X=\begin{pmatrix} 0 & 1 \\ 1 & 0\end{pmatrix}\qquad\text{and}\qquad Z=\begin{pmatrix}1 & 0\\ 0& -1\end{pmatrix}.
\end{align}
To transmit the two bits $j$ and $k$, first Alice applies $X^jZ^k$ to her half of the entangled state, $A$, and then sends it to Bob over the quantum channel. Since the Bell states form a basis, Bob can measure the joint system $AB$ in this basis to determine $j$ and $k$. In this way, one qubit of quantum information can be made to carry two bits of classical information.  

The classical information can heuristically be regarded as one bit of amplitude information and one bit of phase information in the following manner. In the original scheme to transmit one bit per qubit using only the amplitude basis, Alice's actions can be described as \emph{modulating} an initial state $\ket{0}$ by the operator $X^j$, producing $\ket{1}$ if $j=1$ and leaving the state as $\ket{0}$ otherwise. The same modulation scheme works in the phase basis using the operator $Z^k$, starting with $\ket{+}$. In superdense coding, Alice apparently performs \emph{both} and amplitude and a phase modulation, encoding two bits at once. Due to the entanglement with Bob's system, these two actions can coexist without interfering with each other, allowing two bits to be transmitted. 

Teleportation is sort of the inverse of superdense coding; now, preshared entanglement enables Alice to send one qubit to Bob by transmitting two classical bits~\cite{bennett_teleporting_1993}. Again the trick is to use the Bell basis. If Alice measures her half of the entangled state and the qubit to be sent in the Bell basis, she need only forward Bob the measurement results and he will be able to reconstruct the input state. 

Formally, we let $C$ be the qubit input, in an arbitrary state $\ket{\psi}^C$. 
It is not difficult to verify that ${^{AC}}{\bra{\beta_{jk}}}\left(\ket{\Phi}^{AB}\ket{\psi}^C\right)=\frac12 (Z^kX^j)^B\ket{\psi}^B$. This means that after Alice measures her two systems in the Bell basis, each outcome occurring with probability $\frac 14$, Bob ends up with the state $Z^jX^k\ket{\psi}$. Thus, Alice merely has to send Bob the two bits of information $j$ and $k$, and he can apply $X^jZ^k$ to recover the original state $\ket{\psi}$ in system $B$. In this way, the qubit is transmitted by two classical bits, with the help of preshared entanglement. 

We can heuristically think of the two classical bits as being the amplitude and phase of the input state $\ket{\psi}$ for the following reason. One way to perform a Bell state measurement is to first perform the controlled-\textsc{not} (\textsc{cnot}) operation and then measure each qubit separately in the appropriate basis. The \textsc{cnot} gate acts on two qubits, applying $X$ to the second qubit (the target) if the first (the control) is $\ket{1}$ and doing nothing to the target otherwise. It can be thought of as coherently copying the amplitude basis of the control qubit to the target, in that a superposition state $\alpha\ket{0}+\beta\ket{1}$ of the control qubit and a ``blank'' target state $\ket{0}$ become $\alpha\ket{00}+\beta\ket{11}$. To complete the Bell state measurement after applying \textsc{cnot}, one measures the amplitude of the target qubit and the phase of the control. Therefore, in the teleportation protocol, we can choose the input qubit to be the control and Alice's half of the entangled state as the target, and it then appears as if the amplitude information is first copied to the second qubit and read out, while the phase is read out from the first qubit, the system itself. Of course, this is not precisely what happens, or else Alice would obtain both amplitude and phase information of $\ket{\psi}$, in violation of Equation~(\ref{eq:jcbjmr}). Nonetheless, teleportation indicates the important role played by amplitude and phase information.

\section{Quantum Error-Correction}
\label{sec:qec}


In the uncertainty game of the previous chapter, we assumed that the quantum memory used by Bob was noise-free. Clearly this is an unrealistic assumption, and although not particularly relevant for a \emph{gedankenexperiment}, it nevertheless raises the question of what can be done to combat noise is real quantum information processing protocols. The answer, in the quantum case as in the classical case, is to use error-correcting codes. 
The fact that quantum error correction exists at all is of tremendous importance both practically and conceptually. On the one hand it shows that construction of quantum computers is not in principle a hopeless task, and on the other that quantum information itself is essentially digital (discrete-valued) in nature, despite its outward analog (continuous-valued) appearance. Even more, the way in which the first quantum error-correcting codes were constructed is related to the complementarity of quantum information: Arbitrary quantum errors are digitized into amplitude and phase errors, each of which is then corrected by essentially classical means. Before delving into the details of how quantum error-correction works, which illustrates the point more clearly and will be of use in later chapters, we give a brief overview of the issue of analog versus digital computation and the important role played by error-correction for both classical and quantum computers. 

Whether classical or quantum, both analog and digital devices require error-correction to control the effects of noise inescapably present in an actual device. A simple classical error-correction scheme is simply to repeat the calculation three times and take the majority of the results. However, the error-correction procedure itself is presumably not perfect and can only be performed to some finite accuracy in practice. Nonetheless, following an analysis by von Neumann~\cite{von_neumann_probabilistic_1956}, it is possible that the rate at which errors are decreased by the procedure is greater than that at which they are caused.

For digital computers the finite accuracy of the procedure presents no additional difficulties in principle because the device anyway only requires finite accuracy; in a discrete encoding we choose certain continuous parameter ranges of the underlying physical degree of freedom to correspond to discrete logical values. Thus the nature of the encoding accords very well with the finite accuracy of available operations, and errors in the latter transform into errors in the former.  
Analog computers, however, require error-correction to arbitrary precision, so the buildup of errors due to finite-accuracy of operations is ultimately unavoidable. That reliable digital computers can be constructed from imperfect components was shown rigorously by G\'acs~\cite{gacs_reliable_1986}, though in practice current devices usually require error-correction only in the storage of information, not its manipulation, due to the intrinsically low error-rates of semiconductor-based integrated circuits. 

Since the quantum state of the quantum computer is determined by the continuous probability amplitudes appearing in the wavefunction, many of the same difficulties were thought to apply to quantum computers, an issue pointed out by Peres~\cite{peres_reversible_1985} and stressed by Landauer~\cite{landauer_computation_1986,landauer_dissipation_1988,landauer_physical_1996}. Noise-induced modifications to these amplitudes leads to errors in the computation, just as in the analog computer, so it would seem that any advantage promised by quantum computation in principle cannot be achieved in practice. 
Worse still, even the ability to perform error-correction seems suspect in the quantum setting, because the information cannot simply be read out to check for errors, as in the von Neumann repetition scheme, without introducing disturbance~\cite{landauer_physical_1996}. Nonetheless, there was reason for optimism: Zurek observed that owing to the different phase-space structures involved, the kind of exponential blow-up of errors that might be expected for a classical continuous computer would not plague a quantum computer with a discrete spectrum~\cite{zurek_reversibility_1984}.

Happily, the construction of quantum error-correcting codes by Shor~\cite{shor_scheme_1995} and Steane~\cite{steane_error_1996} demonstrates convincingly that quantum information is not analog, but digital.\footnote{Some would still dispute this. See, e.g.\ Laughlin~\cite{laughlin_different_2006}.} Soon thereafter it was established that, just as with classical digital computers, reliable quantum computers could in principle be constructed using imperfect components, a fact known as the threshold theorem~\cite{aharonov_fault-tolerant_1997,aharonov_fault-tolerant_2008,kitaev_quantum_1997,knill_resilient_1998-1,knill_resilient_1998}. Unlike the situation for classical electronic computers, no medium has yet been discovered or engineered which offers intrinsically low quantum noise rates, though much effort is devoted to this question and many major experimental achievements have been made.  
The crux of quantum error-correction is that although continuous errors in the state of the computer are indeed possible, they can be digitized without damaging the encoded quantum information. Instead of accessing the quantum information directly, as one would try in a direct analogue of the repetition scheme, the measurements needed in error-correction are designed only to provide information about the error, not the encoded information. In this way the construction very subtly evades the two objections described above.

\subsection{The Complementarity of Quantum Error-Correcting Codes}
\label{subsec:compqecc}

Somewhat amazingly, quantum errors of any type can be corrected if discrete errors of two complementary types, amplitude and phase, can be corrected~\cite{knill_theory_1997,bennett_mixed-state_1996}. These two errors result from the action of the already-defined $X$ and $Z$ operators, respectively, acting exactly as an unwanted modulation of the quantum state. Often these errors are referred to as bit flips and phase flips, for the following reason. One commonly fixes a basis and calls it the amplitude basis, and then for an arbitrary qubit state $\ket{\psi}=\alpha\ket{0}+\beta\ket{1}$ an amplitude error resulting from an unwanted $X$ operator just flips the states $\ket{0}$ and $\ket{1}$, hence the name bit flip. Similarly, phase flips interchange the states $\ket{+}$ and $\ket{-}$, or equivalently, flips the phase of $\ket{1}$, taking $(\alpha,\beta)$ to $(\alpha,-\beta)$. 

Either type of error by itself could be corrected in exactly the way a classical error would be corrected, through repetition. To correct a single bit flip error classically, we can \emph{encode} it into three bits as follows,
\begin{align}
0\rightarrow \overline{0}=000\qquad 1\rightarrow\overline{1}=111.
\end{align}
These two bitstrings are called \emph{codewords}, and the overline denotes a logical value of the encoded bit, as opposed to the values of the individual physical bits. 
Then, if one error occurs, we can correct it by examining each string and flipping the one bit which is different from the other two. Equivalently, the error may be diagnosed by computing the two parities, generally called \emph{syndromes}, $s_1=b_1\oplus b_3$ and $s_2=b_2\oplus b_3$, where $b_1$, $b_2$, and $b_3$ are the three bit values. The syndromes associated to each error position are shown in Table~\ref{table:3bitrep}. Note that the bit is encoded in the value of $\overline{b}=b_1\oplus b_2\oplus b_3$.  

\begin{table}[h]
\begin{center}
\begin{tabular}{cccc}
Bitstring pair $(\overline{0},\overline{1})$ & Error Position & Syndrome $(s_1,s_2)$\\
$(000,111)$ & $\emptyset$ & $(0,0)$\\
$(100,011)$ & 1 & $(1,0)$\\
$(010,101)$ & 2 & $(0,1)$\\
$(001,110)$ & 3 & $(1,1)$
\end{tabular}
\caption{\label{table:3bitrep} The three-bit repetition code. The first column gives the bitstrings corresponding to the encoded logical zero $\overline{0}$ and logical one $\overline{1}$ after a bitflip error whose position is given in the second column. The third column lists the syndrome information which allows the error position to be diagnosed.}
\end{center}
\end{table}

Seen from a different perspective, the reason this works is that the eight possible three-bit strings are grouped into four pairs, as in Table~\ref{table:3bitrep}. One pair is given by the codewords themselves, and the other pairs are the images of the codewords under the three single-bit errors. In each pair one string corresponds to $\overline{0}$ and the other to $\overline{1}$ as defined by this mapping. The syndromes reveal precisely which pair is present, but importantly they do not reveal anything about the logical bit value. Error-correction corresponds to mapping the noisy pair of strings back to the original pair.

To correct qubit bit flip errors we may simply use the same repetition code in the computational basis. Since the syndrome and correction procedure for a given error are independent of the encoded information, superpositions are also maintained by the error-correcting code. Thus, the state $\ket{\psi}=\alpha\ket{0}+\beta\ket{1}$ is encoded as $\ket{\overline{\psi}}=\alpha\ket{000}+\beta\ket{111}$, a process which can be implemented as a unitary transformation on the input and two auxiliary systems, each in some given state we can take to initially be prepared in the state $\ket{0}$. The necessary syndrome information can be generated by measuring the two \emph{stabilizer} operators $Z\id Z=Z\otimes \id\otimes Z$ and $\id ZZ$, which we can write as $Z_1Z_3$ and $Z_2Z_3$. Each of these has the same action on the two logical states in each subspace, returning the values $(-1)^{s_1}$ and $(-1)^{s_2}$, respectively. 

The name stabilizer reflects the fact that the code subspace is stabilized by the two operators, as it is the simultaneous $+1$ eigensubspace of both operators. The encoded subspace supports a single qubit, and so it must be possible to represent its amplitude and phase operators. One possibility is given by $\overline{Z}=Z_1Z_2Z_3$ and $\overline{X}=X_1X_2X_3$. These each commute with the stabilizers, but anticommute with each other as intended. Note that $\overline{Z}$ gives the encoded bit, just as in the classical case. 

We can also think of the stabilizers and encoded amplitude operator as defining a new complete set of commuting observables for the set of physical qubits. Such a set fixes a basis in the state space of the three qubits, and each of the operators is the amplitude operator for a corresponding ``virtual'' qubit. Labeling the virtual qubit operators with primes, we can write $Z_{1}'=Z_1Z_2$, $Z_2'=Z_2Z_3$, and $Z_3'=\overline{Z}=Z_1Z_2Z_3$. Conjugate to the new amplitude observables are phase observables $X_1'=X_2X_3$, $X_2'=X_1X_3$, and $X_3'=\overline{X}=X_1X_2X_3$, which are found by ensuring that they anticommute with the amplitude operators of the same qubit but commute with all other operators. The entire collection is shown in Table~\ref{tab:repcodestab}.
The code subspace is then defined by the first two virtual qubits being in the $+1$ amplitude state. Bit flip errors change the amplitude of the encoded qubit and at least one of the virtual qubits, and the stabilizer measurement determining the location of the error translates into an amplitude measurement of the first two virtual qubits.  

\begin{table}[h]
\begin{center}
\begin{tabular}{ccc}
Virtual qubit & Amplitude & Phase\\
1 & $ZZ\id$& $\id XX$\\
2 & $\id ZZ$& $XX\id$ \\
3 & $ZZZ$ & $XXX$
\end{tabular}
\caption{\label{tab:repcodestab} Virtual qubits associated with the three-qubit amplitude repetition code. Note that amplitude and phase anticommute for each qubit, but commute for different qubits. 
}
\end{center}
\end{table}

Discretization is automatically provided by the measurement of the stabilizer operators, which is anyway necessary for error-correction. Consider an error operator of the form $E=e_0 I+e_1 X_1$, with $e_0,e_1\in\mathbbm{C}$, which is a sort of combination bit flip error and no error on the first qubit. It produces a superposition between two code subspaces, 
\begin{align}
\ket{\overline{\psi}{'}}=E\ket{\overline{\psi}}=e_0\ket{\overline{\psi}}+e_1X_1\ket{\overline{\psi}}=e_0\left(\alpha\ket{000}+\beta\ket{111}\right)+e_1\left(\alpha\ket{100}+\beta\ket{011}\right).
\end{align}
Measurement of the stabilizer operators destroys this superposition, forcing the system to the state of either one error or no error, but leaves the logical qubit superposition intact. Here the measurement has two possible syndrome outcomes, either $(0,0)$ or $(1,0)$, with probabilities $|e_0|^2/(|e_0|^2+|e_1|^2)$ and $|e_1|^2/(|e_0|^2+|e_1|^2)$, respectively. Conditioned on these outcomes, the state becomes $\ket{\overline{\psi}}$ or $X_1\ket{\overline{\psi}}$, respectively, and can therefore be corrected using the syndrome information.

\subsection{Correcting Both Kinds of Errors}
\label{subsec:shor9}

Since phase flips are just bit flips in the basis $\ket{\pm}$, the above analysis immediately applies to this case upon changing $X\leftrightarrow Z$ and working in the new basis. The insight of Shor and Steane was to realize that a single error of either type can be corrected by appropriately combining these procedures.
Shor's scheme is conceptually somewhat simpler, and is based on concatenating the two error-correcting codes. That is, we take the codewords of the phase flip repetition code and replace each of the three qubits with  qubits appropriately encoded in the bit flip repetition code. This produces codewords of nine qubits, as follows (here ignoring normalization),
\begin{align}
\ket{+}\longrightarrow&\ket{\overline{+}}=\ket{+++}=(\ket{0}+\ket{1})(\ket{0}+\ket{1})(\ket{0}+\ket{1})\\
&\ket{\overline{+}}\longrightarrow \ket{\widetilde{+}}=(\ket{\overline{0}}+\ket{\overline{1}})(\ket{\overline{0}}+\ket{\overline{1}})(\ket{\overline{0}}+\ket{\overline{1}}),\\
\ket{-}\longrightarrow&\ket{\overline{-}}=\ket{---}=(\ket{0}-\ket{1})(\ket{0}-\ket{1})(\ket{0}-\ket{1})\\
&\ket{\overline{-}}\longrightarrow \ket{\widetilde{-}}=(\ket{\overline{0}}-\ket{\overline{1}})(\ket{\overline{0}}-\ket{\overline{1}})(\ket{\overline{0}}-\ket{\overline{1}}).
\end{align}

The repetition in the amplitude basis in the second step protects the encoded qubit from bit flip errors, 
since a single bit flip can always be detected and corrected by applying the 3-qubit repetition procedure to each block of three qubits. This corresponds to measuring the $Z$-parity observables $Z_1Z_3$, $Z_2Z_3$, $Z_4Z_6$, $Z_5Z_6$, $Z_7Z_9$, and $Z_8Z_9$. 
Phase flips are slightly more involved, but consider what happens when a single phase flip error plagues, say, the fourth qubit. This is the first qubit of the second block, so we can zoom in on this block to determine the effect on the encoded states. Applying the error operator $Z_1$ to the encoded states we find $Z_1\ket{\overline{0}}=Z_1\ket{000}=\ket{000}=\ket{\overline{0}}$, while $Z_1\ket{\overline{1}}=Z_1\ket{111}=-\ket{111}=-\ket{\overline{1}}$. Thus, the error causes the action
\begin{align}
\ket{\widetilde{+}}&\rightarrow (\ket{\overline{0}}+\ket{\overline{1}})(\ket{\overline{0}}-\ket{\overline{1}})(\ket{\overline{0}}+\ket{\overline{1}})\\
 \ket{\widetilde{-}}&\rightarrow(\ket{\overline{0}}-\ket{\overline{1}})(\ket{\overline{0}}+\ket{\overline{1}})(\ket{\overline{0}}-\ket{\overline{1}}),
\end{align}
which is precisely a phase flip at the ``inner'' level. We could detect and correct this at the inner level by measuring the $X$-parities $X_1X_3$ and $X_2X_3$. Translating to the outer level of actual qubits, we replace each of the constituent $X$ operators on the inner level by its encoded $\overline{X}$ operator on the outer level and instead measure $X_1X_2X_3X_7X_8X_9$ and $X_4X_5X_6X_7X_8X_9$. The outcomes for the damaged states are $+1$ and $-1$ respectively, for both encoded states, implying that to correct the error we merely need apply $Z_4$.\footnote{$Z_5$ or $Z_6$ would also work just as well. This flexibility is actually a subtle and important feature of quantum error-correcting codes we shall return to in Section~\ref{sec:psqkd}.}

The six amplitude parities and two phase parities commute pairwise and stabilize the code subspace. As with the repetition code, the error analysis is made simpler by thinking in terms of virtual qubits, in this case nine, as shown in Table~\ref{tab:shorcodestab}. Observe that the concatenated structure is reflected in the operators: three copies of the repetition code in virtual qubits one through six, followed by the same repetition code on the three blocks. The code subspace is fixed by requiring virtual qubits one through six to be in the $+1$ amplitude eigenstate and virtual qubits seven and eight in the $+1$ phase eigenstate, but this structure makes it clear that we could have defined the code the other way around.

Using this framework it is easy to see that the Shor code also enables detection and correction of joint bit and phase errors. A joint bit and phase flip of the fourth qubit, for instance, would reveal itself by the fourth virtual qubit having the wrong amplitude and the seventh having the wrong phase, corresponding to $-1$ eigenvalues of the stabilizers $Z_4Z_6$ and $X_4X_5X_6X_7X_8X_9$. From the structure of the virtual amplitude and phase operators it is clear that the code can actually detect and correct one bit and one phase error, irrespective of their locations.

\begin{table}[h]
\begin{center}
\begin{tabular}{ccc}
Virtual qubit \# & Amplitude & Phase\\
1 & $ZZ\idb \idb\idb\idb \idb\idb\idb$ & $\idx XX \idx\idx\idx \idx\idx\idx$\\
2 & $\idb ZZ\idb\idb\idb\idb\idb\idb$& $XX\idx \idx\idx\idx \idx\idx\idx$\\
3 & $\idb\idb\idb ZZ\idb \idb\idb\idb$& $\idx\idx\idx \idx XX \idx\idx\idx$\\
4 & $\idb\idb\idb \idb ZZ \idb\idb\idb$& $\idx\idx\idx XX \idx \idx\idx\idx$\\
5 & $\idb\idb\idb \idb\idb\idb  ZZ \idb$& $\idx\idx\idx\idx\idx\idx \idx XX$\\
6 & $\idb\idb\idb \idb\idb\idb \idb ZZ $& $\idx\idx\idx\idx\idx\idx XX\idx$\\
7 & $ZZZZZZ\idb\idb\idb$& $\idx\idx\idx XXXXXX$\\
8 & $\idb\idb\idb ZZZZZZ$& $XXXXXX\idx\idx\idx$\\
9 &$ZZZZZZZZZ$ & $XXXXXXXXX$\\
\end{tabular}
\caption{\label{tab:shorcodestab} Virtual qubits associated with the nine-qubit Shor code. Note that amplitude and phase anticommute for each qubit, but commute for different qubits. 
}
\end{center}
\end{table}

Again error discretization is provided by the stabilizer measurement, and fortunately, being able to correct just these two types of error is sufficient to correct any conceivable single-site error. Just as with the repetition code, we can consider the effect of arbitrary errors which are linear combinations of all the correctable errors. Since the Shor code can correct any single flip of bit and/or phase, errors of the form $E=e_{00} I+e_{10}X_1+e_{01}Z_1+e_{11}X_1Z_1$ with $e_{jk}\in\mathbbm{C}$ can also be corrected. But, as can be readily verified, any operator can be expressed in this way as a complex combination of these four operators, meaning arbitrary single-site errors can be digitized to amplitude and/or phase errors and corrected. 
Despite initial appearances to the contrary, quantum information is therefore in a critical sense digital.


\section{Entanglement Distillation}
\label{sec:ed1}
Quantum error-correction quickly found use in constructing protocols for \emph{distillation} of entanglement, as well as in proving the cryptographic security of quantum key distribution protocols. We give a brief treatment of these uses here, as they will be generalized in later chapters.

Distilling entanglement refers to transforming imperfect EPR states into approximately perfect ones. For instance, if Alice sends halves of maximally-entangled states through a noisy quantum channel to Bob, then the states which emerge will no longer be maximally-entangled. But it may be possible to repair some fraction of the states by actions undertaken on Alice's and Bob's systems alone, plus classical communication between them to coordinate their actions. To see how this is done, suppose that Alice and Bob share many copies of the state
\begin{align}
\label{eq:belldiagonal}
\psi^{AB}=\sum_{jk}p_{jk}\ket{\beta_{jk}}\bra{\beta_{jk}}^{AB},
\end{align}
with $p_{jk}\geq 0$ and $\sum_{jk}p_{jk}=1$, which is just a probabilistic mixture of the four Bell states. This state is produced, for instance, by sending the $B$ half of the state $\ket{\Phi}=\ket{\beta_{00}}$ through a channel which applies the operator $X^jZ^k$ with probability $p_{jk}$. In principle, Bob can repair the actual state $\psi^{AB}$ to the desired state $\Phi^{AB}$ by determining which of these operators was applied and subsequently undoing it. Thus, the task is reduced to determining the actual sequence of errors, at least for states of this form.

This sounds like a job for a quantum error-correcting code, even though here Alice is not first encoding the qubits she sends to Bob. Nevertheless, Alice and Bob can determine the error pattern by each measuring the stabilizer operators of an error-correcting code. It is simple to show that, just as in the information game, if Alice and Bob make the same stabilizer measurements on collections of EPR states, then they should always obtain the same outcomes. To the extent that they obtain different outcomes, this indicates an error. For instance, suppose that Alice and Bob divide their systems into groups of three and use the simple bit-flip repetition code described above. On each group of three, both Alice and Bob measure the stabilizers $Z_1Z_3$ and $Z_2Z_3$. Perhaps the simplest way to work out what outcomes will occur is to note the following relationship,
\begin{align}
\id\otimes E\ket{\Phi}=E^T\otimes\id\ket{\Phi},
\end{align}
valid for any operator $E$, where $E^T$ is the transpose of the operator when expressed in the amplitude basis $\{\ket{0},\ket{1}\}$. Now we can calculate the effect of the product of Alice's and Bob's stabilizers on the ideal state. Since $Z^T=Z$ and $Z^2=\mathbbm{1}$,
\begin{align}
Z^{A_1}Z^{A_3}Z^{B_1}Z^{B_3}\left(\ket{\Phi}^{A_1B_1}\ket{\Phi}^{A_2B_2}\ket{\Phi}^{A_3B_3}\right)=\ket{\Phi}^{A_1B_1}\ket{\Phi}^{A_2B_2}\ket{\Phi}^{A_3B_3},
\end{align}
which implies that Alice and Bob must indeed obtain identical outcomes for their stabilizer measurements since their product must be $+1$. The same clearly holds for $Z_2Z_3$. If there is one $X$ error in the state, say in the first position, then we find using the same method
\begin{align}
Z^{A_1}Z^{A_3}Z^{B_1}Z^{B_2}\left(X^{B_1}\ket{\Phi}^{A_1B_1}\ket{\Phi}^{A_2B_2}\ket{\Phi}^{A_3B_3}\right)&=Z^{A_1}X^{A_1}Z^{A_1}\left(\ket{\Phi}^{A_1B_1}\ket{\Phi}^{A_2B_2}\ket{\Phi}^{A_3B_3}\right)\\
&=-X^{B_1}\ket{\Phi}^{A_1B_1}\ket{\Phi}^{A_2B_2}\ket{\Phi}^{A_3B_3}.
\end{align}
Now the state is a $-1$ eigenstate of the product of stabilizers, meaning the product of syndromes is $-1$, and hence that Alice and Bob obtain different outcomes for these stabilizer measurements. A single $X$ error on the first qubit will of course not affect the $Z_2Z_3$ measurements. But together the two stabilizer measurements suffice to locate a single $X$ error in the three pairs, exactly as in the error-correction scenario.

The story is essentially the same for any quantum error-correcting code, so we may create a protocol for entanglement distillation as follows, following Bennett\etalsp\cite{bennett_mixed-state_1996}. First, Alice and Bob use a small fraction of their pairs in order to determine the number of each type of error $X$, $Z$, and $XZ$, simply by both measuring in the appropriate basis and recording how often they obtained the same outcome. The bases are just the amplitude basis, the phase basis, and the basis consisting of the eigenstates $\frac{1}{\sqrt{2}}\left(\ket{0}\pm i\ket{1}\right)$ of $XZ$, respectively. 
Next, given the expected number of errors, they choose an appropriate error-correcting code, but if no suitable codes exist they must abort the procedure. If a suitable code does exist, they proceed by measuring the stabilizers to determine, with high probability, the actual pattern of errors, which can then be corrected by local operations on Bob's systems. 

This does not quite leave them with the desired states $\ket{\Phi}$, however, since they have made the stabilizer measurements. Instead, the $\ket{\Phi}$ reside in the encoded subspaces specified by the error-correcting code, their number corresponding to the number of encoded qubits. To recover these states, they each apply the decoding operation (the inverse of the encoding operation) to their systems.
The above protocol is designed to work for states of the form given in Equation~(\ref{eq:belldiagonal}), but actually applies to any input state since the stabilizers used in the protocol will automatically digitize arbitrary errors to amplitude and phase errors. 


\section{Quantum Key Distribution}
\label{sec:qkd}
Quantum key distribution (QKD) provides a means for the two separated parties Alice and Bob to communicate in private using only public communication channels. The security of the scheme is based only on the laws of physics and not the perceived computational difficulty of some task, like factoring large integers, as commonly used in classical schemes today. Needless to say, the problem of private communication is ancient, but it was first put on a firm mathematical footing by Shannon~\cite{shannon_communication_1949}. There the task is broken into two parts, establishing a \emph{secret key} between the two parties, a random string of bits shared by both parties, and then using it to \emph{encrypt} and \emph{decrypt} the actual messages. One can imagine Alice and Bob creating a secret key together at some point in the past when they could do so secretly, but if they are already separated and can only communicate publicly, the situation seems hopeless. They could communicate privately if they had a key, but they need to communicate privately to create the key. 

Quantum information offers a way out of this dilemma 
in the form of entanglement. Returning to the uncertainty game, recall that Bob can, on demand, predict either the amplitude or phase measurement on Alice's system when they share an EPR pair. Moreover, the uncertainty principle Equation~(\ref{eq:jcbjmr}) implies that any would-be eavesdropper Eve could not predict either measurement using her system $C$ any better than by just blindly guessing, a property of entanglement known as \emph{monogamy}. By measuring each of their systems in identical bases, Alice and Bob can therefore generate one bit of a secret key from each entangled pair. 

They can attempt to create such pairs by using a public quantum channel in the manner described in the previous subsection: Alice prepares EPR pairs and sends one system of each to Bob. If the channel is noisy, perhaps due to Eve's interference, Alice and Bob can simply first run an entanglement distillation protocol to extract the required high-quality EPR pairs. Even though this requires them to exchange classical syndrome information over a public channel, it does not help any would-be eavesdropper as the measurements on the EPR pairs are completely independent of this information, a fact again insured by the uncertainty principle. The usefulness of entanglement distillation in this context was first treated by Deutsch et al.~\cite{deutsch_quantum_1996} and the security of this scheme was first rigorously proven by Lo and Chau~\cite{lo_unconditional_1999}. 

The protocol will require a large quantum memory in which to store the various systems, as well as the ability to perform all the necessary stabilizer measurements. We did not worry about the practicalities of doing so in the previous section, but luckily for QKD all of the required operations can be reduced to just measuring in either the amplitude or phase basis, and subsequent processing of the resulting classical data, as shown by Shor and Preskill~\cite{shor_simple_2000}. The reason this works is that ultimately we want to distill EPR states but then immediately measure them in some basis to generate the key, and this gives us some flexibility in how we describe the entire process. By picking the right kind of error-correction code this flexibility allows us to get rid of essentially all (difficult) operations on quantum systems apart from measuring them individually and replace them with (easy) operations on classical data. 

The necessary codes are called Calderbank-Shor-Steane (CSS) codes and include the original codes found by Shor and Steane as mentioned in Section~\ref{sec:qec}. Their defining property, as first described by Calderbank and Shor~\cite{calderbank_good_1996} and Steane~\cite{steane_multiple-particle_1996}, is that the stabilizers of the code can be broken into two groups, those composed of products of $X$ operators and those composed of products of $Z$ operators. Similarly, the logical amplitude operators only consist of $Z$-type  operators, while the logical phase operators only consist of $X$-type operators. The more general formalism of stabilizer codes constructed by Gottesman~\cite{gottesman_stabilizer_1997} also includes codes whose stabilizers and logical operators are of mixed type, but importantly, these cannot be used for the present purposes. 

Consider the QKD scheme above using a CSS-based entanglement distillation scheme to correct for noise in the quantum channel. The entanglement distillation part proceeds in two steps, the first involving measurement of the $Z$-type stabilizers, which give Alice and Bob information about the bit errors, and the second involving the $X$-type stabilizers, which give information about the phase errors. Now assume that the key is generated by measuring, in the amplitude basis, each half of the pairs output by the decoding step of the distillation protocol. This is equivalent to skipping the decoding step and instead measuring the logical $\overline{Z}$ operators directly. But in a CSS code these operators are composed entirely of products of $Z$ operators on the individual qubits, and one can reconstruct the value of any desired product from the collection of all the individual outcomes. Knowing $Z_1$, $Z_2$, and $Z_3$ enables us to calculate $Z_1Z_2Z_3$, for instance. 

Thus, Alice and Bob could generate the outcomes of measuring the logical amplitude operators as well as all the $Z$-type stabilizers by first measuring each of their respective qubits in the amplitude basis and then forming the appropriate products of the outcomes. However, the $X$-type stabilizers cannot be generated in this way; in fact, all phase information will be destroyed by making amplitude measurements. \emph{The crucial fact is that Alice and Bob do not need the $X$-type stabilizers at all.} Intuitively this makes sense, as these stabilizers give information about phase errors, but Alice and Bob only care about amplitude information. 

The protocol now proceeds as follows. Alice transmits halves of entangled pairs to Bob, and a random subset are used to estimate the rate of bit and phase errors in order to choose an appropriate CSS code, while the rest are immediately measured in the amplitude basis. Just as in the entanglement distillation protocol, if no suitable code exists because the noise rates are too high, they must abort the procedure. If one does exist, Alice proceeds by constructing the $Z$-type stabilizers according to the chosen code and transmitting them to Bob, who corrects the amplitude errors. They then forget about the phase stabilizers and each constructs the outcomes of measuring the  logical amplitude operators for use as the secret key. 

From the outside there is no way to tell if Alice and Bob have performed the above procedure or actually measured the $X$- and $Z$-type stabilizers directly. Although the phase information has not been exchanged, correction of the phase errors is nevertheless possible in principle. Therefore, the procedure inherits the security of the Lo and Chau protocol in which Alice and Bob actually do create EPR pairs. 

In contrast, from Alice and Bob's point of view, the key is created by two classical information processing protocols. First, Alice sends Bob the stabilizer information which enables him to correct his observed amplitude measurements to match hers. This step is referred to as \emph{information reconciliation} since the goal is to reconcile Bob's amplitude information with Alice's. 
In the next step they use the logical operators to construct a function of the amplitude data, which serves as the key. Due to the entanglement-based picture of the protocol, this has the effect of extracting that part of the amplitude data which is completely uncorrelated with any eavesdropper, and this part of the protocol is termed \emph{privacy amplification}. We can think of the amplitude measurements as a sort of raw key which is then distilled to a truly secret key by running these two protocols in succession.  

Remarkably, we can also remove the need for entanglement entirely. Suppose that in the above protocol Alice immediately measures her halves of the EPR states as she sends the other halves to Bob. These measurements essentially prepare amplitude and phase basis states in the systems underway to Bob. For instance, if her amplitude measurement is $\ket{0}$, Bob's system is now in the state $\ket{0}$, and so on. Originally Alice and Bob agree in advance which observable to measure for each qubit, but suppose instead that they each make a random choice.  
Half the time they choose the same basis, and these outputs are ``sifted'' out by public announcement of the bases and kept for use as the key and for error estimation. 

From the outside there is no way to tell if Alice measures her system after the transmission, so that she is distributing half of an entangled pair, or before, so that she is randomly preparing amplitude or phase eigenstates for Bob to measure. 
Thus, just as the classical key distillation scheme inherits security from entanglement distillation, the \emph{prepare and measure} protocol inherits security from the EPR based version. In fact, this prepare and measure scheme is the original QKD protocol proposed by Bennett and Brassard~\cite{bennett_quantum_1984} and known as BB84; the connection to the version using entanglement was noted by Bennett, Brassard, and Mermin~\cite{bennett_quantum_1992}. Shor and Preskill prove that the BB84 protocol is secure using the reduction of entanglement distillation to information reconciliation and privacy amplification using CSS codes and the reduction of an entanglement-based protocol to a prepare and measure protocol.

\chapter{Characterizing Quantum Information}
\label{chap:char}

In the information game of Chapter~\ref{chap:intro} we made use of the fact that EPR pairs have the property that measurements of either amplitude or phase on one subsystem are entirely predictable using the other subsystem, and we argued that this property is central to the notion of quantum information itself. Here we make good on this claim by showing the converse is true, amplitude and phase predictability implies entanglement, as well as providing two other characterizations of entanglement based on complementarity and showing how these can be extended to characterizations of secret keys. 

The present chapter is divided into four sections. The first presents the converse as stated above. Specifically, following~\citeme{renes_physical_2008}, we show that if there exist measurements on Bob's system which predict Alice's amplitude and phase measurements with low error probability, then Bob can adapt these measurements to create a new system forming an approximate EPR pair with Alice's system. Essentially this is done by coherently performing both measurements in succession, as depicted in  Figure~\ref{fig:entdec}.   

Our approach was inspired by Koashi's complementary control scenario~\cite{koashi_complementarity_2007} in which Bob either tries to guess Alice's amplitude information or somehow help her to prepare a phase eigenstate, and we remark on the connections below. Furthermore, the entanglement recovery procedure is useful in several other scenarios, such as approximate quantum error correction and the quantum information processing protocol known as state merging. 

In the second section we give two other sufficient conditions for entanglement recovery using the uncertainty principle, recounting the results of~\cite{renes_duality_2010}. Again amplitude and phase information play the decisive role, but now the conditions involve a third system. In the first of these, entanglement is implicitly present in the systems shared by Alice and Bob if the amplitude measurement is predictable with low error probability using Bob's system, but high error probability using \emph{any} other system. In the second, entanglement is present if both amplitude and phase are unpredictable in this sense using any other system. These conditions are not as constructive as the first, and instead rely on a powerful method often used in quantum information theory called \emph{decoupling}. 

The third section modifies a result of~\citeme{renes_conjectured_2009} and details how the three characterizations above can be formulated in terms of conditional entropy. Finally, by appealing to the uncertainty principle, we can make a slight modification to the entanglement recovery procedure to instead create \emph{private states}, which are the most general quantum-mechanical description of secret keys. Indeed this was actually the original motivation of~\citeme{renes_physical_2008}.

\definecolor{tud1b}{RGB}{0,90,169}
\definecolor{tud9b}{RGB}{230,0,26}
\definecolor{tud6a}{RGB}{255,224,92}

\newcommand{\ctrol}{\fill(0,0.1) circle(2pt);}
\newcommand{\target}{\draw(0,0.1) circle(4pt);}
\def\vgap{.75}
\def\hgap{1.75}
\def\hgapfudge{.25}
\def\re{7.5}
\begin{figure}[t!]
\begin{center}
\begin{tikzpicture}[thick]
\tikzstyle{empty} = [inner sep=1pt,outer sep=1pt]
\tikzstyle{gate} = [fill=white, draw]
\tikzstyle{ctrl} = [fill,shape=circle,minimum size=4pt,inner sep=0pt,outer sep=0pt]
\tikzstyle{targ} = [draw,shape=circle,minimum size=8pt,inner sep=0pt,outer sep=0pt]

\draw[dotted] (-3.75,0) -- (\re+2,0);
\node at (-2.75,.5) {Alice};
\node at (-2.75,-.5) {Bob};

\node[anchor=east] at (-.5,1) (al) {$A$};
\node[anchor=east] at (0.5,-1) (dl) {$\ket{0}^{C_Z}$};
\node[anchor=east] at (0.5,-1*\vgap-1) (cl) {$\ket{0}^{C_X}$};
\node[anchor=east] at (-.5,-2*\vgap-1) (bl) {$B$};

\node[empty,anchor=west] at (\re,1) (ar) {};
\node[empty,anchor=west] at (\re,-1) (dr) {};
\node[empty,anchor=west] at (\re,-\vgap-1) (cr) {};
\node[empty,anchor=west] at (\re,-2*\vgap-1) (br) {};

\draw (al) -- (ar);
\draw (bl) -- (br);
\draw (cl) -- (cr);
\draw (dl) -- (dr);

\node[gate] (mzc) at (2,-2*\vgap-1) {$\mathcal{M}_Z$};
\node[targ] (mzt) at (2,-1) {};
\draw (mzc.north) -- (mzt.north);

\node[targ] (mxt) at (2+\hgap,-\vgap-1) {};
\node[gate] (mxc) at (2+\hgap,-2*\vgap-1) {$\mathcal{M}_X$};
\draw (mxc.north) -- (mxt.north);

\node[targ] (vt) at (2+2*\hgap+.05,-\vgap-1) {};
\node[ctrl] (vc) at (2+2*\hgap+.05,-1) {};
\draw (vc) -- (vt.south);
\node[empty] (vn) at (2+2*\hgap+.05,-2*\vgap-1) {};

\draw[decorate,decoration={brace,aspect=0.5,amplitude=4},thick] (bl.south west) to node[midway,left,inner sep=5pt] {$\psi^{AB}$} (al.north west);
\draw[decorate,decoration={brace,amplitude=3},thick] (cr.north east) to node[midway,right,inner sep=5pt] {$\psi^{C_XB}$} (br.south east);
\draw[decorate,decoration={brace,aspect=0.75,amplitude=3},thick] (ar.north east) to node[near end,right,inner sep=5pt] {$\Phi^{AC_Z}$} (dr.south east);
\node [below of=mzc] (lb1) {$U_{\!\mathcal{M}_Z}^{BC_Z}$};
\node [below of=vn] (lb2) {$U_{\textsc{cnot}}^{C_ZC_X}$};
\node [below of=mxc] (lb3) {$U_{\!\mathcal{M}_X}^{BC_X}$};

\begin{pgfonlayer}{background}
\node [fill=tud1b!50,rounded corners, dashed,fit=(mzt) (lb1)] {};
\node [fill=tud9b!50,rounded corners,fit=(mxt) (lb3)] {};
\node [fill=tud6a!50,rounded corners,fit=(vc) (vt) (lb2)] {};
\end{pgfonlayer}

\end{tikzpicture}
\caption{\label{fig:entdec} The quantum circuit enabling entanglement recovery from a bipartite state $\psi^{AB}$ by Bob, when he can approximately predict measurement of either conjugate observable $X$ or $Z$ by Alice. It proceeds in three steps. First, Bob coherently performs the measurement $\mathcal{M}_Z$ allowing him to predict $Z$, storing the result in auxiliary system $C_Z$ (unitary $U_{\mathcal{M}_Z}^{BC_Z}$). Next, he coherently performs the measurement $\mathcal{M}_X$ allowing him to predict $X$, storing the result in auxiliary system $C_X$ (unitary $U_{\mathcal{M}_X}^{BC_X}$). Finally, to recover a maximally entangled state in system $C_Z$, he applies a controlled-\textsc{not} gate, with control $C_Z$ and target $C_X$ (unitary $U_{\textsc{cnot}}^{C_ZC_X}$). This procedure also leaves Bob holding the original input state $\psi^{AB}$ in systems $C_X$ and $B$.}
\end{center}
\vspace{-14pt}
\end{figure}
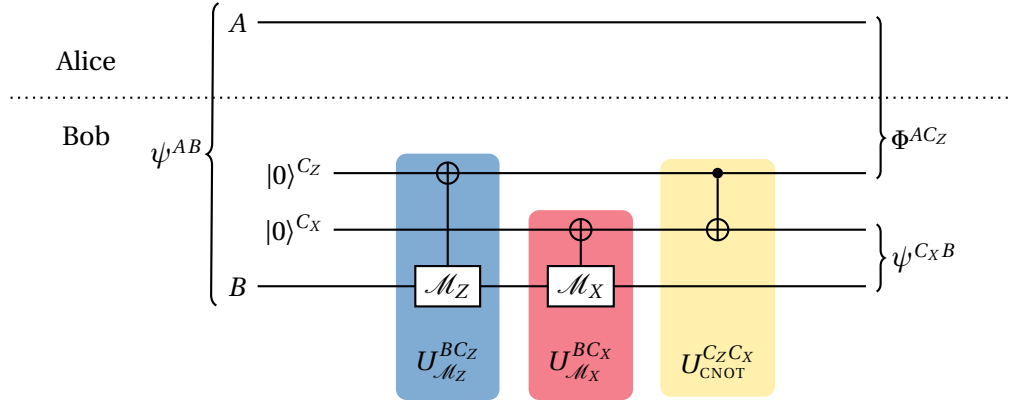

\section{Amplitude and Phase Predictability \& Entanglement}
\label{sec:main}
Let us now specify the setup under consideration more formally. Our two parties Alice and Bob are located some distance apart and each have a technologically-advanced laboratory in which they can manipulate quantum systems. Suppose now that Alice and Bob share a generic bipartite quantum state $\psi^{AB}$. Without loss of generality this state is the $AB$ subsystem of a pure state $\ket{\psi}^{ABE}$ for $E$ the ``environment''.\footnote{In the context of the pure state $\ket{\psi}^{ABE}$, $\psi^{AB}$ denotes the marginal state of the $AB$ system.} We can express this pure state in two ways by expanding Alice's system in the amplitude or phase basis,
\begin{align}
\ket{\psi}^{ABE}=\sum_{z=0}^1 \sqrt{p_z}\ket{z}^A\ket{\varphi_z}^{BE}\qquad \text{and}\qquad \ket{\psi}^{ABE}=\sum_{x=0}^1 \sqrt{q_x}\ket{\widetilde{x}}^A\ket{\vartheta_x}^{BE}.
\end{align}
Here $\ket{z}$ denote amplitude eigenstates according to $Z\ket{z}=(-1)^z\ket{z}$ and similarly $\ket{\widetilde{x}}$ denote phase eigenstates according to $X\ket{\widetilde{x}}=(-1)^x\ket{\widetilde{x}}$. The states $\ket{\varphi_z}^{BE}$ and $\ket{\vartheta_x}^{BE}$ are normalized pure states, but otherwise arbitrary; $p_z$ and $q_x$ are the probabilities that Alice obtains the outcome $z$ and $x$ for amplitude and phase measurements, respectively. 

If Alice makes the amplitude measurement corresponding to the observable $Z$ on her system, Bob can attempt to match her outcome by performing some generalized measurement on his system. This measurement is described most generally by a positive operator valued-measure (POVM) $\mathcal{M}_Z$, which consists of elements positive semidefinite operators $\Lambda_z$ such that $\sum_z\Lambda_z=\mathbbm{1}$.  The probability that he can correctly guess her outcome is given by 
\begin{align}
p_{\rm guess}(Z^A|\mathcal{M}_Z^B)_\psi\equiv\sum_{z=0}^1 {\rm Tr}\left[\left(P_z^A\otimes \Lambda^B_z\right)\psi^{AB}\right]=\sum_{z=0}^1p_z {\rm Tr}\left[\Lambda^B_z\varphi_z^B\right],
\end{align}
where $P_z$ is the projector onto the amplitude state $\ket{z}$, i.e.\ $P_z=\ket{z}\bra{z}$. The subscript on the guessing probability denotes which state we should use to evaluate it.  
To predict Alice's phase measurement, Bob would use a different POVM $\mathcal{M}_X$ with POVM elements $\Gamma_x$. His guessing probability is 
\begin{align}
p_{\rm guess}(X^A|\mathcal{M}_X^B)_\psi\equiv\sum_{x=0}^1 {\rm Tr}\left[\left(\widetilde{P}_x^A\otimes \Gamma^B_x\right)\psi^{AB}\right]=\sum_{x=0}^1 q_x{\rm Tr}\left[\Gamma^B_x\vartheta_x^B\right],
\end{align}
where now $\widetilde{P}_x$ is the projector onto the phase state $\ket{\widetilde{x}}$.

\subsection{Approximate Entanglement Implies Approximate Predictability}
Given the EPR state $\ket{\Phi}^{AB}$, we saw in Chapter~\ref{chap:intro} that Bob can perfectly predict Alice's amplitude and phase measurements. Before moving on to the converse, we can strengthen this to an approximate condition, that if Alice and Bob share a good approximation to $\ket{\Phi}^{AB}$, then the amplitude and phase measurements are approximately predictable in that Bob's error probabilities are small. The trick is to choose the correct notion of ``good approximation''. Such approximate conditions are important in considering practical scenarios, since Bob's prediction of Alice's measurement outcome will never be perfect. For the same reason, they ensure that the idea that predictability of complementary information is what counts for quantum information is truly a physical statement, not merely a mathematical curiosity of the theory. 

If we want the resulting error probabilities to be small, a natural choice is to demand the \emph{trace distance} between the state Alice and Bob actually share, $\psi^{AB}$, and the ideal state $\Phi^{AB}$ be less than some prescribed approximation parameter $\epsilon$. The trace distance between any two states $\rho$ and $\sigma$ is defined as $\tfrac12\left\|\rho-\sigma\right\|_1$, where $\|M\|_1\equiv {\rm Tr}\sqrt{M^\dagger M}$ for any operator $M$. The reason this is an appropriate choice stems from the fact that the trace distance cannot increase under quantum operations such as measurement, and that the guessing probability is directly related to the trace distance of the measured state. More generally, the trace distance between two states is also related to the maximum probability that the two states give different outcomes under any possible measurement.\footnote{See, e.g.~\cite{nielsen_quantum_2000} for an excellent introduction to and explication of the basic results in quantum information theory.} Thus, for small trace distance, the two states behave essentially identically under any possible measurement. 

Using the trace distance we can show that for approximate EPR pairs amplitude and phase are approximately predictable. Suppose that $\tfrac12\left\|\psi^{AB}-\Phi^{AB}\right\|_1\leq \epsilon$ and imagine both Alice and Bob perform the amplitude measurement on their respective systems. If $\Phi^{AB}$ were the actual state, the result would be ${\Phi}^{Z^AZ^B}=\tfrac12\sum_z P_z^A\otimes P_z^B$, where we use the observable $Z^A$, respectively $Z^B$, to denote that the state has been measured and to specify which measurement has been made. For $\psi^{AB}$ they obtain ${\psi}^{Z^AZ^B}=\sum_{z,z'} p_z{\rm Tr}[P_{z'}^B\varphi_z^B]P_z^A\otimes P_{z'}^B$. Computing the trace distance, we find
\begin{align}
\tfrac12\left\|{\Phi}^{Z^AZ^B}-{\psi}^{Z^AZ^B}\right\|_1&=\tfrac12\sum_{z,z'}\left|\tfrac 12\delta_{z,z'}-p_z{\rm Tr}[P_{z'}^B\varphi_z^B]\right|\left\|P_z^A\otimes P_{z'}^{B}\right\|_1\\&=\tfrac12\sum_{z,z'}\left|\tfrac 12\delta_{z,z'}-p_z{\rm Tr}[P_{z'}^B\varphi_z^B]\right|,
\end{align}
which is just the \emph{variational distance} between the ideal distribution $\tfrac 12\delta_{z,z'}$ and the actual distribution $p_z{\rm Tr}[P_{z'}^B\varphi_z^B]$. But the variational distance can also be expressed as follows
\begin{align}
\tfrac12\sum_{z,z'}\left|\tfrac 12\delta_{z,z'}-p_z{\rm Tr}[P_{z'}^B\varphi_z^B]\right|=\max_S\sum_{(z,z')\in S}\left|\tfrac 12\delta_{z,z'}-p_z{\rm Tr}[P_{z'}^B\varphi_z^B]\right|,
\end{align}
where $S$ is any subset of the pairs $(z,z')$. Choosing $S=(z,z')$ for $z\neq z'$ gives
\begin{align}
\tfrac12\left\|{\Phi}^{Z^AZ^B}-{\psi}^{Z^AZ^B}\right\|_1&\geq \sum_{z\neq z'}p_z{\rm Tr}[P_{z'}^B\varphi_z^B]\\
&=1-p_{\rm guess}(Z^A|Z^B)_\psi.
\end{align}
The latter expression is a slight abuse of notation, using the observable $Z^B$ to denote Bob's measurement.  
Because the trace distance cannot increase under the measurement, $p_{\rm guess}(Z^A|Z^B)_\psi\geq 1-\epsilon$. The same conclusion holds for the phase measurement, $p_{\rm guess}(X^A|X^B)_\psi\geq 1-\epsilon$. 

\subsection{Approximate Amplitude and Phase Predictability Implies Entanglement}
\label{sub:impent}

Now we can state and prove the converse, that approximate amplitude and phase predictability implies entanglement. We first establish a lemma which will also be used in the subsequent entanglement characterizations.  


\begin{lemma}
\label{lem:entdec}
Given a state $\ket{\psi}^{ABE}=\sum_z \sqrt{p_z}\ket{z}^A\ket{\varphi_z}^{BE}$, let $\ket{\psi}^{C_XBE}$ be the identical state with system $C_X$ replacing $A$, and define $\ket{\psi_Z}^{AC_ZBE}=\sum_{z}\sqrt{p_z}\ket{z}^A\ket{z}^{C_Z}\ket{\varphi_z}^{BE}$. If there exist partial isometries $U_1^{B\rightarrow C_ZB}$ and $U_2^{C_ZB\rightarrow C_ZC_XB}$ such that 
\begin{align}
&^{AC_ZBE\!\!}\braket{\psi_Z|U_1^{B\rightarrow C_ZB}|\psi}^{ABE}\geq 1-\epsilon_1,\qquad \text{and}\\
&\left(^{AC_Z\!\!}\bra{\Phi}\,^{C_XBE\!\!}\bra{\psi}\right)U_2^{C_ZB\rightarrow C_ZC_XB}\ket{\psi_Z}^{AC_ZBE}\geq 1-\epsilon_2,
\end{align}
then for $U^{B\rightarrow C_ZC_XB}=U_2^{C_ZB\rightarrow C_ZC_XB}U_1^{B\rightarrow C_ZB}\ket{\psi}^{ABE}$,
\begin{align}
\tfrac 12\left\|\ket{\Phi}^{AC_Z}\ket{\psi}^{C_XBE}-U^{B\rightarrow C_ZC_XB}\ket{\psi}^{ABE}\right\|_1\leq \sqrt{2\epsilon_1}+\sqrt{2\epsilon_2}.
\end{align}
\end{lemma}

\begin{proof}
The fidelity $F(\psi,\phi)=\braket{\psi|\phi}$ between two pure states gives an upper bound on their trace distance, $\tfrac12\left\|\psi-\phi\right\|_1\leq \sqrt{1-F(\psi,\phi)^2}$, so that  fidelity greater than $1-\epsilon$ translates into trace distance less than $\sqrt{2\epsilon}$. Since the trace distance is invariant under unitaries and partial isometries, the lemma follows from the triangle inequality.
\end{proof}
\begin{theorem}
\label{thm:entdec}
If  $p_{\rm guess}(Z^A|\mathcal{M}_Z^B)_\psi\geq 1-\epsilon_1$ and $p_{\rm guess}(X^A|\mathcal{M}_Z^B)_{\psi}\geq 1-\epsilon_2$ for some measurements $\mathcal{M}_Z^B$ and $\mathcal{M}_X^{B}$ on a state $\psi^{AB}$, then there exists a partial isometry $U^{B\rightarrow C_ZC_XB}$ such that 
\begin{align}
\tfrac 12\left\|\ket{\Phi}^{AC_Z}\ket{\psi}^{C_XBE}-U^{B\rightarrow BC_ZC_X}\ket{\psi}^{ABE}\right\|_1\leq \sqrt{2\epsilon_1}+\sqrt{2\epsilon_2}.
\end{align}
\end{theorem}
\begin{proof}
We use the measurements to define the two isometries required for Lemma~\ref{lem:entdec}. 
For the first isometry $U_1^{B\rightarrow C_ZB}$ we may use the coherent implementation of the measurement $\mathcal{M}_Z^B$, which stores the measurement result in system $C_Z$. Performing the measurement coherently produces the state $U_1^{B\rightarrow C_ZB}\ket{\psi}^{ABE}$, which without loss of generality takes the form 
\begin{align}
U_1^{B\rightarrow C_ZB}\ket{\psi}^{ABE}=\sum_{z,z'}\sqrt{p_z}\ket{z}^A\ket{z'}^{C_Z}\sqrt{\Lambda^B_{z'}}\ket{\varphi_z}^{BE}.
\end{align}
Now compute the overlap of this state with the state $\ket{\psi_Z}^{AC_ZBE}=\sum_z \sqrt{p_z}\ket{z}^A\ket{z}^{C_Z}\ket{\varphi_z}^{BE}$, which would be the ideal output of the coherent measurement process.
\begin{align}
\braket{\psi_Z|U_1^{B\rightarrow C_ZB}|\psi}^{ABE}&=\sum_z p_z\bra{\varphi_z}\sqrt{\Lambda^B_z}\ket{\varphi_z}^{BE}\\
&\geq \sum_z p_z\bra{\varphi_z}\Lambda^B_z\ket{\varphi_z}^{BE}\\
&=p_{\rm guess}(Z^A|\mathcal{M}_Z^B)_\psi,
\end{align}
using the fact that $\sqrt{\Lambda}\geq \Lambda$ for $0\leq \Lambda\leq \mathbbm{1}$. Hence, we have the first condition of Lemma~\ref{lem:entdec}. 

For the second, let $V^{B\rightarrow C_XB}$ be the partial isometry which coherently implements the measurement $\mathcal{M}_X^B$, storing the result in system $C_X$. Coherently measuring $\ket{\psi_Z}^{AC_ZBE}$ gives 
\begin{align}
V^{B\rightarrow C_XB}\ket{\psi_Z}^{AC_ZBE}&=\sum_z \sqrt{p_z}\ket{z}^A\ket{z}^{C_Z}V^{B\rightarrow C_XB}\ket{\varphi_z}^{BE}\\
&=\tfrac{1}{\sqrt{2}}\sum_z \ket{z}^A\ket{z}^{C_Z}V^{B\rightarrow C_XB}\sum_x(-1)^{xz}\sqrt{q_x}\ket{\vartheta_x}^{BE}\\
&=\tfrac{1}{\sqrt{2}}\sum_z \ket{z}^A\ket{z}^{C_Z}\sum_{x,x'}(-1)^{xz}\sqrt{q_x}\ket{\widetilde{x}'}\sqrt{\Gamma_{x'}^B}\ket{\vartheta_x}^{BE}.
\end{align}
Here we have made use of the algebraic relationship between the two bases. 
Ideally the output would be the state
\begin{align}
\ket{\psi'_Z}^{AC_ZC_XBE}&=\tfrac{1}{\sqrt{2}}\sum_{z} \ket{z}^A \ket{z}^{C_Z}\sum_x\sqrt{q_x}\,(-1)^{xz}\,\ket{\widetilde{x}}^{C_X}\ket{\vartheta_x}^{BE},
\end{align}
and computing the fidelity between the ideal and actual outputs gives
\begin{align}
^{AC_ZC_XBE\!\!}\bra{\psi'_Z}V^{B\rightarrow C_XB}\ket{\psi_Z}^{AC_ZBE}
&=\tfrac12\sum_{z,x,x',x''}\sqrt{q_xq_{x''}}(-1)^{z(x-x'')}\braket{\widetilde{x}''|\widetilde{x}'}
\bra{\vartheta_{x''}}\sqrt{\Gamma_{x'}^B}\ket{\vartheta_x}^{BE}\\
&=\sum_x q_x\bra{\vartheta_{x}}\sqrt{\Gamma_{x}^B}\ket{\vartheta_x}^{BE}\\
&\geq \sum_x q_x\bra{\vartheta_{x}}{\Gamma_{x}^B}\ket{\vartheta_x}^{BE}\\
&=p_{\rm guess}(X^A|\mathcal{M}_X^B)_\psi.
\end{align}

We may also express $\ket{\psi'_Z}^{AC_ZC_XBE}$ as $\tfrac{1}{\sqrt{2}}\sum_{z} \ket{z}^A \ket{z}^{C_Z}(X^z)^{C_X}\ket{\psi}^{C_XBE}$, and therefore applying a control-\textsc{not} $W_{\textsc{cnot}}^{C_ZC_X}$ with $C_Z$ as the control and $C_X$ as the target to the ideal output gives $\ket{\Phi}^{AC_Z}\ket{\psi}^{C_XBE}$. Since the fidelity is invariant under partial isometries, the second condition of Lemma~\ref{lem:entdec} holds for $U_2^{C_ZB\rightarrow C_ZC_XB}=W_{\textsc{cnot}}^{C_ZC_X}V^{B\rightarrow C_XB}$, completing the proof. 
\end{proof}

\subsection{Further Uses of the Entanglement Recovery Operation}

We originally introduced the environment system $E$ as the purification of the joint state held by Alice and Bob, but of course we can look at it the other way around; generically Alice and Bob jointly hold the purification of system $E$. However, our entanglement recovery operation has done more than just recover entanglement, as it reveals that when Alice's amplitude and phase measurements are predictable by Bob, he implicitly holds the purification of $E$ by himself. This follows because the recovery operation also produces (a good approximation to) the state $\ket{\psi}^{C_XBE}$, which is identical to the initial state except system $A$ is replaced by $C_X$, held by Bob. In the following chapter we shall use this property to construct protocols for \emph{state merging}, in which Alice attempts to merge her state with Bob by using classical or quantum communication.   

Theorem~\ref{thm:entdec} may be regarded as giving necessary and sufficient conditions on the existence of an approximate quantum error-correction scheme: Approximate error correction is possible when amplitude and phase information can each approximately be recovered. The schemes discussed in Chapter 2 based on quantum error-correcting codes were \emph{perfect} in the sense that the input quantum state can be perfectly recovered if the error is of the correctable type. Approximate error-correction sets the more modest goal of only recovering the input approximately. 

Entanglement recovery is relevant to this goal because we can always mimic the initial single-system input to the error-correction problem as half of an EPR pair, the other half of which is then measured in an appropriate basis. The basis is given by complex conjugating the coefficients of the original basis, which follows because the EPR state can be written as $\ket{\Phi}^{AB}=\frac{1}{\sqrt{2}}\sum_j \ket{\xi_j^*}^A\ket{\xi_j}^B$ for any basis $\{\ket{\xi_j}=\xi_{j0}\ket{0}+\xi_{j1}\ket{1}\}_{j=0}^1$, where $\ket{\xi_j^*}=\xi_{j0}^*\ket{0}+\xi_{j1}^*\ket{1}$.  This was precisely the method used in reducing the QKD scheme based on EPR pairs to one involving only preparation and measurement of single systems. From this line of reasoning it follows from a result of Schumacher~\cite{schumacher_sending_1996} that if entanglement can be approximately recovered by the scheme, then the approximation parameter sets a lower bound on the average fidelity with which single systems can be recovered by the same procedure.

In the preceding analysis, we have assumed that Alice's system is a qubit, whereas Bob's system is arbitrary. But the result may be easily extended to the case that Alice holds a $d$-level system by using the more general amplitude and phase operators defined by 
\begin{align}
\label{eq:wh}
X=\sum_{k=0}^{d-1}\ket{k\oplus 1}\bra{k}\qquad\text{and}\qquad Z=\sum_{k=0}^{d-1}e^{2\pi i k/d}\ket{k}\bra{k}.
\end{align}
Often these are called the Weyl-Heisenberg operators, as they have similar properties to the position and momentum operators of continuous-variable systems. Here the crucial point is that the algebraic properties of the amplitude and phase operators used in Theorem~\ref{thm:entdec} hold for higher-dimensional systems as well. In the sequel, we shall continue to specialize to the qubit case. 

\section{Duality \& Decoupling}
\label{sec:dnd}
The uncertainty principle Equation~(\ref{eq:jcbjmr}) establishes a tradeoff in how well Alice's amplitude measurement can be predicted using system $B$ and how well her phase measurement can be predicted using system $E$. In the previous section the sufficient conditions for entanglement were of the former type, but the tradeoff suggests that we might to be able to find sufficient conditions of the latter type and focus instead on what information system $E$ does \emph{not} have, rather than what information system $B$ does have. Concentrating on lack of information and building protocols by destroying correlations is the essence of the decoupling approach to quantum information processing, which
goes back to work on approximate error-correction by Schumacher and Westmoreland~\cite{schumacher_approximate_2002}
and has found wide application to constructing information processing protocols such as state merging~\cite{horodecki_partial_2005,horodecki_quantum_2007} and noisy channel coding~\cite{horodecki_quantum_2008,hayden_random_2008,hayden_decoupling_2008,klesse_random_2008} that we shall encounter in Chapter~\ref{chap:proc}.\footnote{The decoupling approach has also been extended to quantum channels, instead of quantum states as described here, in~\cite{klesse_approximate_2007,kretschmann_information-disturbance_2008,kretschmann_complementarity_2008,beny_general_2010,ng_simple_2010}.} 

In the decoupling approach one tries to show that Alice's system is completely uncorrelated with system $E$ in order to infer that Bob's system is entangled with Alice's. Here, however, we shall be able to show that it suffices for this purpose to ensure that $E$ has no information about Alice's amplitude or phase. This reflects our main theme that what really counts in quantum information is classical information about complementary observables.  
Part of the appeal of decoupling is that it allows us to avoid the problem of constructing the isometries needed for Lemma~\ref{lem:entdec}. Instead, the isometries are automatically constructed by appealing to Uhlmann's theorem on the relationship between fidelity of mixed states and that of their possible purifications.

However, from the uncertainty principle we are only entitled to expect that if $Z^A$ is predictable from $B$, then $X^A$ is unpredictable from $E$, but not the converse. Were the converse true in general, we could immediately establish that lack of information in $E$ is sufficient to imply the presence of entanglement because it would imply the conditions we already have in Lemma~\ref{lem:entdec} and Theorem~\ref{thm:entdec}. Absent the converse, it is not immediately clear that this approach will work. 

Note that the converse does hold if it happens the state $\ket{\psi}^{ABE}$ saturates Equation~(\ref{eq:jcbjmr}), so that $H(Z^A|B)_\psi+H(X^A|E)_\psi=1$. Two separate sufficient conditions for equality in the uncertainty principle are derived in~\cite{renes_conjectured_2009}, and these take the simple form $p_{\rm guess}(X^A|B)_\psi=1$ and $p_{\rm guess}(Z^A|E)_\psi=1$. Equivalently, each of these conditions implies the complementary form $H(X^A|B)+H(Z^A|E)=1$ is trivially saturated, since either $H(X^A|B)=0$ and thus $H(Z^A|E)=1$ or $H(Z^A|E)=0$ and thus $H(X^A|E)=1$. Luckily, it turns out that due to the structure of Lemma~\ref{lem:entdec}, either of these two equality conditions can be satisfied without loss of generality to the entanglement criteria. We shall make use of both in the two results presented next. 

First, we need to formally characterize the unpredictability of measurements on Alice's system when making use of the purification system $E$. The most straightforward approach would be to say that the associated guessing probabilities are small, even for the optimal measurement. However, optimal measurements are quite often difficult to specify in quantum information theory. To sidestep this problem, we may instead use the following quantity,
\begin{align}
\label{eq:securedef}
p_{\rm secure}(Z^A|E)_\psi=1-\tfrac12\big\|{\psi}^{Z^AE}-\tfrac12\mathbbm{1}^A\otimes\psi^E\big\|_1,
\end{align} 
and say that $E$ has no information about $Z^A$ when $p_{\rm secure}(Z^A|E)_\psi$ is nearly one. Another possibility would be to phrase matters in terms of the conditional entropy, stating that $H(Z^A|E)$ is large. It turns out that such an entropic condition implies that of Equation~(\ref{eq:securedef}) and we shall return to this point in the next section. 

A bound on $p_{\rm secure}(Z^A|E)_\psi$ also implies a bound on the guessing probability, as follows. Suppose system $E$ is measured with some POVM $\mathcal{M}_Z^E=\{\Lambda_z^E\}$. The probability distributions of measurement outcomes on the real and ideal states are $p_{z,z'}=p_z{\rm Tr}[\Lambda_{z'}^E\varphi_z^E]$ and $p'_{z,z'}=\tfrac12{\rm Tr}[\Lambda_{z'}^E\varphi^E]$, respectively. For the variational distance we find
\begin{align}
\tfrac12\sum_{z,z'}|p_{z,z'}-p'_{z,z'}|&\geq \tfrac12\sum_{z}p_{z,z}-p'_{z,z}\\
&=\tfrac12\sum_{z}\left(p_z{\rm Tr}[\Lambda_{z}^E\varphi_z^E]-\tfrac12{\rm Tr}[\Lambda_z^E\varphi^E]\right)\\
&=\tfrac12 p_{\rm guess}(Z^A|\mathcal{M}_Z^E)_\psi-\tfrac14,
\end{align}
and therefore $p_{\rm secure}(Z^A|E)_\psi\geq 1-\epsilon$ implies $p_{\rm guess}(Z^A|\mathcal{M}_Z^E)_\psi\leq \tfrac 12+2\epsilon$ for any measurement $\mathcal{M}_Z^E$. 
Observe that the quantity $p_{\rm secure}$ also implies the outcome of the $Z^A$ measurement is nearly random. This accounts for the name `secure' since effectively this means Alice can generate a secure secret key bit by this measurement. 

Now we are ready to state the new entanglement conditions. The first says that Alice and Bob implicitly share entanglement if Alice's amplitude measurement can be predicted using $B$ but not $E$. This is almost the same as saying that Alice and Bob can generate a shared secret key from their state, a point we return to in Section~\ref{sec:secretkey}.  
\begin{theorem}
\label{thm:dw}
If $p_{\rm secure}(Z^A|E)_\psi\geq 1-\epsilon_2$ and $p_{\rm guess}(Z^A|\mathcal{M}_Z^B)_\psi\geq 1-\epsilon_1$ for some measurement $\mathcal{M}_Z^B$, then there exists a partial isometry $U^{B\rightarrow C_ZC_XB}$ such that 
\begin{align}
\tfrac 12\left\|\ket{\Phi}^{AC_Z}\ket{\psi}^{C_XBE}-U^{B\rightarrow BC_ZC_X}\ket{\psi}^{ABE}\right\|_1\leq \sqrt{2\epsilon_1}+\sqrt{2\epsilon_2}.
\end{align}
\end{theorem}
\begin{proof}
From the proof of Theorem~\ref{thm:entdec}, the first condition of Lemma~\ref{lem:entdec} is fulfilled for $U_1^{B\rightarrow C_ZB}$ the coherent implementation of the measurement $\mathcal{M}_Z^B$. For the second condition, consider the implications of $p_{\rm secure}(Z^A|E)_\psi\geq 1-\epsilon_2$ for the state $\ket{\psi_Z}^{AC_ZBE}$. Since tracing out $C_ZB$ from $\ket{\psi_Z}^{AC_ZBE}$ gives the same result as measuring the amplitude $Z^A$ of the state $\ket{\psi}^{ABE}$, we have
\begin{align}
\tfrac12\left\|\psi_Z^{AE}-\tfrac12\mathbbm{1}^A\otimes \psi^E\right\|_1\leq \epsilon_2.
\end{align}
The trace distance gives a lower bound to the fidelity, so that 
\begin{align}
F\left(\psi_Z^{AE},\tfrac12\mathbbm{1}^A\otimes \psi^E\right)\geq 1-\epsilon_2.
\end{align}
By Uhlmann's theorem, the fidelity of two mixed states is identical to the largest fidelity of their possible purifications. Two possible purifications of the two states in question are $\ket{\psi_Z}^{AC_ZBE}$ and $\ket{\Phi}^{AC_Z}\ket{\psi}^{C_XBE}$, and since all other purifications are related by isometries involving the purifying system, we have 
\begin{align}
\left(^{AC_Z\!\!}\bra{\Phi}\,^{C_XBE\!\!}\bra{\psi}\right)U_2^{C_ZB\rightarrow C_ZC_XB}\ket{\psi_Z}^{AC_ZBE}\geq 1-\epsilon_2
\end{align}
for some $U_2^{C_ZB\rightarrow C_ZC_XB}$. This is the sought-after second condition and completes the proof.
\end{proof}

Observe that although Theorem~\ref{thm:dw} calls for $p_{\rm secure}(Z^A|E)_\psi$ to be large, we actually apply this condition to the state $\ket{\psi_Z}^{AC_ZBE}$ for which $p_{\rm guess}(Z^A|C_ZB)=1$. Therefore the uncertainty relation $H(X^A|C_ZB)_\psi+H(Z^A|E)_\psi=1$ holds, and we are essentially able to trade large $p_{\rm secure}(Z^A|E)_\psi$ for large $p_{\rm guess}(X^A|C_ZB)_\psi$ as is needed for Lemma~\ref{lem:entdec}. Indeed, it follows from the discussion prior to Lemma~\ref{lem:entdec} that the following is an immediate corollary to Theorem~\ref{thm:dw}. Using the isometry $U^{B\rightarrow C_ZC_XB}$ to define the measurement $\mathcal{M}_X^B=\{U^{\dagger B\rightarrow C_ZC_XB}\widetilde{P}_x^{C_Z}U^{B\rightarrow C_ZC_XB}\}$, we have
\addtocounter{corollary}{1}
\begin{corollary}
\label{cor:secguessB}
If $p_{\rm secure}(Z^A|E)_\psi\geq 1-\epsilon$ and $p_{\rm guess}(Z^A|\mathcal{M}_Z^B)_\psi= 1$ for some measurement $\mathcal{M}_Z^B$, then there exists a measurement $\mathcal{M}_X^{B}$ such that $p_{\rm guess}(X^A|\mathcal{M}_X^{B})_{\psi}\geq 1-\sqrt{2\epsilon}$. 
\end{corollary}
\noindent This is Theorem 4.2(a) of~\citeme{renes_duality_2010}.
Note that we have dropped the explicit use of $\ket{\psi_Z}$ by stipulating that $p_{\rm guess}(Z^A|\mathcal{M}_Z^B)_\psi=1$. 

Continuing the trend of denying information to $E$, we might hope that Alice and Bob share entanglement when the amplitude and phase are unpredictable using $E$. However, a simple counterexample shows that this cannot be true in general. Define $\ket{\psi}^{ABE}=\tfrac{1}{\sqrt{2}}\left(\ket{0}+i\ket{1}\right)^A\otimes \ket{\varphi}^{BE}$; here Alice's system is an eigenstate of the observable $XZ$. Due to the product structure, both $Z^A$ and $X^A$ are unpredictable using either $E$ or $B$. This is to be expected in light of the preceding discussion on the need to saturate the uncertainty principle. One way to avoid this problem is to require that not only is the amplitude measurement unpredictable using $E$, but the phase measurement is unpredictable using $E$ \emph{even assuming $E$ could predict the amplitude}. Formally, we require $p_{\rm secure}(X^A|C_ZE)_{\psi_Z}$ to be large, which again involves the state $\ket{\psi_Z}$ that saturates the uncertainty principle, though this time $C_Z$ is joined with $E$, not $B$. 

\begin{theorem}
\label{thm:2xdecoupling}
If $p_{\rm secure}(X^A|C_ZE)_{\psi_Z}\geq 1-\epsilon_1$ and $p_{\rm secure}(Z^A|E)_\psi\geq 1-\epsilon_2$, then there exists a partial isometry $U^{B\rightarrow C_ZC_XB}$ such that 
\begin{align}
\tfrac 12\left\|\ket{\Phi}^{AC_Z}\ket{\psi}^{C_XBE}-U^{B\rightarrow BC_ZC_X}\ket{\psi}^{ABE}\right\|_1\leq \sqrt{2\epsilon_1}+\sqrt{2\epsilon_2}.
\end{align}
\end{theorem}
\begin{proof}
From the proof of Theorem~\ref{thm:dw}, the second condition here implies the second condition of Lemma~\ref{lem:entdec} is satisfied. To show the first, consider the implications of 
$p_{\rm secure}(X^A|C_ZE)_{\psi_Z}\geq 1-\epsilon_1$ for the state $\ket{\psi_Z}^{AC_ZBE}$. We may express the state as follows
\begin{align}
\ket{\psi_Z}^{AC_ZBE}&=\tfrac{1}{\sqrt{2}}\sum_x\ket{\widetilde{x}}^A\sum_z \sqrt{p_z}\,(-1)^{xz}\,\ket{z}^{C_Z}\ket{\varphi_z}^{BE}\\
&=\tfrac{1}{\sqrt{2}}\sum_x\ket{\widetilde{x}}^A(Z^x)^{C_Z}\sum_z \sqrt{p_z}\ket{z}^{C_Z}\ket{\varphi_z}^{BE}\\
&=\tfrac{1}{\sqrt{2}}\sum_x\ket{\widetilde{x}}^A(Z^x)^{C_Z}\ket{\psi}^{C_ZBE}.
\label{eq:psizphase}
\end{align}
Defining $\ket{\eta_x}^{C_ZBE}\equiv(Z^x)^{C_Z}\ket{\psi}^{C_ZBE}$ and $\eta\equiv\tfrac12\sum_x\eta_x^{C_ZE}$, we have 
\begin{align}
1-p_{\rm secure}(X^A|C_ZE)_{\psi_Z}&=\tfrac12\Big\|\tfrac12\sum_x\ket{\widetilde{x}}\bra{\widetilde{x}}^A\otimes \eta^{C_ZE}_x-\tfrac12\mathbbm{1}^A\otimes\eta^{C_ZE}\Big\|_1\\
&=\tfrac14\sum_x\left\|\eta^{C_ZE}_x-\eta^{C_ZE}\right\|_1.
\end{align}
But since the trace distance is invariant under unitary operations, in particular $(Z^x)^{C_Z}$,
$\left\|\eta^{C_ZE}_x-\eta^{C_ZE}\right\|_1=\big\|\eta^{C_ZE}_{x'}-\eta^{C_ZE}\big\|_1$ for all $x,x'$. Observe that $\eta^{C_ZE}=\psi_Z^{C_ZE}$ since the random phase flip has the effect of ``measuring'' the amplitude of $C_Z$. As $\ket{\eta_0}^{C_ZBE}=\ket{\psi}^{C_ZBE}$, we can therefore infer that $\tfrac12\big\|\psi^{C_ZE}-\psi_Z^{C_ZE}\big\|_1\leq \epsilon_1$, or equivalently $\tfrac12\big\|\psi^{AE}-\psi_Z^{AE}\big\|_1\leq \epsilon_1$. Converting trace distance to fidelity and applying Uhlmann's theorem, we find there exists an isometry $U_1^{B\rightarrow C_ZB}$ such that 
\begin{align}
^{AC_ZBE\!\!}\bra{\psi_Z}U_1^{B\rightarrow C_ZB}\ket{\psi}^{ABE}\geq 1-\epsilon_1.
\end{align}
Thus, the first condition of Lemma~\ref{lem:entdec} is satisfied, completing the proof.
\end{proof}

Again, the resulting isometry can be used to define the measurement $\mathcal{M}_Z^B$ in the following, which is Theorem 4.2(b) of~\citeme{renes_duality_2010}.
\begin{corollary}
\label{cor:secguessR}
If $p_{\rm secure}(X^A|E)_\psi\geq 1-\epsilon$ and $p_{\rm guess}(Z^A|\mathcal{M}_Z^E)_\psi= 1$ for some measurement $\mathcal{M}_Z^E$, then there exists a measurement $\mathcal{M}_Z^{B}$ such that $p_{\rm guess}(Z^A|\mathcal{M}_Z^{B})_{\psi}\geq 1-\sqrt{2\epsilon}$. 
\end{corollary}
Figure~\ref{fig:triality} illustrates the contents of Theorems 1, 2, and 3 by indicating which system must have what kind of information, or lack thereof, in order to infer the presence of entanglement between Alice and Bob.

\def\dist{3.5}
\def\dx{.5}
\def\height{.25}
\def\gap{.1}
\def\nodes{
\fill[tud6a!50] (0,0) circle (.75);
\node[anchor=center] at (0,0) (alice) {\Large $A$};
\fill[tud6a!50] (0.866*\dist,.5*\dist) circle (.75);
\node[anchor=center] at (0.866*\dist,.5*\dist) (bob) {\Large $B$};
\fill[tud6a!50] (0.866*\dist,-.5*\dist) circle (.75);
\node at (0.866*\dist,-.5*\dist) (eve) {\Large $E$};
}
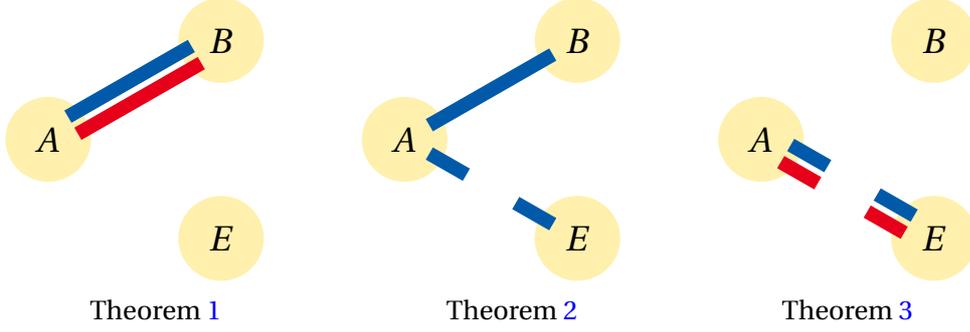
\begin{figure}[th!]
\begin{center}
\begin{tikzpicture}[scale=.75,thick]

\begin{scope}[shift={(-6.25,0)}]
\nodes
\fill[tud9b,rotate around={30:(0,0)}] (\dx,-\height-\gap/2) rectangle (\dist-\dx,-\gap/2);
\fill[tud1b,rotate around={30:(0,0)}] (\dx,\height+\gap/2) rectangle (\dist-\dx,+\gap/2);
\node at (0.5376*\dist,-3) {Theorem~\ref{thm:entdec}};
\end{scope}

\begin{scope}[shift={(0,0)}]
\nodes
\fill[tud1b,rotate around={30:(0,0)}] (\dx,-\height/2) rectangle (\dist-\dx,\height/2);
\fill[tud1b,rotate around={-30:(0,0)}] (\dx,-\height/2) rectangle (\dist-\dx,\height/2);
\fill[white,rotate around={-30:(0,0)}] (2.5*\dx,-\height) rectangle (\dist-2.5*\dx,\height);
\node at (0.5376*\dist,-3) {Theorem~\ref{thm:dw}};

\end{scope}

\begin{scope}[shift={(6.25,0)}]
\nodes
\fill[tud9b,rotate around={-30:(0,0)}] (\dx,-\height-\gap/2) rectangle (\dist-\dx,-\gap/2);
\fill[tud1b,rotate around={-30:(0,0)}] (\dx,\height+\gap/2) rectangle (\dist-\dx,+\gap/2);
\fill[white,rotate around={-30:(0,0)}] (2.5*\dx,-\height-.1) rectangle (\dist-2.5*\dx,\height+.1);
\node at (0.433*\dist,-3) {Theorem~\ref{thm:2xdecoupling}};
\end{scope}

\end{tikzpicture}
\caption{\label{fig:triality} Graphical depiction of the contents of Theorems~\ref{thm:entdec}, \ref{thm:dw}, and \ref{thm:2xdecoupling}. All three theorems specify conditions under which Alice and Bob can transform their shared state $\psi^{AB}$ into a collection of EPR pairs using only local operations. Theorem~\ref{thm:entdec} shows this can be done when both the amplitude and phase of Alice's system are correlated with Bob's system in that he could reliably predict either, depicted by the blue (amplitude) and red (phase) lines joining Alice and Bob. Theorem~\ref{thm:dw} shows that amplitude correlation with Bob and \emph{un}correlation with the environment leads to the same conclusion. Finally, Theorem~\ref{thm:2xdecoupling} establishes that appropriate uncorrelation of both amplitude and phase suffices to infer that the Alice-Bob system to be entangled.}
\end{center}
\vspace{-14pt}
\end{figure}

\section{Entropic Characterizations}
As advertised after Equation~(\ref{eq:securedef}), another possible formalization of ``unpredictability'' is using the conditional entropy: The amplitude measurement outcome $Z^A$ is unpredictable using $E$ when $H(Z^A|E)$ is large. Owing to the connection between conditional entropy and the quantity $p_{\rm secure}$ given by the following lemma, we can establish entropic conditions on entanglement from the results of the previous section. This was partially investigated in~\citeme{renes_conjectured_2009}.
\begin{lemma}
\label{lem:condent}
If $H(Z^A|E)_\psi\geq 1-\epsilon^2$, then $p_{\rm secure}(Z^A|E)_\psi\geq 1-\epsilon$.
\end{lemma}
\begin{proof}
The proof relies on the connection between relative entropy $D(\rho||\sigma)$ of two states $\rho$ and $\sigma$ and trace distance between them, in particular the bound $\tfrac{1}{2\ln 2}\left\|\rho-\sigma\right\|_1^2\,\leq\, D(\rho||\sigma)$~\cite{schumacher_approximate_2002}. This may be more conveniently expressed as 
$\tfrac{1}{2}\left\|\rho-\sigma\right\|_1\,\leq\, \sqrt{D(\rho||\sigma)}$.
By direct calculation it is easy to show
\begin{align}
D\left(\psi_Z^{AE}\,\big\|\,\tfrac12\mathbbm{1}^A\otimes\psi^E\right)=1-H(Z^A|E)_\psi.
\end{align}
Using the bound on $H(Z^A|E)_\psi$ and the definition of $p_{\rm secure}(Z^A|E)_\psi$ completes the proof.  
\end{proof}

\begin{theorem}
\label{thm:entropic}
Given any of the following pairs of conditions,
\begin{align*}
\begin{array}{lll}
{\rm (1)}\quad H(Z^A|B)_\psi\leq \epsilon_1^2,\hspace{5mm}\phantom{,} & 
{\rm (2)}\quad H(Z^A|B)_\psi\leq \epsilon_1^2,\hspace{10mm}\phantom{,} &
{\rm (3)}\quad H(X^A|C_ZE)_{\psi_Z}\geq 1-\epsilon_1^2,\\[1mm]
\phantom{(1)\quad }H(X^A|B)_\psi\leq \epsilon_2^2&
\phantom{(2)\quad }H(Z^A|E)_\psi\geq 1-\epsilon_2^2 &
\phantom{(3)\quad }H(Z^A|E)_\psi\geq 1-\epsilon_2^2 
\end{array}
\end{align*}
there exists a partial isometry $U^{B\rightarrow C_ZC_XB}$ such that 
\begin{align}
\tfrac 12\left\|\ket{\Phi}^{AC_Z}\ket{\psi}^{C_XBE}-U^{B\rightarrow BC_ZC_X}\ket{\psi}^{ABE}\right\|_1\leq \sqrt{2\epsilon_1}+\sqrt{2\epsilon_2}.
\end{align}
\end{theorem}
\begin{proof}
Using Lemma~\ref{lem:condent} for the last pair, we can apply Theorem~\ref{thm:2xdecoupling}. But since $H(Z^A|B)_\psi=H(Z^A|B)_{\psi_Z}$, (1) and (2) each separately imply (3) by the uncertainty principle Equation~(\ref{eq:jcbjmr}). 
\end{proof}
That the first pair of entropic conditions implies that Alice and Bob share entanglement is a variation of a related result found by Christandl and Winter that quantum channels are useful for transmitting entanglement if they could be used to reliably transmit classical amplitude and phase information~\cite{christandl_uncertainty_2005}. 
The fact that the first pair of entropic conditions in Theorem~\ref{thm:entropic} are sufficient for systems $A$ of arbitrary dimension is actually somewhat surprising, as it is known that just because the conditional entropy $H(Z^A|B)$ is small does not imply that there exists a measurement $\mathcal{M}_Z^B$ such that $H(Z^A|\mathcal{M}_Z^B)=H(Z^A|B)$, let alone that the guessing probability $p_{\rm guess}(Z^A|\mathcal{M}_Z^B)$ is large.
In fact, Ruskai has shown that Bob's conditional marginal states must all commute pairwise for the conditional entropy to be achievable~\cite{ruskai_inequalities_2002}. 

The gap can be simply illustrated by the following example implicitly given by Holevo~\cite{holevo_statistical_1973}, in which Bob's state conditioned on Alice's amplitude basis measurement is a randomly-selected amplitude or phase eigenstate,
\begin{align}
\ket{\psi}^{AB}=\tfrac 12\sum_{t=0}^3 \ket{t}^A\ket{\varphi_t}^B,
\end{align}
for $\ket{\varphi_0}=\ket{0}$, $\ket{\varphi_1}=\ket{1}$, $\ket{\varphi_2}=\ket{+}$, and $\ket{\varphi_3}=\ket{-}$. By direct calculation we find $H(Z^A|B)=1$, where now $Z^A$ is any non-degenerate observable diagonal in the $\ket{z}$ basis. On the other hand, a derivation by DiVincenzo \emph{et al.}~\cite{divincenzo_locking_2004} using the Maassen-Uffink uncertainty relation Equation~(\ref{eq:maassen}) shows that the optimal measurement $\mathcal{M}_Z^B$ is such that $H(Z^A|\mathcal{M}_Z^B)=\frac 12$. The optimal measurement can also be found by exploiting the group covariance of Bob's states and appealing to a theorem of Davies~\cite{davies_information_1978}. Nonetheless, fulfilling \emph{both} entropy conditions evidently circumvents this issue, as the necessary measurements are defined by Corollaries 2 and 3.

\section{Secret Keys \& Private States}
\label{sec:secretkey}
With a very slight modification, we can extend the results above to give necessary and sufficient conditions on the ability to extract a secret key instead of an EPR pair from the state $\psi^{AB}$. In Section~\ref{sec:qkd} we discussed the fact that EPR pairs can be used to create secret keys, but entanglement of this form is not actually necessary, a fact first observed by Aschauer and Briegel~\cite{aschauer_security_2002}. Instead, bipartite quantum states which are capable of producing secret keys are called \emph{private states} and their general form was established by Horodecki \emph{et al.}~\cite{horodecki_secure_2005}. In this section we show that just like entanglement, knowledge of complementary observables plays a decisive role in characterizing private states.  

Private states have two defining features, as alluded to prior to Theorem~\ref{thm:dw}. First, the key measurements by Alice and Bob clearly must produce identical results. Second, the key should be completely random and uncorrelated with any third party, i.e.\ a would-be eavesdropper Eve. Without loss of generality we can assume that Alice and Bob have two systems each, $A$, $A'$ and $B$, $B'$, respectively, and the key bit is generated by amplitude measurements of $A$ and $B$. If they start with any other state having only systems $A'$ and $B'$, they can coherently perform the key generation measurements and store the result in the amplitude of systems $A$ and $B$, respectively.  Horodecki \emph{et al.}~\cite{horodecki_secure_2005} give the following characterization of \emph{ideal} private states, whose proof we include here for completeness. 
\begin{theorem}[Horodecki \emph{et al.}~\cite{horodecki_secure_2005}]
$\psi^{AA'BB'}$ is a private state iff there exists a \emph{twisting operator} $U^{AA'B'}$ of the form $U^{AA'B'}=\sum_{z=0}^1\ket{z}\bra{z}^{A}\otimes V_{z}^{A'B'}$ with $V_{z}$ unitary such that for some $\xi^{A'B'}$,
\begin{align}
\psi^{AA'BB'}=U^{AA'B'}\left(\Phi^{AB}\otimes \xi^{A'B'}\right)U^{\dagger AA'B'}.
\end{align}
\end{theorem}
\begin{proof}
Consider a purification of a private state. By the first requirement, it must have the form 
\begin{align}
\ket{\psi}^{AA'BB'E}=\tfrac{1}{\sqrt{2}}\sum_{z=0}^1\ket{z,z}^{AB}\ket{\varphi_z}^{A'B'E}.
\end{align}
The second requirement implies that the states $\varphi_z^E$ are all identical, so that the key bit $z$ is secret from any eavesdropper. All possible purifications of a state are related by unitaries on the purifying system, meaning $\ket{\varphi_z}^{A'B'E}=V_z^{A'B'}\ket{\varphi_0}^{A'B'E}$ for some unitaries $V_z$. Using these to define the twisting operator and letting $\xi^{A'B'}=\varphi_0^{A'B'}$ completes the proof.
\end{proof}
Thus, private states are ``twisted'' versions of entangled states in which the $A'B'$ system is transformed in some way conditioned on the value of the key. Since the function of the $A'B'$ system is to block correlations of the key with $E$, it is called the \emph{shield}. Here we have defined the twisting operator as conditioning on Alice's key system $A$, but since her key is always equal to Bob's, the twisting operator can just as well be conditioned on $B$. Private states are conceptually distinct from entangled states because the distributed nature of the $A'B'$ system prevents Alice and Bob from undoing the twisting operator on a general private state. Indeed, there exist private states from which no entanglement can be locally extracted~\cite{horodecki_secure_2005}.

As with entanglement, we are more interested in characterizations of approximate secret keys, since perfection will be impossible to achieve in practice. The following lemma shows that the above definition of secret keys can be extended to a sensible approximate version. Here we denote by $\psi^{Z^AZ^BE}$ the state $\psi^{ABE}$ after measurement of the observables $Z^A$ and $Z^B$, and we say that an approximate secret key is $\epsilon$-good when its trace distance to a perfect key is less than $\epsilon$. 
\begin{lemma}
If $p_{\rm guess}(Z^A|Z^B)_\psi\geq 1-\epsilon_1$ and $p_{\rm secure}(Z^A|E)_\psi\geq 1-\epsilon_2$, then ${\psi}^{Z^AZ^BE}_Z$ is an $(\epsilon_1+\epsilon_2)$-good secret key.
\end{lemma}
\begin{proof}
Start with $p_{\rm guess}(Z^A|Z^B)_\psi\geq 1-\epsilon_1$. By the triangle inequality we have
\begin{align}
\td{\sum_{z,z'} p_{zz'}P_z^A\otimes P_{z'}^B\otimes \varphi_{zz'}^E}{\sum_{z,z'} p_{zz'}P_z^A\otimes P_{z}^B\otimes \varphi_{zz'}^E}&=\tfrac12 \sum_{z,z'\neq z}p_{zz'}+\tfrac12\sum_z\big\|\sum_{z'\neq z}p_{zz'}\varphi_{zz'}^E\big\|_1\\
&\leq \sum_{z,z'\neq z}p_{zz'}\\
&\leq \epsilon_1.
\end{align}
But the state $\sum_{z,z'} p_{zz'}P_z^A\otimes P_{z}^B\otimes \varphi_{zz'}^E$ can be thought of as $U_{\textsc{cnot}}^{AB}\left(\sum_{z,z'} p_{zz'}P_z^A\otimes P_{0}^B\otimes \varphi_{zz'}^E\right)U^{\dagger AB}_{\textsc{cnot}}$. From the second condition it follows, for $\varphi^E=\sum_{zz'}\varphi_{zz'}^E$, that
\begin{align}
\td{\sum_{z,z'} p_{zz'}P_z^A\otimes P_{0}^B\otimes \varphi_{zz'}^E}{\tfrac12\mathbbm{1}^A\otimes P_0^B\otimes \varphi^E}\leq \epsilon_1,
\end{align}
since the presence of $P_0^B$ doesn't change the trace distance. Using unitary invariance of the trace distance and the triangle inequality once more completes the proof.
\end{proof}
To give an approximate characterization of private states based on knowledge of complementary information, we merely need to show that a converse of Corollaries~\ref{cor:secguessB} and~\ref{cor:secguessR} holds, namely that if Bob can accurately guess the amplitude of Alice's system, then the phase is unpredictable using the purification $E$. We formalize this in the following lemma, which is Theorem 4.1 of~\citeme{renes_duality_2010}.
\begin{lemma}
\label{lem:guesssec}
If there exists a measurement $\mathcal{M}_Z^{B}$ such that $p_{\rm guess}(Z^A|\mathcal{M}_Z^{B})_\psi\geq 1-\epsilon_2$ for pure state $\ket{\psi}^{ABE}$, then $p_{\rm secure}(X^A|E)_\psi\geq 1-\sqrt{2\epsilon}$. 
\end{lemma}
\begin{proof}
Following the proof of Theorem~\ref{thm:entdec}, we know that $\bra{\psi_Z}U^{B\rightarrow C_ZB}_{\mathcal{M}_Z}\ket{\psi}^{ABE}\geq 1-\epsilon$. Using the form of $\ket{\psi_Z}^{AC_ZBE}$ in Equation~(\ref{eq:psizphase}), it follows that $E$ is completely decoupled from the phase measurement of $A$, i.e.\ the post-measurement state is $\tfrac12\mathbbm{1}^A\otimes \psi^E$. Since $U^{B\rightarrow C_ZB}_{\mathcal{M}_Z}$ does not involve $A$ or $E$, the post-measurement state of $AE$ is the same for $\ket{\psi}^{ABE}$ as for $U^{B\rightarrow C_ZB}_{\mathcal{M}_Z}\ket{\psi}^{ABE}$. Converting fidelity to trace distance, we find that $\td{\sum_x q_x \widetilde{P}_x^A\otimes \vartheta_x^E}{\tfrac12\mathbbm{1}^A\otimes \psi^E}\leq \sqrt{2\epsilon}$.
\end{proof}

\def\vgap{.75}
\def\hgap{2}
\def\hgapfudge{.25}
\def\re{7}
\begin{figure}[th!]
\begin{center}
\begin{tikzpicture}[thick]
\tikzstyle{empty} = [inner sep=1pt,outer sep=1pt]
\tikzstyle{gate} = [fill=white, draw]
\tikzstyle{ctrl} = [fill,shape=circle,minimum size=4pt,inner sep=0pt,outer sep=0pt]
\tikzstyle{targ} = [draw,shape=circle,minimum size=8pt,inner sep=0pt,outer sep=0pt]

\draw[dotted] (-4,0) -- (\re+3,0);
\node at (-3,.5) {Alice};
\node at (-3,-.5) {Bob};

\node[anchor=east] at (0,\vgap+1) (al) {$A$};
\node[anchor=east] at (0,1) (apl) {$A'$};
\node[anchor=east] at (0,-1) (bpl) {$B'$};
\node[anchor=east] at (1,-\vgap-1) (cl) {$\ket{0}^{C_Z}$};
\node[anchor=east] at (1,-2*\vgap-1) (dl) {$\ket{0}^{C_X}$};
\node[anchor=east] at (0,-3*\vgap-1) (bl) {$B$};

\node[empty,anchor=west] at (\re+2,\vgap+1) (ar) {};
\node[empty,anchor=west] at (\re,1) (apr) {};
\node[empty,anchor=west] at (\re,-1) (bpr) {};
\node[empty,anchor=west] at (\re,-\vgap-1) (cr) {};
\node[empty,anchor=west] at (\re,-2*\vgap-1) (dr) {};
\node[empty,anchor=west] at (\re+2,-3*\vgap-1) (br) {};

\draw (al) -- (ar);
\draw (apl) -- (apr);
\draw (bl) -- (br);
\draw (bpl) -- (bpr);
\draw (cl) -- (cr);
\draw (dl) -- (dr);

\node[ctrl] (ug) at (2,-3*\vgap-1) {};
\node[targ] (ut) at (2,-\vgap-1) {};
\draw (ut.north) -- (ug);

\node[targ] (mt) at (2+\hgap,-2*\vgap-1) {};
\node[gate,minimum height=3.25cm] (mm) at (2+\hgap,-0.35) {$\mathcal{M}_X$};
\draw (mm.south) -- (mt.south);
\node[empty] (mn) at (2+\hgap,-3*\vgap-1) {};

\node[targ] (vt) at (2+2*\hgap,-2*\vgap-1) {};
\node[ctrl] (vc) at (2+2*\hgap,-3*\vgap-1) {};
\draw (vc) -- (vt.north);

\draw[decorate,decoration={brace,amplitude=4},thick] (bl.south west) to node[midway,left,inner sep=5pt] {$\psi^{AA'BB'}$} (al.north west);
\draw[decorate,decoration={brace,amplitude=3},thick] (apr.north east) to node[midway,right,inner sep=5pt] {$\psi^{C_XA'C_ZB'}$} (dr.south east);
\draw[decorate,decoration={brace,amplitude=3},thick] (ar.north east) to node[midway,right,inner sep=5pt] {$\Phi^{AB}$} (br.south east);
\node [below of=vc,anchor=south] (lb1) {$U_{\textsc{cnot}}^{BC_X}$};
\node [below of=ug,anchor=south] (lb2) {$U_{\textsc{cnot}}^{BC_Z}$};
\node [below of=mn,anchor=south] (lb3) {$U_{\mathcal{M}_X}^{C_XA'C_ZB'}$};

\begin{pgfonlayer}{background}
\node [fill=tud1b!50,rounded corners, dashed,fit=(ug) (ut) (lb2)] {};
\node [fill=tud9b!50,rounded corners,fit=(mt) (mm) (lb3)] {};
\node [fill=tud6a!50,rounded corners,fit=(vc) (vt) (lb1)] {};
\end{pgfonlayer}

\end{tikzpicture}
\caption{\label{fig:twistingop} The quantum circuit implementing the (un)twisting operator on the state $\psi^{AA'BB'}$, when Bob can approximately predict Alice's key (amplitude) measurement of system $A$ and there exists a measurement $\mathcal{M}_X^{A'BB'}$ approximately predicting her phase measurement. It proceeds in three steps. First, Bob coherently copies his key (amplitude) to an auxiliary system $C_Z$ using a controlled-\textsc{not} gate (unitary $U_{\textsc{cnot}}^{BC_Z}$). Next, he coherently performs the measurement $\mathcal{M}_X$ allowing him to predict $X$, storing the result in auxiliary system $C_X$ (unitary $U_{\mathcal{M}_X}^{C_XA'C_ZB'}$). Finally, to recover a maximally entangled state in system $B$, he applies another controlled-\textsc{not} gate, with control $B$ and target $C_X$ (unitary $U_{\textsc{cnot}}^{BC_X}$). Observe that the overall action is a controlled operation with Bob's key as the control and the shield and auxiliary systems the target, i.e.\ a twisting operator.}
\end{center}
\vspace{-14pt}
\end{figure}
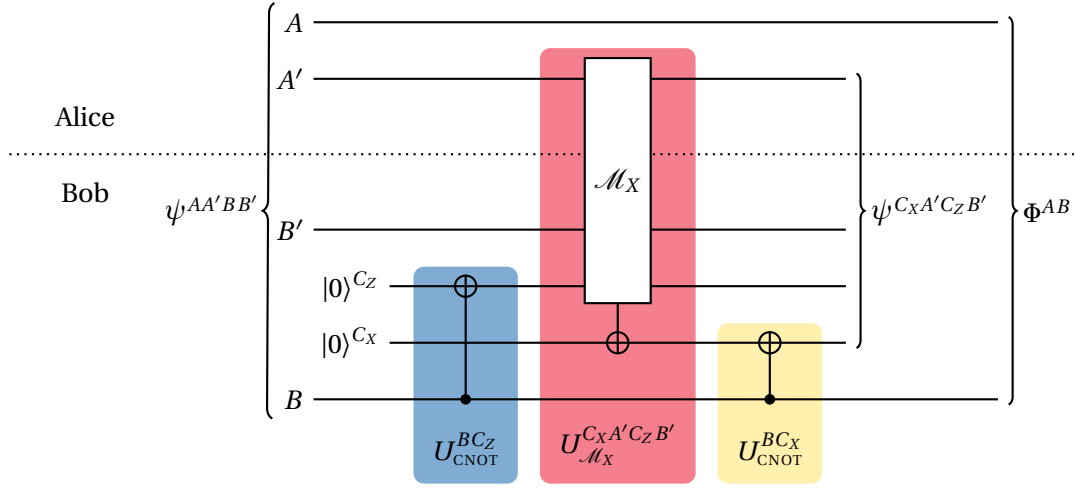

With this lemma the following theorem, first shown in~\citeme{renes_physical_2008}, is immediate.
\begin{theorem}
\label{thm:compps}
Suppose $p_{\rm guess}(Z^A|Z^B)_\psi\geq 1-\epsilon_1$ and there exists a measurement $\mathcal{M}_X^{A'BB'}$ for which $p_{\rm guess}(X^A|\mathcal{M}_X^{A'BB'})_\psi \geq 1-\epsilon_2$. Then ${\psi}^{Z^AZ^BE}$ is an $(\epsilon_1+\sqrt{2\epsilon_2})$-good secret key.
\end{theorem}
It is also interesting to see how the untwisting operator can be directly constructed using the measurement $\mathcal{M}_X^{A'BB'}$. First write the initial state as $\ket{\psi}^{AA'BB'E}=\sum_{z,z'}\sqrt{p_{z,z'}}\ket{z}^A\ket{z'}^B\ket{\varphi_{z,z'}}^{A'B'E}$ consider the action of a \textsc{cnot} operation from $B$ to an ancilla system $C_Z$ prepared in the state $\ket{0}^{C_Z}$. This copies the value of $z'$ and gives 
\begin{align}
U_{\textsc{cnot}}^{BC_Z}\ket{\psi}^{AA'BB'E}=\sum_{z,z'}\sqrt{p_{z,z'}}\ket{z}^A\ket{z'}^{C_Z}\ket{z'}^B\ket{\varphi_{z,z'}}^{A'B'E}.
\end{align}
From the first condition it follows that $\bra{\psi_Z}U_{\textsc{cnot}}^{BC_Z}\ket{\psi}^{AA'BB'E}\geq 1-\epsilon_1$, where 
\begin{align}
\label{eq:needlabel}
\ket{\psi_Z}^{AC_ZA'BB'E}=\sum_{z,z'}\sqrt{p_{z,z'}}\ket{z}^A\ket{z}^{C_Z}\ket{z'}^B\ket{\varphi_{z,z'}}^{A'B'E}.
\end{align}
Now make the replacement $A'BB'\rightarrow B$ in this state and apply the latter half of the proof of Theorem~\ref{thm:entdec}, from which it follows that 
\begin{align}
^{AC_Z\!\!}\bra{\Phi}\,^{C_XA'BB'E\!\!}\bra{\psi}U_{\textsc{cnot}}^{C_ZC_X}V^{A'BB'\rightarrow C_XA'BB'}\ket{\psi_Z}^{AC_ZA'BB'E}\geq 1-\epsilon_2.
\end{align}
The fidelity is unchanged by inserting the identity operator in the form $U_{\textsc{swap}}^{BC_Z}U_{\textsc{swap}}^{\dagger BC_Z}$, yielding
\begin{align}
^{AB\!\!}\bra{\Phi}\,^{C_XA'C_ZB'E\!\!}\bra{\psi}U_{\textsc{cnot}}^{BC_X}V^{A'C_ZB'\rightarrow C_XA'C_ZB'}U_{\textsc{swap}}^{BC_Z}\ket{\psi_Z}^{AC_ZA'BB'E}\geq 1-\epsilon_2.
\end{align}
Since $U_{\textsc{swap}}U_{\textsc{cnot}}^{BC_Z}\ket{\psi}^{AA'BB'E}=U_{\textsc{cnot}}^{BC_Z}\ket{\psi}^{AA'BB'E}$, the same method applied to the first fidelity condition gives
\begin{align}
\bra{\psi_Z}U_{\textsc{swap}}^{\dagger BC_Z}U_{\textsc{cnot}}^{BC_Z}\ket{\psi}^{AA'BB'E}\geq 1-\epsilon_1.
\end{align}
Lemma~\ref{lem:entdec} now implies that the operator $U_{\textsc{cnot}}^{BC_X}V^{A'C_ZB'\rightarrow C_XA'C_ZB'}U_{\textsc{cnot}}^{BC_Z}$ produces a high-fidelity entangled state in systems $AB$. But owing to its form, this is a twisting operator, as depicted in Figure~\ref{fig:twistingop}.

\chapter{Processing Quantum Information}
\label{chap:proc}
Having concretely developed the relationship between quantum information in the form of entanglement and classical information about complementary amplitude and phase observables in Chapter~\ref{chap:char}, we may now apply it to the problem of constructing various quantum information processing protocols and understanding why they work. Being able to do so is the second stated goal of this thesis, and will be the subject of this and the remaining chapters. This chapter considers the particular tasks of entanglement distillation, quantum state merging, and secret key distillation in three respective sections. The complementarity approach to the first and last was developed in detail in~\citeme{renes_physical_2008}, while state merging was treated from this approach in~\citeme{boileau_optimal_2009}.

Entanglement distillation is one of the fundamental protocols in quantum information processing and can be used as a building block in a variety of other protocols. In particular, one-way protocols for entanglement distillation can be repurposed for use in reliable communication of quantum information over noisy channels. This allows us to apply our results to that problem and show that the quantum capacity of a channel can be achieved when the sender uses CSS codes. Meanwhile, the secret key distillation results imply that the capacity of a quantum channel to send classical information privately can likewise be achieved when the sender uses CSS codes. 




\section{Optimal Entanglement Distillation}
\label{sec:opted}
We begin by returning to the problem of entanglement distillation, introduced in Section~\ref{sec:ed1}. In this setting, Alice and Bob share a supply of identical, somewhat-entangled bipartite resource states which they would like to use to create EPR pairs. An entanglement distillation protocol is a sequence of local operations they should perform on their respective systems, supplemented by classical communication to coordinate their actions and exchange information. The protocol produces approximate EPR pairs at a given rate $r$, converting $n$ resource states to $nr$ pairs. For instance, the rate of the protocol described in Chapter~\ref{chap:ill} using the Shor 9-qubit code is given by the rate of the error-correcting code, namely $1/9$, since the output was taken from the encoded subspace of the code. The asymptotically optimal rate is the largest $r$ one can find among protocols for $n\rightarrow \infty$ such that the approximation parameter vanishes in this limit. 

As we saw in Chapter~\ref{chap:char}, Alice and Bob implicitly share an entangled state if Alice's amplitude and phase measurements are predictable by Bob. But a generic bipartite state does not share this property; at best Bob has only partial information about either observable. Heuristically, one way to manufacture entangled states would therefore be to increase Bob's information about these measurements somehow. And since such information is classical, we may be able to arrange for Alice to send it over the classical communication channel. In the following we shall develop this heuristic notion into a concrete protocol. To do so we must first overcome two immediate hurdles. First, what sort of information can she send which will be sufficient for this purpose? And second, how do we make sure Alice does not violate the uncertainty principle when sending information about complementary observables? We take up these two questions in turn.

\subsection{Information Reconciliation}
\label{subsec:ir}

If we consider either observable alone, the present task is a more general version of the information reconciliation task mentioned in conjunction with QKD in Section~\ref{sec:qkd}. If we only care about, say, amplitude, then we can imagine Alice measures the amplitude of all of her systems, and these outcomes are described as a classical random variable. Formally, we can describe the state of all their systems after the measurement by 
\begin{align}
\label{eq:jointstate}
\Psi^{Z^AB}=\sum_\bz p_\bz \ket{\bz}\bra{\bz}^A\otimes \varphi_\bz^B=\sum_\bz P_\bz^A\otimes {\rm Tr}_A[P_\bz^A(\psi^{AB})^{\otimes n}]
\end{align}
using the state $\ket{\psi}^{ABE}=\sum_{z}p_z\ket{z}^A\ket{\varphi_z}^{BE}$ and defining $p_{\bz}=p_{z_1}\cdots p_{z_n}$ and $\varphi_\bz=\varphi_{z_1}\otimes \cdots\otimes \varphi_{z_n}$. We  use boldface to denote sequences or strings of indices. 
For each sequence of outcomes Bob is left with the quantum state $\varphi_\bz^B$, but generally there is no measurement which will indicate which one he has with any accuracy. 

However, if Alice gives him some extra information about her outcome $\bz$, then the set of states he is attempting to distinguish between gets smaller, and the task gets easier. For instance, if Alice simply tells him that the sum of the first two outcomes (thought of as binary outcomes) is 0 modulo 2, then he excludes from consideration all the $\varphi_\bz^B$ for which this is not true and attempts to distinguish between the remaining states with a new measurement. Of course, she could just send Bob her entire measurement record $\bz$, but the goal of information reconciliation is for Alice to transmit as few bits as necessary to enable Bob to reconstruct $\bz$ with high probability.

It turns out that in the asymptotic limit $n\rightarrow\infty$, Alice only needs to send information at rate $H(Z^A|B)_\psi$. This expression accords with the interpretation of conditional entropy as the uncertainty about $Z^A$ given $B$: Bob is missing this much information about $Z^A$ and in the protocol Alice simply provides it. Importantly, the information in question can be generated by the technique of \emph{universal hashing} and Alice does not need to know anything about Bob's system except the value of $H(Z^A|B)_\psi$. In universal hashing, Alice randomly picks a so-called hash function $f$ from a universal family of hash functions and sends Bob a description of $f$ along with the output $f(\bz)$. 

First defined by Carter and Wegman~\cite{carter_universal_1977,carter_universal_1979}, universal hashing is meant to mimic certain behavior of random functions: A family of functions is universal when the probability that two different inputs to a randomly-chosen family member have the same output is the same as if the function had been chosen at random from all possible functions. This latter probability is simply the inverse of the number of possible function outputs, so formally we say a set $\mathcal{F}$ of functions $f:\{0,1\}^n\rightarrow \{0,1\}^{m}$ is universal when
\begin{align}
{\rm Pr}_f\left[f(x)=f(y)\right]\leq \frac{1}{2^{m}}\qquad \forall x,y\in\{0,1\}^n.
\end{align}
Above we illustrated the information Alice might send to Bob by a linear function, and in fact the set of all linear functions forms a universal family~\cite{carter_universal_1977,carter_universal_1979}. We shall make extensive use of linear functions for hashing in the next section. 

As shown in~\citeme{renes_physical_2008}, when the size of the hash is roughly $nH(Z^A|B)_\psi$ bits, Bob can reliably predict $Z^A$. More concretely, for each hash value $\widehat{\bz}=f(\bz)$ there exists a measurement $\mathcal{M}_{Z;\widehat{\bz}}^{B}$ with elements $\Lambda^{B}_{\bz;\widehat{\bz}}$ such that the guessing probability averaged over $\widehat{\bz}$ is nearly one, $\sum_\bz p_\bz {\rm Tr}[ \Lambda^{B}_{\bz;\widehat{\bz}}\varphi_\bz^B]\approx 1$.
 The proof, following ideas from Holevo~\cite{holevo_capacity_1998} and Schumacher and Westmoreland~\cite{schumacher_sending_1997} in the study of transmission of classical information over quantum channels, explicitly constructs $\mathcal{M}_Z^{BB'}$ as a variant of the \emph{pretty-good measurement} first used by Holevo for pure states~\cite{holevo_asymptotically_1978} and later extended to mixed states (and so-named) by Hausladen and Wootters~\cite{hausladen_pretty_1994}. Essentially, Bob's measurement is given by 
\begin{align}
\Lambda_{\bz;\widehat{\bz}}\approx p_\bz \varphi_{\widehat{\bz}}^{-1/2}\varphi_\bz\varphi_{\widehat{\bz}}^{-1/2},\qquad \varphi_{\widehat{\bz}}=\sum_{\bz:f(\bz)=\widehat{\bz}}p_\bz\varphi_\bz,
\end{align}
with some small modifications. 
We can simplify the formalism somewhat by imagining that Bob stores the hash value in an auxiliary system $B'$ and uses the measurement $\mathcal{M}_{Z}^{BB'}$ with elements $\Gamma_\bz^{BB'}=\Lambda_{\bz;\widehat{\bz}}^{B}\otimes P_{\widehat{\bz}}^{B'}$, for which $p_{\rm guess}(Z^A|\mathcal{M}_Z^{BB'})\approx 1$.

In the case Bob holds classical information, i.e.\ the states $\varphi_\bz$ are all simultaneously diagonalizable, information reconciliation is closely related to the famous Slepian-Wolf problem of coding of correlated sources~\cite{slepian_noiseless_1973}. The case of Bob having quantum information was studied and solved in the present i.i.d.\ scenario by Winter~\cite{winter_coding_1999} and Devetak and Winter~\cite{devetak_classical_2003} using random coding techniques. In Section~\ref{subsec:oneshot} we very briefly describe how the result can be generalized to the case of arbitrary resources. 

\subsection{Reconciling Complementary Information}
\label{sec:reconcile}
Having seen that the output of a suitably-sized random hash function enables Bob to reconstruct the outcome of Alice's amplitude or phase measurement, we now turn to the problem of how Alice can generate \emph{both} pieces of information without violating the uncertainty principle. Calling the hash function used for the amplitude measurement $f$ and the phase measurement $g$, Bob separately requires both $f(\bz)$ and $g(\bx)$ so that he can predict the amplitude outcome $\bz$ and the phase outcome $\bx$. Naively, it seems impossible to generate both $f(\bz)$ and $g(\bx)$, since this would apparently require Alice to measure both the amplitude and phase of her systems. 

Crucially, however, the \emph{input} $\bx$ (\bz) is not required to fix the \emph{output} $g(\bx)$ ($f(\bz)$). Instead, Alice need only measure appropriate observables which generate the output directly, \emph{and the necessary observables for $f(\bz)$ and $g(\bx)$ can commute}. Such a structure is in fact provided by CSS codes. Recall again the very simple example above, in which Alice transmitted the output of the linear function $z_1\oplus z_2$ to Bob. As we saw in Section~\ref{subsec:compqecc}, this can equally-well be thought of as the outcome of measuring the operator $Z_1Z_2$, since its eigenvalues are $(-1)^{z_1\oplus z_2}$. But every linear function is a sequence of one-bit functions, each of which is just a sum of particular amplitude outcomes $z_k$, so to each linear function corresponds a sequence of products of amplitude operators. In other words, every linear function of the amplitude measurement outcome is associated with a collection of $Z$-type stabilizers, and similarly for $X$-type stabilizers and functions of the phase measurement outcomes. 

If the two functions $f$ and $g$ are chosen so that the corresponding $Z$- and $X$-type stabilizers commute, together they define a CSS code, and Alice can then generate both $f(\bz)$ and $g(\bx)$ by measuring the stabilizers of the code. The commutation condition on the stabilizers can be succinctly expressed in the following way. For an $n$-qubit stabilizer, the corresponding linear function can be specified by the $n$-dimensional binary $\mathbbm{F}_2$ vector with entries 1 at position $k$ if $z_k$ appears in the sum, and zero otherwise. For instance, in the 9-qubit Shor code, the stabilizer $Z_1Z_2$ corresponds to the vector $(1,1,0,0,0,0,0,0,0)$ while the stabilizer $X_1X_2X_3X_7X_8X_9$ corresponds to $(1,1,1,0,0,0,1,1,1)$. In this representation, two stabilizers commute if the corresponding vectors are orthogonal over $\mathbbm{F}_2$. 

The only requirement on the functions $f$ and $g$ is that they come from universal families $\mathcal{F}$ and $\mathcal{G}$ of hash functions, respectively. Suppose $n_Z$ and $n_X$ are the required number of amplitude and phase type stabilizers, respectively, as determined by the rate requirements of the respective information reconciliation tasks. Then it is easy to show that one simple universal family encompassing both hash functions is the set of $(n_Z+n_X)\times n$ matrices over $\mathbbm{F}_2$ consisting of pairwise orthogonal rows. The first $n_Z$ rows give the $Z$-type stabilizers and the remaining $n_X$ rows the $X$-type stabilizers. 

Given these stabilizers, it is convenient to think of the code as partitioning Alice's qubits in system $A$ into three different sets of virtual qubits, the encoded qubits in subsystem $\overline{A}$, the $n_Z$ qubits whose amplitude measurement gives $f(\bz)$ in $\widehat{A}$, and the $n_X$ qubits whose phase measurement gives $g(\bx)$ in $\widetilde{A}$. Then by the properties of the stabilizer operators, $\widehat{\bz}=f(\bz)$ and $\widetilde{\bx}=g(\bx)$, where $\widehat{\bz}$ denotes a particular sequence of amplitude measurement outcomes for system $\widehat{A}$.   

Now we have all the pieces needed to construct an entanglement distillation protocol. Starting from $n$ copies of the resource state, Alice will measure $n_Z$ $Z$-type stabilizers and $n_X$ $X$-type stabilizers and communicate the resulting syndromes to Bob. Then, for large enough $n$, he will be able to predict Alice's measurement of encoded amplitude and phase operators using the appropriate pretty-good measurements, and thus create an approximate EPR state following Theorem~\ref{thm:entdec}.

Formally, they begin with the state
\begin{align}
\ket{\Psi}^{ABE}&=\sum_\bz \sqrt{p_\bz}\ket{\bz}^A\ket{\varphi_\bz}^{BE}\\
&=\sum_{\overline{\bz},\widehat{\bz},\widetilde{\bz}}\sqrt{p_{\bz}}\ket{\overline{\bz}}^{\overline{A}}\ket{\widehat{\bz}}^{\widehat{A}}\ket{\widetilde{\bz}}^{\widetilde{A}}\ket{\varphi_{\bz}}^{BE},
\end{align}
where in the second line we use the decomposition of Alice's system into the three sets of virtual qubits and consider $\bz$ to be a function of $(\overline{\bz},\widehat{\bz},\widetilde{\bz})$. The number $n_Z$ is chosen so that $p_{\rm guess}(Z^A|\mathcal{M}_Z^{BB'})_\psi\approx 1$, where again $B'$ is the system in which Bob stores $\widehat{\bz}$. Since $\overline{\bz}$ is a (linear) function of $\bz$, this implies that $p_{\rm guess}(\overline{Z}|\mathcal{M}^{BB'}_Z)\approx 1$. As much is true for the phase in that given the value of $\widetilde{\bx}$ stored in $B''$, there exists a measurement $\mathcal{M}_X^{BB''}$ for which $p_{\rm guess}(\overline{X}|\mathcal{M}^{BB''}_Z)_\psi\approx 1$. Therefore Bob can recover approximate EPR pairs by performing these measurements coherently, as shown in Theorem~\ref{thm:entdec}. In this way they can distill $n-n_X-n_Z$ approximate EPR pairs, provided this quantity is positive.
Note that here we have only utilized communication from Alice to Bob, making this a \emph{one-way} protocol. Using back and forth communication Alice and Bob could in principle increase the distillation rate, as pointed out by Bennett \emph{et al.}~\cite{bennett_mixed-state_1996} for protocols where \emph{both} parties use quantum error correction, as described in Section~\ref{sec:ed1}. 

\subsection{Constructing an Optimal Protocol}
\label{subsec:optimal}
The final question is how small $n_Z$ and $n_X$ can be made, and here there arises an additional subtlety. From the above discussion, we would expect that $n_Z\approx nH(Z^A|B)_\psi$ and $n_X\approx nH(X^A|B)_\psi$. However, Alice and Bob can do better. Initially, the purification of their shared state is 
\begin{align}
\label{eq:edstart}
\ket{\Psi}^{ABE}=\sum_\bz\sqrt{p_\bz}\ket{\bz}^A\ket{\varphi_\bz}^{BE}.
\end{align}
After receiving the amplitude information, Bob has full information about \bz, which he could store in system $C_Z$. Then, for the purposes of predicting Alice's hypothetical phase measurement, it is as if they originally shared (a close approximation to) the following state,
\begin{align}
\ket{\Psi_Z}^{AC_ZBE}= \sum_\bz\sqrt{p_\bz}\ket{\bz}^A\ket{\bz}^{C_Z}\ket{\varphi_\bz}^{BE},
\end{align}
and this may simplify Bob's phase-prediction task in general. 

One might worry that Alice's phase measurement is no longer possible even hypothetically, due to the amplitude stabilizer measurement. However, Bob can still use the conditional marginal states $\vartheta_\bx^{C_ZB}$ for $\sqrt{q_\bx}\ket{\vartheta_\bx}^{C_ZBE}=\frac{1}{\sqrt{2^n}}\sum_\bz\sqrt{p_\bz}(-1)^{\bx\cdot\bz}\ket{\bz}^{C_Z}\ket{\varphi_\bz}^{BE}=\frac{1}{\sqrt{2^n}}(Z^\bx)^{C_Z}\ket{\Psi}^{C_ZBE}$ to build the unitary operator $U_{\mathcal{M}_X}$. This gives him what would have been the phase measurement outcome, and therefore the outcome of the encoded phase measurement. The existence of the former has indeed been destroyed by the amplitude stabilizer measurement, but the latter has not. 

A concrete example in which amplitude information is relevant to phase is provided by the following. Suppose the state Alice and Bob share is a maximally-entangled state $\Phi^{AB}$ afflicted only with errors of the form $XZ$. Then amplitude $Z$ errors are completely correlated with  phase $X$ errors. Thus, Bob need only know the positions of amplitude errors in order to infer the positions of phase errors. In other words, $H(X^A|BC_Z)_\psi=0$. 

Taking the above consideration into account, the rate of entanglement distillation becomes $r(\psi)=1-H(Z^A|B)_\psi-H(X^A|BC)_{\psi_Z}$, as all the approximation parameters can be taken to arbitrarily small values by choosing a large enough $n$. 
It turns out that $r(\psi)=-H(A|B)_\psi$, which we can see by direct computation. First evaluate the latter entropy $H(X^A|C_ZB)_{\psi_Z}$, using the form of $\ket{\psi_Z}$ derived in (\ref{eq:psizphase}). We find 
\begin{align}
H(X^A|CZ_B)_{\psi_Z}&\equiv H(X^AC_ZB)_{\psi_Z}-H(C_ZB)_{\psi_Z}\\
&=H(X^A)_{\psi_Z}+H(C_ZB|X^A)_{\psi_Z}-H(C_ZB)_{\psi_Z}\\
&=1-H(C_ZB)_\psi-H(C_ZB)_{\psi_Z}\\
&=1-H(E)_\psi-H(AE)_{\psi_Z}\\
&=1-H(E)_\psi-H(Z^AE)_{\psi}.
\label{eq:xtoz}
\end{align}
The first step follows from the general relation between conditional and unconditional von Neumann entropies, while the second follows because the state of $C_ZBE$ conditioned on outcome $X^A=x$ is $(Z^x)^{C_Z}\ket{\psi}^{C_ZBE}$. As these are all unitarily equivalent, each term $H(C_ZB|X^A=x)_{\psi_Z}$ has the same value $H(C_ZB)_\psi$. In the third step we have used the fact that $H(S_1)=H(S_2)$ for a bipartite pure state on systems $S_1$ and $S_2$. The last step follows because $\psi_{Z}^{AE}$ is identical to the result of measuring the amplitude of $A$ for the initial state $\psi^{AE}$. Hence $r(\psi)=H(Z^A|E)_\psi-H(Z^A|B)_\psi$. But since the $BE$ system given the measurement outcome $Z^A=z$ is pure, $H(B|Z^A=z)_\psi=H(E|Z^A=z)_\psi$. Therefore $H(Z^AE)_\psi=H(Z^AB)_\psi$ and $r(\psi)=-H(A|B)_\psi=H(A|E)_\psi$. This rate is sometimes called the \emph{hashing bound}. 

Two further modifications lead to the optimal entanglement distillation rate. First, Alice is free to first apply any quantum operation $\mathcal{Q}^A$ to her system before the protocol begins, and this increases the rate to 
\begin{align}
\mathsf{D}_1(\psi)=\max_\mathcal{Q}\left(-H(A|B)_{\psi_\mathcal{Q}}\right).
\end{align}
Second, the rate can be further improved by \emph{regularization}.
Although we have described the protocol above for system $A$ a qubit, it works almost precisely the same for any dimension $d$ which is a prime power.\footnote{The restriction to prime powers comes from the structure of the stabilizer operators. These require the vector-representation described in Section~\ref{sec:reconcile} which only exists when the symbols come from a finite field.} 
Given a state $\psi^{AB}$, we could then imagine considering $\Psi^{AB}=(\psi^{AB})^{\otimes m}$ to be the fundamental input to the protocol, and Alice and Bob starting with $n$ copies thereof. The difference is that now Alice and Bob can ignore the product structure of $\Psi^{AB}$
, which leads to the possibly-higher rate
\begin{align}
\mathsf{D}(\psi)=\lim_{n\rightarrow\infty}\frac1n\mathsf{D}_1(\psi^{\otimes n}).
\end{align}
Devetak and Winter show this rate, the \emph{distillable entanglement}, is in fact optimal in~\cite{devetak_distillation_2005}.

\subsection{Quantum Noisy Channel Coding}
\label{subsec:channelcoding}
With a small modification, this entanglement distillation protocol can be used for reliable transmission of quantum information over a noisy channel $\mathcal{N}$. As mentioned in the discussion of approximate error-correction in Section~\ref{sec:main}, we can always mimic the quantum communication task by sending half of an EPR pair through the channel and measuring the half remaining with Alice in the appropriate basis. Thus, by deferring the measurement indefinitely, we only need to consider reliably transmitting halves of EPR states.
 
Now consider the entanglement distillation protocol applied to $\psi^{AB}=[{\rm id}^A\otimes\mathcal{N}^B](\theta^{AB})$, for an arbitrary pure state $\ket{\theta}^{AB}$. The protocol is constructed so that, averaged over all values that the syndromes could take on, the distilled state closely approximates the ideal of $nr$ EPR pairs. Pick the syndrome with the best approximation parameter, which is surely better than the average. Since in the communication scenario Alice can choose the input, she can always do so in a way which ensures her stabilizer measurement yields precisely this syndrome. Therefore, Alice and Bob could agree on the syndrome value in advance. 

But this defines an encoder and decoder in an error-correction scheme! Alice directly creates the bipartite state resulting from measuring the code stabilizers on many instances of $\theta^{AB}$ and obtaining the specified syndrome. She then sends Bob's halves through the channel, and he is able to decode the result by applying the entanglement distillation procedure. Since entanglement can be faithfully transmitted, so could any particular single-system state. 

Applied to single inputs, this implies that reliable quantum communication must be possible over the channel at rate (here we dispense with the operation $\mathcal{Q}$)
\begin{align}
\mathsf{Q}_1(\mathcal{N})&=\max_{\theta}\,\left(-H(A|B)_\psi\right).
\end{align}
Despite its nonstandard appearance, this is equal to a maximization over the \emph{coherent information} $I_c$ introduced by Schumacher and Nielsen~\cite{schumacher_quantum_1996} and more frequently used in this context. 
To see this, write $\ket{\theta}^{AB}=\sum_k \sqrt{p_k}\ket{k}^A\ket{\vartheta_k}^{B}$ for some probabilities $p_k$ and normalized states $\ket{\vartheta_k}^{B}$. The action of the channel on $B$ can be thought of as an isometry $U_\mathcal{N}^{B\rightarrow BE}$ and $\ket{\psi}^{ABE}=\sum_k\sqrt{p_k}\ket{k}^AU_\mathcal{N}^{B\rightarrow BE}\ket{\vartheta_k}^{B}$ is the output. Computing the conditional entropy $H(A|B)_\psi$ we find 
\begin{align}
\mathsf{Q}_1(\mathcal{N})&=\max_\theta\left(H(B)_\psi-H(R)_\psi\right)\\
&=\max_\vartheta\left(H(\mathcal{N}(\vartheta))-H(\mathcal{N}^*(\vartheta))\right)\\
&\equiv\max_\vartheta I_c(\vartheta,\mathcal{N}),
\end{align}
where $\vartheta=\sum_k p_k \vartheta_k^B$ and $\mathcal{N}^*$ is the channel complementary to $\mathcal{N}$ obtained by applying $U_\mathcal{N}^{B\rightarrow BE}$ and keeping $R$ instead of $B$. In the first line $H(AB)_\psi=H(R)_\psi$ since $\psi$ is pure, and maximization over $\theta$ is equivalent to maximization over $\vartheta$ in the second line. 

Regularization could improve the result, and we have therefore we have constructed a noisy-channel coding scheme which achieves a rate $\mathsf{Q}(\mathcal{N})$, where
\begin{align}
\mathsf{Q}(\mathcal{N})=\lim_{n\rightarrow \infty}\frac1n \mathsf{Q}_1(\mathcal{N}^{\otimes n}).
\end{align}
In fact, this is the ultimate capacity of the channel. In a sequence of papers~\cite{schumacher_sending_1996,schumacher_quantum_1996,barnum_information_1998, barnum_quantum_2000}, Barnum, Knill, Nielsen, and Schumacher established $\mathsf{Q}$ as an upper bound on the capacity, while Lloyd~\cite{lloyd_capacity_1997}, Shor~\cite{shor_quantum_2002}, and Devetak~\cite{devetak_private_2005} used random-coding arguments to show that $\mathsf{Q}$ can be attained. 

Here we have shown that CSS codes can achieve the capacity, since the resulting code inherits this structure from Alice's use of CSS-type stabilizers in the entanglement distillation protocol. Previously, CSS codes were only known to achieve a lower rate, as implicitly shown by Shor and Preskill~\cite{shor_simple_2000} and explicitly by Hamada~\cite{hamada_reliability_2004}. The more-general stabilizer codes were shown to achieve the capacity by Hayden \emph{et al.}~\cite{hayden_decoupling_2008}. 

Devetak's coding scheme has some CSS-like properties in that it essentially consists of an amplitude error-correction step followed by a privacy amplification step. From the discussion of the previous chapter, particularly Lemma~\ref{lem:guesssec} but with amplitude and phase trading places, we are tempted to view the latter step as error correction of a phase observable, and indeed we shall examine this in more detail in Section~\ref{subsec:edirpa}, but the amplitude and phase observables implicitly used in~\cite{devetak_private_2005} are functions of the coding scheme itself and not identical to the (code-independent) amplitude and phase as we have used here. 

One appealing aspect of the use of CSS structure is the possibility of constructing efficiently encodable and decodable codes which approach or even achieve the capacity. For classical communication over classical channels, Forney exhibited such a construction by concatenating random codes with structured codes known as polynomial or Reed-Solomon codes~\cite{forney_concatenated_1966}. Hamada has extended this to the quantum case in a sequence of papers~\cite{hamada_conjugate_2006,hamada_efficient_2008,hamada_concatenated_2008}, but only up to the suboptimal rate mentioned above. It would be interesting to see if the methods presented here can be combined with those of Hamada to reach the capacity efficiently.

\section{Optimal State Merging}
Since the unitary Bob eventually uses to distill the entangled states also transfers the state of Alice's system $A$ to his laboratory, the above protocol can be used for state merging, a process first studied by Horodecki \emph{et al.}~\cite{horodecki_partial_2005}. Here the goal is to merge Alice's part $\psi^{AB}$ of the joint state $\psi^{AB}$ with Bob's so that he ends up with $\psi^{AB}$, using as little quantum or classical communication as possible. Additionally, if we consider the purification $\ket{\psi}^{ABE}$, all correlations with the purifying system $R$ should be transferred to Bob as well. Not only should Bob end up with a good approximation to $\psi^{AB}$, but together with $R$ the final state should closely approximate $\ket{\psi}^{ABE}$. 

When $\psi^{AR}$ is itself pure, state merging reduces to quantum data compression. Since Bob has no initial information about Alice's state, whatever she sends must be sufficient to reconstruct her state and can be regarded as the compressed version of it. Schumacher has shown that a state $\psi^{A}$ can be compressed at rate no greater than $H(A)$~\cite{schumacher_quantum_1995}, meaning Alice and Bob will need to use a quantum channel at this rate. 

However, when Bob's system is correlated with Alice's, they can take advantage of these correlations to reduce the amount of communication needed. Indeed, if Alice and Bob share the EPR state $\ket{\Phi}^{AB}$, then no communication is required at all! This follows because a maximally-entangled state is not correlated with any third system, and so Bob can simply recreate the state at his end. For example, applying Theorem~\ref{thm:entdec} to the input state $\ket{\Phi}^{AB}$ yields output $\ket{\Phi}^{AD}\ket{\Phi}^{BC}$ upon application of the partial isometry $U^{B\rightarrow BCD}$. 

Moreover, sometimes sending only classical information is sufficient for transferring a quantum state. This is precisely the case when using the entanglement distillation protocol, which works for all $\psi^{AB}$ such that $H(A|B)<0$. That classical communication is sometimes sufficient is somewhat surprising, but with entanglement Alice could teleport her system to Bob using only classical information, and this is effectively what happens as a byproduct of the entanglement distillation protocol. By expressly using teleportation, we can also apply the distillation protocol to cases when $H(A|B)>0$, as described in~\cite{horodecki_partial_2005}. For $n$ resource states $\psi^{AB}$ Alice and Bob can create $nH(A|B)_\psi$ EPR pairs to go with their $n$ resource states, and the overall conditional entropy of the entire collection of systems is now roughly zero. Running the entanglement distillation protocol produces no new EPR pairs, but does transfer Alice's part of the resource state to Bob. 

Horodecki \emph{et al.} have shown that state merging requires quantum communication at the rate $H(A|B)_\psi$ when this quantity is positive, but only classical communication at rate $I(A:E)_\psi$ when $H(A|B)_\psi$ is negative~\cite{horodecki_partial_2005,horodecki_quantum_2007}. Using the entanglement distillation procedure above is therefore optimal in the first setting but not always in the second, as the rate of classical communication needed is $(n_Z+n_X)/n=1-H(A|E)_\psi\geq I(A:E)_\psi$. 

However, we can make a small alteration to the protocol to make it optimal, as shown in~\citeme{boileau_optimal_2009}. Observe that when $H(A)_\psi=1$, the procedure is in fact optimal. This suggests that we ought to first compress system $A$ and then perform entanglement distillation.
The difficulty in making this work is to ensure that the compression step does not interfere with the amplitude and phase information reconciliation steps. Since compression of quantum systems can be thought of as essentially just classical compression in the eigenbasis, it simplifies matters to choose the amplitude basis to be the eigenbasis of Alice's state $\psi^A$. 

Formally, the tripartite system $ABE$ starts in the pure state given in Equation~(\ref{eq:edstart}). The compressor projects the system onto a subspace spanned by a set of eigenvectors $\ket{\bz}$ whose total probability is nearly equal to one, a so-called \emph{typical set}. Even though the typical set contains almost all of the probability, it only contains roughly $2^{nH(Z)}$ of the $2^n$ total eigenvectors. Thus, with probability nearly one the projection operation succeeds and the subspace needed to support the state drastically shrinks. Rarely, the projection operation fails, and the state must be written off as a total loss. 

When the compressor succeeds, the state can be expressed as 
\begin{align}
\ket{\psi'}^{ABE}=\frac{1}{\mathcal{N}}\sum_{\bz\in {\rm Typ}}\sqrt{p_{\bz}}\ket{\bz}^A\ket{\varphi_{\bz}}^{BE},
\end{align}
where Typ is the typical set and $\mathcal{N}$ is the required normalization factor. On the \emph{typical subspace} we can order the basis elements lexicographically and define a new amplitude observable $Z'$ as in Equation~(\ref{eq:wh}), as well as the phase observable corresponding to the shift operator of said basis. After the compression step, the idea is for Alice and Bob to run the entanglement distillation procedure for the new observables $Z'$ and $X'$. However, the distribution of measurement results for these two operators is no longer i.i.d., and thus the results of information reconciliation we used previously no longer apply. We have no direct way of knowing how many stabilizers Alice should measure, nor how Bob should construct his measurement.

This poses no serious problem for the new amplitude observable, since it is essentially the same as the old one, just missing the non-typical values. Indeed, the information reconciliation protocol also makes use of typicality in that Bob's measurement does not bother to look for non-typical \bz\ in the first place. Thus, explicitly rejecting these possibilities in the compression step will only serve to reduce the error probability for information reconciliation of $Z'$. Alice can perform precisely the same $Z$-type stabilizer measurements as before, and Bob's original measurement will accurately reconstruct $Z$ and therefore $Z'$. 

However, this sort of argument does not work for the new phase observable $X'$. Since $X$ and $X'$ are not so simply related, Bob's knowledge of $X$ generally does not pertain at all to his knowledge of $X'$. Luckily, the extra system $C_Z$ which was used to achieve the optimal entanglement distillation rate comes the rescue. In the entanglement distillation protocol it gave Bob's marginal states conditioned on Alice's phase measurement a group-covariant structure, and it does so in the present scheme as well. In turn, this makes it possible to transform the information reconciliation protocol in the original i.i.d.\ setting to one appropriate for the new non-i.i.d.\ setting. 

After the amplitude information reconciliation step, phase information reconciliation proceeds as if Alice and Bob shared the state
\begin{align}
\ket{\psi'_{Z}}^{ABE}=\frac{1}{\mathcal{N}}\sum_{\bz\in {\rm Typ}}\sqrt{p_{\bz}}\ket{\bz}^A\ket{\bz}^{C}\ket{\varphi_{\bz}}^{BE}.
\end{align}
The group covariance arises just as before, due to the ``copy'' of $\bz$ in system $C_Z$.
The required number of $X'$-stabilizer outcomes must be computed in the construction of the reconciliation protocol, and it turns out to be $n_{X'}=H(A)_{\psi_Z}+H(C_ZB|X^A)_{\psi_Z}-H(C_ZB)_{\psi_Z}$. Following the calculation in the previous section, this is just $n_{X'}=H(A)_\psi-H(Z^A|E)_\psi$. For the first term we have used the fact that $H(A)_{\psi_Z}=H(Z^A)_\psi=H(A)_\psi$ since the amplitude basis is the eigenbasis. The communication cost of the protocol is now $n_Z+n_{X'}=H(A)_\psi+H(Z^A|B)_\psi-H(Z^A|E)_\psi=I(A:E)_\psi$. 

Since they are working in the typical subspace, Alice and Bob can expect to extract roughly 
$n \log|{\rm Typ}|-n_Z-n_{X'}$ entangled pairs, where $|{\rm Typ}|\approx 2^{nH(A)}$ is the size of the typical set. This works out to an entanglement distillation rate of $H(A)_\psi-I(A:E)_\psi=-H(A|B)_\psi$, just as before.  Therefore, by adding a compression step and choosing the amplitude basis to be the eigenbasis of Alice's system, we have managed to convert the optimal entanglement distillation protocol into an optimal state merging protocol. 

\section{Secret Key Distillation and Private Communication}
\label{sec:skd}
Section~\ref{sec:secretkey} detailed the close connection between private and entangled states, and in this section we show that the same methods used in Section~\ref{sec:opted} to construct entanglement distillation protocols can be used to construct protocols for creating a shared secret key from a supply of bipartite quantum states. Due to the CSS nature of this approach, we really only need to construct a private state distillation scheme, and it will work for secret key distillation as well. As explained in Section~\ref{sec:qkd}, Alice and Bob ultimately only need to ensure that the information about Alice's hypothetical phase measurement is somewhere to be found in the systems under their control. 

The private state distillation protocol works almost exactly as the entanglement distillation protocol. Given $n$ copies of the resource state $\psi^{AB}$, Alice is free to decide how to define the prospective key and first performs a quantum operation $\mathcal{Q}^{A\rightarrow A\,A'}$ which maps her system $A$ into two systems $A\,A'$. The first is used as the key and the second as a shield. This operation may additionally involve a measurement whose outcome $T$ is publicly transmitted to Bob, and the resulting state is $\psi^{AA'BT}_\mathcal{Q}=\sum_t \mathcal{Q}_t^{A\rightarrow AA'}(\psi^{AB})\otimes \ket{t}\bra{t}^T$. 

In the second step Alice measures enough amplitude and phase stabilizers on $A$ so that the amplitude $Z^A$ can be reconstructed from system $B$ and the phase $X^A$ from the compound system $A'B$. The number of stabilizers needed is set by the requirements for information reconciliation of each task separately, and again the amplitude information may be useful in recovering the phase information. Therefore the number of stabilizers needed amounts to $n_Z\approx nH(Z^A|BT)_{\psi_{\mathcal{Q}}}$ and $n_X\approx nH(X^A|C_ZA'BT)_{\psi_{\mathcal{Q},Z}}$, where $\psi_{\mathcal{Q},Z}^{AC_ZA'B}$ is the state defined by coherently copying the amplitude in $A$ to system $C_Z$. 
Alice may choose the optimal operation $\mathcal{Q}$, yielding the distillation rate  
\begin{align}
\label{eq:k1}
\mathsf{K}_1(\psi)&=\max_{\mathcal{Q}}\left(1-H(Z^A|BT)_{\psi_{\mathcal{Q}}}-H(X^A|C_ZA'BB'T)_{\psi_{\mathcal{Q},Z}}\right)\\
&=\max_{\mathcal{Q}}\left(H(Z^A|RT)_{\psi_{\mathcal{Q}}}-H(Z^A|BT)_{\psi_{\mathcal{Q}}}\right),
\end{align}
where the second line follows by the same calculations which led to Equation~(\ref{eq:xtoz}).

Alice need only transmit the amplitude syndromes to Bob since they use the encoded amplitude $\overline{Z}$ as the final key. The phase syndromes need not be transmitted, since according to Theorem~\ref{thm:compps} the mere existence of a phase-predicting measurement $\mathcal{M}_X^{C_ZA'BB'}$ ensures the secrecy of the key. This means the protocol can be immediately converted into a secret-key distillation scheme in which Alice and Bob make their amplitude measurements first, Alice then transmits the amplitude syndromes, and finally Alice and Bob compute the final key from the encoded amplitude operator $\overline{Z}$. From the outside, they \emph{could} have actually run the private state distillation protocol, phase stabilizer measurement and all, and so the secret key distillation protocol inherits security from the private state distillation protocol. 

Regularization can again in principle increase the rate further, and the resulting rate is identical to the upper bound found by Devetak and Winter~\cite{devetak_distillation_2005}. Thus we have constructed a secret key distillation protocol which achieves the optimal rate 
\begin{align}
\label{eq:regK}
\mathsf{K}(\psi)=\lim_{n\rightarrow \infty}\frac1n \mathsf{K}_1(\psi^{\otimes n}).
\end{align}

Given a shared, secret key Alice can transmit secret messages to Bob over a public communication channel simply by encrypting the message with the key. For absolute security, Shannon showed that one requires a key exactly as long as the message~\cite{shannon_communication_1949}, and the message may be encrypted by simply computing the exclusive-\textsc{or} of the key, a scheme known as a \emph{one-time pad} or \emph{Vernam cipher} after its inventor~\cite{vernam_cipher_1926}. 

Therefore Alice and Bob may use the secret key distillation scheme above for private communication over public channels. As Alice can choose the input to the channel, she may simply select that input which gives the output with the largest distillable key. Then they proceed with secret key distillation and the one-time pad. This gives a private communication rate of at least $\mathsf{P}_1(\mathcal{N})$ using $\mathsf{K}_1$ above, at least when assisted with public communication. This quantity is sometimes referred to as the \emph{private information} and we shall encounter it again in Section~\ref{sec:psqkd}. Once more, regularization may improve the rate, and the resulting expression $\mathsf{P}(\mathcal{N})$ was shown to be an upper bound in~\cite{devetak_private_2005}. The the protocol for private communication constructed in this way achieves the capacity. Here we have not attempted to remove the public communication from Alice to Bob as we did in the case of quantum communication, but it is also shown in~\cite{devetak_private_2005} that the private capacity can be achieved even without such assistance.

\chapter{Duality of Protocols}
\label{chap:duality}
In Chapter~\ref{chap:proc} we saw that reconciling Bob's quantum information in system $B$ with Alice's amplitude observable $Z^A$ requires her to send Bob extra information about $Z^A$ at rate $H(Z^A|B)_\psi$. This quantity trades off with $H(X^A|E)_\psi$ in the uncertainty principle Equation~(\ref{eq:jcbjmr}), $H(Z^A|B)_\psi+H(X^A|E)_\psi\geq 1$. As it happens, $H(X^A|E)_\psi$ is also the rate at which Alice can perform privacy amplification of $X^A$, extracting uniformly-distributed bits from $X^A$ which are completely uncorrelated with $E$. Thus, the less information Alice has to send to Bob about $Z^A$, the more randomness she can extract from $X^A$ unknown to $E$. There exists a duality between these two protocols due to the uncertainty principle.
The fact that the rates of the two protocols are connected invites us to think that the protocols themselves may be connected as well---that it may be able to transform one protocol into the other. 

Here we show that this is indeed the case, recounting results from~\citeme{renes_duality_2010,renes_noisy_2010} and presenting some new material. This chapter is divided into four main sections. In the first, we recount how information reconciliation and privacy amplification protocols based on linear hash functions can be transformed into each other, following~\citeme{renes_duality_2010}. The duality extends to non-i.i.d.\ resources where the notion of asymptotic rates is no longer valid, and we remark that this implies a more general form of the uncertainty principle in terms of generalized entropies suitable for such unstructured resources. In the second section, we explore the implications of this duality for constructing entanglement distillation protocols, and by extension, the other related protocols discussed in Chapter~\ref{chap:proc}. This material has not been previously published. The third section is devoted to the result of~\citeme{renes_noisy_2010} which shows that coding schemes for communication of either public or private classical information over noisy channels can be constructed by combining privacy amplification and information reconciliation. Thus, the two dual protocols occupy a very fundamental place in the study of information theory, as they can be combined to generate a variety of protocols for other tasks.

\section{Duality of Privacy Amplification and Information Reconciliation}
\label{subsec:duality}
The duality of information reconciliation and privacy amplification protocols both based on linear universal hashing essentially comes down to complementarity, specifically the fact that amplitude measurements destroy phase information and vice versa. Roughly speaking, if Alice measures amplitude stabilizers to perform information reconciliation of $Z^A$ with Bob, this can also be seen as randomizing the conjugate phase $X^A$ stabilizers, as would be useful in privacy amplification. With Lemma~\ref{lem:guesssec} in mind, we expect that if information reconciliation succeeds and Bob can reliably recover the encoded amplitude $\overline{Z}$, then the encoded phase $\overline{X}$ must be uncorrelated with system $E$. Making this work in reverse is slightly more complicated, and there are two versions, corresponding to Corollaries~\ref{cor:secguessB} and~\ref{cor:secguessR}.

\subsection{Privacy Amplification}
\label{sec:pa}
Before delving into the duality of these protocols, we first describe the process of privacy amplification and the known results in more detail. Imagine that Alice has an $n$-bit classical random variable $X^A$ which is correlated with an external system $E$ in some way. Letting $X^A$ be the phase observable, we can describe this state of affairs as 
\begin{align}
\Psi^{AE}
=\sum_\bx q_\bx\ket{\bx}\bra{\bx}^A\otimes \vartheta_\bx^E.
\end{align}
If the $\vartheta_\bx$ were identical for all $\bx$, then $E$ would have no information about the value of $X^A$. Conversely, if the $\vartheta_\bx$ have disjoint supports, then a measurement of $E$ projecting onto these supports can determine $\bx$ without error. 

First introduced by Bennett, Brassard, and Robert~\cite{bennett_how_1986,bennett_privacy_1988}, the goal of privacy amplification is twofold, to compute some function $\overline{X}^A=f(X^A)$ of $X^A$ which is both uniformly distributed and independent of $E$. Keeping only the function output $f(\bx)$ means that, for a given output $\overline{\bx}$, the state in $E$ is averaged over all the $\bx$ for which $f(\bx)=\overline{\bx}$. The goal is then to average over enough values of $\bx$ so that the conditional states $\vartheta_\bx=\sum_{\bx:f(\bx)=\overline{\bx}}q_\bx \vartheta_\bx^E$ are identical for all $\bx$.
Of course, it is unrealistic to expect such an ideal output, so we settle for $p_{\rm secure}(\overline{X}^A|E)\geq 1-\epsilon$. When the state $\Psi^{AE}$ is $n$ instances of a state $\psi^{AE}$ pertaining to a single bit in $A$, the asymptotically-optimal rate at which private random bits can be extracted is defined by the largest rate achievable in the simultaneous limits $n\rightarrow \infty$, $\epsilon\rightarrow 0$. 

In the case that the states $\vartheta_\bx^E$ are classical, i.e.\ simultaneously diagonalizable, Bennett \emph{et al.}\ have shown that universal hashing can be used for privacy amplification~\cite{bennett_how_1986,bennett_privacy_1988,bennett_generalized_1995}. 
Using random coding techniques in the i.i.d.\ setting, Devetak and Winter proved that the rate $H(Z^A|E)_\psi$ is achievable in the asymptotic limit for quantum $\vartheta_\bx^E$ ~\cite{devetak_distillation_2005}, while Renner and K\"onig show that universal hashing is also effective against quantum adversaries even for unstructured, non-i.i.d.\ resources~\cite{renner_universally_2005}. 

One drawback of approaches based on universal hashing is the need for a large amount of randomness to select the hash function from the family, $\Theta(n)$ \emph{seed} bits for $n$ input bits. Smaller function families would naturally be preferable. If we are unconcerned with privacy, the task reduces to extracting the maximum amount of randomness inherent in the distribution of $Z^A$, and constructing efficient \emph{extractors} has been the subject of much research in theoretical computer science (see e.g.\ Shaltiel~\cite{shaltiel_recent_2004} for a review). 

In particular, Trevisan's breakthrough construction showed that essentially all the randomness may be extracted from the input using extractors with seeds of size $O({{\rm polylog}(n)})$~\cite{trevisan_construction_1999,trevisan_extractors_2001}. Recently De \emph{et al.}\ showed that Trevisan's construction can be extended to privacy amplification against quantum adversaries~\cite{de_trevisans_2009}. 


\subsection{Privacy Amplification from Information Reconciliation}
Now we examine how an information reconciliation protocol using linear functions for universal hashing can be used for privacy amplification. Use of CSS codes makes this simple. Consider, as usual, a tripartite pure state $\ket{\psi}^{ABE}$. Instead of taking system $A$ to be a qubit, we now assume that it has dimension $2^n$ for some $n$. This can done without loss of generality by embedding $A$ into a state space larger than the support of $\psi^{A}$, and allows us to think of system $A$ as a collection of $n$ qubits. 

Suppose that there exists a protocol for information reconciliation of Bob's information with the Alice's amplitude $Z^A$ which calls for Alice to compute a linear function of $Z^A$ and send it to Bob. This computation can be thought of as measuring the stabilizers of a CSS code which contains only $Z$-type stabilizers. In terms of virtual qubits as described in Section~\ref{subsec:compqecc}, the entire collection of qubits can be grouped into two subsets, the encoded qubits and the stabilizer qubits. Denoting the amplitude values of the encoded qubits by $\overline{\bz}$ and those of the stabilizer qubits by $\widehat{\bz}$, we can express the initial state as (abusing notation slightly)
\begin{align}
\label{eq:pafromir}
\ket{\psi}^{ABE}=\sum_\bz \sqrt{p_\bz}\ket{\bz}^A\ket{\varphi_\bz}^{BE}=\sum_{\overline{\bz},\widehat{\bz}}\sqrt{p_{\overline{\bz},\widehat{\bz}}}\ket{\overline{\bz}}^{\overline{A}}\ket{\widehat{\bz}}^{\widehat{A}}\ket{\varphi_{\overline{\bz},\widehat{\bz}}}^{BE},
\end{align}
 where $\overline{A}$ ($\widehat{A}$) denotes the virtual subsystem of the encoded (stabilizer) qubits. 

The information reconciliation protocol assures us that given the value of $\widehat{\bz}$, Bob can determine the value of $\bz$ and therefore $\overline{\bz}$. That is, there exists a measurement on $\widehat{A}B$ which can reliably predict the amplitude of $\overline{A}$ with guessing probability greater than $1-\epsilon$ for some small $\epsilon$. Then by Lemma~\ref{lem:guesssec}, $p_{\rm secure}(X^{\overline{A}}|E)_\psi\geq 1-\sqrt{2\epsilon}$. Therefore, to generate a random secret string from the phase observable $X^A$, Alice can simply compute the encoded phase $\overline{X}$.

\subsection{Information Reconciliation from Privacy Amplification}
\label{sec:irfrompa}
Showing that a privacy amplification protocol can be repurposed for information reconciliation is somewhat more involved. Here we encounter the same complications as in Section~\ref{sec:dnd}: Just because $E$ has no knowledge of $X^A$  does not imply that $B$ can predict $Z^A$. But the same technique used there of imposing extra conditions so that the uncertainty principle is saturated works here as well. There are two separate cases to consider.\\

In the first of these we require $p_{\rm guess}(Z^A|E)_\psi=1$, meaning we might as well write the state as 
\begin{align}
\label{eq:pascenario}
\ket{\psi}^{ABE}=\sum_\bz\sqrt{p_\bz}\ket{\bz}^A\ket{\bz}^{E_1}\ket{\varphi_\bz}^{BE_2},
\end{align}
for $E=E_1E_2$. This is somewhat more natural for the goal of amplitude information reconciliation, as it ensures that the $AB$ state describes a classical variable in $A$ and a quantum state in $B$: $\psi^{AB}=\sum_\bz p_\bz P_\bz^A\otimes \varphi_\bz^B$. 

Now suppose that there exists an encoded $\overline{X}$ such that $p_{\rm secure}(\overline{X}^A|E)_\psi\geq 1-\epsilon$. Again using the encoded and stabilizer qubits for system $A$, it follows from Corollary~\ref{cor:secguessR} that there exists a measurement $\mathcal{M}_{\overline{Z}}$ on $\widehat{A}B$ which can recover $\overline{Z}$ with error probability less than $\sqrt{2\epsilon}$. However, Bob does not have access to $\widehat{A}$, and Alice must take care in what information she sends to Bob, lest it leak any information about the phase to $E$. Intuitively, however, measuring amplitude stabilizers on $\widehat{A}$ destroys any phase information that might be present, so it should be safe to transmit the resulting syndromes to Bob. 

Indeed, the formal nature of the state shared by Alice and Bob makes this clear, since $\widehat{A}$ is effectively already measured. Tracing out $E$, we obtain
\begin{align}
\psi^{AB}=\sum_{\overline{\bz},\widehat{\bz}}p_{\overline{\bz},\widehat{\bz}}P_{\overline{\bz}}^{\overline{A}}\otimes P_{\widehat{\bz}}^{\widehat{A}}\otimes \varphi_{\overline{\bz},\widehat{\bz}}^B.
\end{align}
Due to the classical structure of system $\widehat{A}$, we can assume without loss of generality that the measurement $\mathcal{M}_Z^{\widehat{A}B}$ has this structure, too. For let $\Lambda_{\overline{\bz}}^{\overline{A}B}$ be the POVM elements of the $\mathcal{M}_{\overline{Z}}^{\widehat{A}B}$ and consider the joint probability of obtaining the outcome $\mathcal{M}_{\overline{Z}}^{\widehat{A}B}=\overline{\bz}'$ and  $\overline{Z}=\overline{\bz}$,
\begin{align}
{\rm Pr}\big[\mathcal{M}_{\overline{Z}}^{\widehat{A}B}=\overline{\bz}',\overline{Z}=\overline{\bz}\big]=\sum_{\widehat{\bz}}p_{\overline{\bz},\widehat{\bz}}{\rm Tr}\big[\Lambda_{\overline{\bz}'}^{\widehat{A}B}P_{\widehat{\bz}}^{\widehat{A}}\otimes \varphi_{\overline{\bz},\widehat{\bz}}^B\big].
\end{align}
Clearly the same probability results if we first determine the value of $\widehat{\bz}$ and then use a POVM on $B$ having elements $\Pi^B_{\overline{\bz};\widehat{\bz}}={\rm Tr}[P_{\widehat{\bz}}^{\widehat{A}}\Lambda_{\overline{\bz}}^{\overline{A}B}]$. But this is precisely how we expected the information reconciliation process to work: after learning $\widehat{Z}$, Bob can measure $B$ and recover $\overline{Z}$.\\

In the second case we require $p_{\rm guess}(X^A|B)_\psi=1$, so that Bob already has information about the phase. Should he also learn the amplitude, Alice and Bob would have created an entangled state, so this scenario is essentially the latter half of an entanglement distillation scheme. In fact, the protocol of Devetak and Winter in~\cite{devetak_distillation_2005} is constructed along these lines. 
Just as in the previous scenario, if a privacy amplification protocol can construct an encoded phase $\overline{X}$ uncorrelated with $E$, then the conjugate encoded amplitude $\overline{Z}$ must be reliably recoverable by measurement on $\widehat{A}B$, though now the implication follows from Corollary~\ref{cor:secguessB}. However, we cannot use the same argument to show that measurement of the amplitude of $\widehat{A}$ is sufficient to enable information reconciliation using system $B$. 

Instead, we can proceed as follows. From the requirement $p_{\rm guess}(X^A|B)_\psi=1$, the 
marginal state of the $AE$ subsystems takes the form $\psi^{AE}=\sum_\bx q_\bx P_\bx^A\otimes \vartheta_\bx^E$ for some probabilities $q_\bx$ and normalized states $\vartheta_\bx^E$. Decomposing Alice's qubits into virtual encoded and stabilizer qubits, the state is, in a slight abuse of notation, just
$\psi^{AE}=\sum_{\overline{\bx},\widehat{\bx}}q_{\overline{\bx},\widehat{\bx}}P_{\overline{\bx}}^{\overline{A}}\otimes P_{\widehat{\bx}}^{\widehat{A}}\otimes \vartheta_{\overline{\bx},\widehat{\bx}}$. Since the $\widehat{A}$ system is in a phase eigenstate, measuring its amplitude delivers a completely random outcome and results in precisely the same state as if $\widehat{A}$ were traced out. But the encoded phase is chosen by the privacy amplification protocol so that disposing of the stabilizer qubits leaves a nearly ideal key, and the amplitude measurement of the stabilizer qubits does not change this. Thus, for every measurement result we can conclude by Corollary~\ref{cor:secguessB} that there exists a measurement on $B$ which gives $\overline{\bz}$ with high probability.

In both of these situations the desired measurement is only shown to exist, but is not directly constructed. However, due to a result by Barnum and Knill, this presents no real difficultly, as the pretty-good measurement has an error probability which is at most a factor of two worse than the optimal case~\cite{barnum_quantum_2000}. Thus, if privacy amplification is possible so that $p_{\rm secure}(\overline{X}^A|E)_\psi\geq 1-\epsilon$, then using the amplitude stabilizer measurement and the pretty good measurement for Bob's conditional marginal states results in information reconciliation protocols with error probability less than $2\sqrt{2\epsilon}$.

\subsection{One-Shot Protocols and a Generalized Uncertainty Principle} 
\label{subsec:oneshot}
In the preceding sections we have treated Alice's system as a collection of $n$ qubits, but it is important to note that the duality holds for arbitrary resource states, not just i.i.d.\ states. The i.i.d.\ setting is only necessary to define the asymptotically-achievable rates of the various protocols. Recently, a new framework has been constructed which makes it possible to characterize protocols operating on arbitrary, structureless resource states in terms of \emph{smooth entropies}. A proper treatment of smooth entropies and their calculus is beyond the scope of this thesis, but we remark that they can be thought of as generalizations of R\'enyi entropies which are somewhat more familiar in standard information theory and obey many of the same chain rules as the usual Shannon or von Neumann entropies. Here we wish to point out that the duality above, in particular the former duality of Section~\ref{sec:irfrompa}, implies a new entropic uncertainty principle formulated in terms of smooth entropies. 

There are two different smooth entropies, the smooth min-entropy and the smooth max-entropy, and each comes in both conditional and unconditional varieties. It turns out that the number $\ell^\epsilon_{\rm ext}(X^A|E)_\psi$ of $\epsilon$-good random bits one can extract from $Z^A$ which are secret from $E$ is characterized by the smooth min-entropy,  $\ell^\epsilon_{\rm ext}(X^A|E)_\psi\approx H_{\rm min}^\epsilon(Z^A|E)_\psi$~\cite{renner_universally_2005,renner_security_2005,koenig_sampling_2007,tomamichel_leftover_2010}. More precisely, $\ell^\epsilon_{\rm ext}(X^A|E)_\psi$ equals $H_{\rm min}^\epsilon(Z^A|E)_\psi$ up to small deviations involving the smoothing parameter $\epsilon$. Much the same holds for information reconciliation, except using the smooth max-entropy. As shown by the present author and Renner~\cite{renes_one-shot_2010}, the number of bits Alice needs to send to Bob, generated by universal hashing, is given by $\ell_{\rm rec}^\epsilon(Z^A|B)_\psi\approx H_{\rm max}^\epsilon(Z^A|B)_\psi$. Though it might not appear so, the definitions of the smooth entropies are logically distinct from the operational quantities $\ell_{\rm ext}^\epsilon$ and $\ell_{\rm rec}^{\epsilon}$. It should be noted, however, that the smooth entropies are themselves related to the operational quantities $p_{\rm guess}$ and $p_{\rm secure}$, a fact discovered by K\"onig\etalsp\cite{koenig_operational_2009}. 

Now consider a quantum state of the form given in Equation~(\ref{eq:pafromir}). Information reconciliation of the amplitude requires that Alice send $\ell_{\rm rec}^\epsilon(Z^A|B)_\psi$ bits obtained via universal hashing of $Z^A$ to Bob. But this implies Alice can equally-well use the encoded phase to generate random bits uncorrelated with $E$. In all she can create $n-\ell_{\rm rec}^\epsilon(Z^A|B)_\psi$ random bits this way, which must of course be less than the bound on privacy amplification established by the smooth min-entropy. Similarly, $\ell_{\rm rec}^\epsilon(Z^A|B)_\psi$ is bounded by the smooth max-entropy, so we anticipate from this heuristic argument that 
\begin{align}
H_{\rm min}^\epsilon(X^A|E)_\psi+H_{\rm max}^\epsilon(Z^A|E)_\psi\apprge  n.
\end{align}
Indeed, the full analysis performed in~\citeme{renes_duality_2010} shows that the above expression is correct, up to terms of order $\log(1/\epsilon)$. The state in Equation~(\ref{eq:pafromir}) is arbitrary, so this generalized uncertainty principle holds for conjugate observables and any tripartite quantum state. Recently, Tomamichel and Renner have found a simple proof which extends the above uncertainty relation to arbitrary observables in the manner of Equation~(\ref{eq:jcbjmr})~\cite{tomamichel_uncertainty_2011}.

\section{Different Approaches to Entanglement Distillation}
The entanglement distillation protocol presented in Section~\ref{sec:opted} was built by combining information reconciliation protocols for both Alice's amplitude and phase observables. By the duality of information reconciliation and privacy amplification, we expect to be able to trade one task for the other, and base the construction of the protocol on either Theorem~\ref{thm:dw} or Theorem~\ref{thm:2xdecoupling} rather than Theorem~\ref{thm:entdec}. In the following we present these two alternate approaches. It should be stressed that ultimately the alternate approaches followed here yield the same protocol as in  Section~\ref{sec:opted}, but they have completely independent justifications.

In the first approach, we may think of the phase information reconciliation in the original protocol as amplitude privacy amplification, which makes the goal of entanglement distillation to simultaneously give Bob full information about Alice's amplitude while ensuring that $E$ has none. Formally, the goal in constructing the protocol is to fulfill the conditions of Theorem~\ref{thm:dw}. Clearly this approach is quite closely related to secret-key distillation, which has nearly the same goals, and indeed was the original approach followed by Devetak and Winter~\cite{devetak_distillation_2005} for entanglement distillation and Devetak in establishing the quantum capacity of a quantum channel~\cite{devetak_private_2005}. Here we  construct an entanglement distillation protocol having the same aims but a somewhat different structure, namely the use of CSS codes by Alice.  

In the second approach, we can give up on Bob altogether and focus entirely on removing amplitude and phase correlations from $E$, with the aim of fulfilling the conditions of Theorem~\ref{thm:2xdecoupling}. To our knowledge, this approach is new. It shows that the commonly used quantum decoupling method can be broken down into two classical decoupling steps, further reinforcing the claim that quantum information processing can be understood as a combination of classical information processing of amplitude and phase information.  Figure~\ref{fig:protocolcomp} depicts the relationship between the three approaches.  

\def\xwid{4.5}
\def\zwid{5.2}
\def\nwid{3}
\def\dx{.5}
\def\height{.25}
\def\gap{.1}

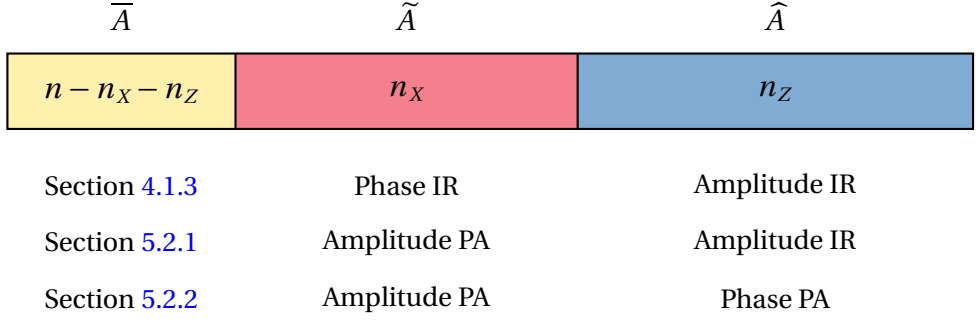
\begin{figure}[h!]
\begin{center}
\begin{tikzpicture}[thick]

\fill[tud6a!50] (0,0) rectangle (\nwid,1);
\fill[tud9b!50] (\nwid,0) rectangle (\nwid+\xwid,1);
\fill[tud1b!50] (\nwid+\xwid,0) rectangle (\nwid+\xwid+\zwid,1);

\draw (0,0) rectangle (\nwid+\xwid+\zwid,1);
\draw (\nwid,0) -- (\nwid,1);
\draw (\nwid+\xwid,0) -- (\nwid+\xwid,1);

\node at (\nwid/2,.5) {\large $n-n_X-n_Z$};
\node at (\nwid/2,1.5) {\large $\overline{A}$};

\node at (\nwid+\xwid/2,.5) {\color{black}\large $n_X$};
\node at (\nwid+\xwid/2,1.5) {\large $\widetilde{A}$};

\node at (\nwid+\xwid+\zwid/2,.5) {\large $n_Z$};
\node at (\nwid+\xwid+\zwid/2,1.5) {\large $\widehat{A}$};

\node at (\nwid/2,-.75) {Section~\ref{subsec:optimal}};
\node at (\nwid/2,-1.5) {Section~\ref{subsec:edirpa}};
\node at (\nwid/2,-2.25) {Section~\ref{subsec:edpapa}};

\node at (\nwid+\xwid/2,-.75) {Phase IR};
\node at (\nwid+\xwid+\zwid/2,-.75) {Amplitude IR};

\node at (\nwid+\xwid/2,-1.5) {Amplitude PA};
\node at (\nwid+\xwid+\zwid/2,-1.5) {Amplitude IR};

\node at (\nwid+\xwid/2,-2.25) {Amplitude PA};
\node at (\nwid+\xwid+\zwid/2,-2.25) {Phase PA};

\end{tikzpicture}
\caption{\label{fig:protocolcomp} Breakdown of Alice's $n$ physical qubits into three subsets of virtual qubits in subsystems $\overline{A}$, $\widehat{A}$, and $\widetilde{A}$ and what the different subsets are used for in the various approaches to entanglement distillation presented here. Theorem~\ref{thm:entdec} is the goal of the construction in Section~\ref{subsec:optimal}, where $\widehat{A}$ is used to reconcile the amplitude information with Bob and $\widetilde{A}$ the phase information. The construction in Section~\ref{subsec:edirpa} takes Theorem~\ref{thm:dw} as its goal, and $\widehat{A}$ is again used for amplitude information reconciliation with Bob, but $\widetilde{A}$ is used for privacy amplification of the same amplitude information against the environment. Finally, Theorem~\ref{thm:2xdecoupling} is the aim of construction in Section~\ref{subsec:edpapa}, where $\widehat{A}$ is used to decouple Alice's phase information from the environment and $\widetilde{A}$ her amplitude information.}
\end{center}
\vspace{-14pt}
\end{figure}

\subsection{Amplitude Information Reconciliation \& Privacy Amplification}
\label{subsec:edirpa}
Although the approach based on Theorem~\ref{thm:dw} is substantially similar to that pursued in~\cite{devetak_distillation_2005}, we include it here for completeness. Again we consider the case in which Alice and Bob share asymptotically-many copies of a resource state $\psi^{AB}$ which may be purified to $\ket{\psi}^{ABE}$. We will construct the protocol by choosing two sets of amplitude stabilizers, first a number $n_{Z}$ large enough to enable information reconciliation with Bob and the second $n_{X}$ to achieve privacy amplification against $E$. Thinking in terms of virtual qubits and their associated amplitude and phase operators, let us call the encoded amplitude operators $\overline{Z}$, of which there are $n-n_{Z}-n_{X}$, the $n_{Z}$ stabilizers associated with information reconciliation $\widehat{Z}$, and those $n_{X}$ associated with privacy amplification $\widetilde{Z}$. 

To ensure that Bob can reconstruct the original amplitude, and therefore the encoded $\overline{Z}$, Alice measures the $\widehat{Z}$ stabilizers and sends the resulting syndromes to Bob. This could give additional information about $\overline{Z}$ to $E$, but if the $\widetilde{Z}$ stabilizers are numerous enough, averaging over their syndromes destroys whatever information $E$ had about the original amplitude $Z$. By itself, $\widetilde{Z}$ is independent of $\overline{Z}$, since they belong to different sets of virtual qubits, so Alice can be certain that no information leaks to $E$ in this process. 

We are not ready to apply Theorem~\ref{thm:dw}, however. The shared state at this step in the protocol is 
\begin{align}
\ket{\Psi'}^{\overline{A}\widetilde{A}BB'EE'}=\sum_{\overline{\bz},\widehat{\bz},\widetilde{\bz}}\sqrt{p_{\overline{\bz},\widehat{\bz},\widetilde{\bz}}}\ket{\overline{\bz}}^{\overline{A}}\ket{\widetilde{\bz}}^{\widetilde{A}}\ket{\widehat{\bz}}^{B'}\ket{\widehat{\bz}}^{E'}\ket{\varphi_{\overline{\bz},\widehat{\bz},\widetilde{\bz}}}^{BE},
\end{align}
where the amplitude of $\widehat{A}$ has been transferred and copied to new systems $B'$ and $E'$, which mimics the classical measurement of $\widehat{Z}$ and broadcast of the result $\widehat{\bz}$. From information reconciliation there is a measurement $\mathcal{M}_{\overline{Z}}$ on $BB'$ such that $p_{\rm guess}(\overline{Z}|\mathcal{M}_{\overline{Z}}^{BB'})$ is close to one, and via the above discussion of privacy amplification $p_{\rm secure}(\overline{Z}|EE')$ is likewise nearly one. To apply Theorem~\ref{thm:dw} we still need to discard $\widetilde{A}$ without changing either of these conditions. 

This situation is precisely that of the second case of the previous section, from which it follows that measuring the phase $\widetilde{X}$ will not decrease Bob's guessing probability and will also not leak any information about $Z$ to $E$. Formally, we can see this by examining the state after the phase stabilizer measurement, 
\begin{align}
\label{eq:ampphasmeas}
\ket{\Psi''}^{\overline{A}BB'B''EE'E''}&=\frac{1}{\sqrt{2^{n_{\rm PA}}}}\sum_{\widetilde{\bx},\overline{\bz},\widehat{\bz},\widetilde{\bz}}\sqrt{p_{\overline{\bz},\widehat{\bz},\widetilde{\bz}}}(-1)^{\widetilde{\bx}\cdot\widetilde{\bz}}\ket{\overline{\bz}}^{\overline{A}}\ket{\widehat{\bz}}^{B'}\ket{\widetilde{\bx}}^{B''}\ket{\widehat{\bz}}^{E'}\ket{\widetilde{\bx}}^{E''}\ket{\varphi_{\overline{\bz},\widehat{\bz},\widetilde{\bz}}}^{BE}\\
&=\frac{1}{\sqrt{2^{n_{\rm PA}}}}\sum_{\widetilde{\bx},\overline{\bz},\widehat{\bz},\widetilde{\bz}}\sqrt{p_{\overline{\bz},\widehat{\bz},\widetilde{\bz}}}\ket{\overline{\bz}}^{\overline{A}}\ket{\widehat{\bz}}^{B'}(X^{\widetilde{\bz}})^{B''}\ket{\widetilde{\bx}}^{B''}\ket{\widehat{\bz}}^{E'}\ket{\widetilde{\bx}}^{E''}\ket{\varphi_{\overline{\bz},\widehat{\bz},\widetilde{\bz}}}^{BE}.
\end{align}
Because $\widetilde{\bz}$ only shows up as part of a unitary operator on $B''$, tracing out all of Bob's systems means the state in $E$ is averaged over these values, which was precisely the goal of privacy amplification. Moreover, $\widehat{\bz}$ by itself is uncorrelated with $\overline{\bz}$.   
Thus, in transferring the phase of $\widetilde{A}$ to systems $B''$ and $E''$, we have $p_{\rm guess}(\overline{Z}|\mathcal{M}_{\overline{Z}}^{BB'B''}), p_{\rm secure}(\overline{Z}|EE'E'')\approx 1$. Hence we can apply Theorem~\ref{thm:dw} to infer that Alice and Bob can recover a high-quality entangled state from their systems. By the known results on information reconciliation and privacy amplification, we can pick $n_Z\approx nH(Z^A|B)_\psi$ and $n_X\approx n-nH(Z^A|E)_\psi$ so that the rate achievable by this protocol is $H(Z^A|E)_\psi-H(Z^A|B)_\psi=-H(A|B)_\psi$, the hashing bound.

\subsection{Privacy Amplification of Both Amplitude and Phase}
\label{subsec:edpapa}
The method of the previous section can serve as a stepping stone towards a protocol which is based entirely on decoupling both amplitude and phase from $E$. All we have to do is turn the amplitude information reconciliation into privacy amplification of phase. From the discussion prior to Theorem~\ref{thm:2xdecoupling}, we know that it will be insufficient to decouple $E$ from $X$ and $Z$, rather we must aim to simultaneously decouple $E$ from $Z$ on the one hand, and $C_ZE$ from $X$ on the other. Note that in the latter case the state $\ket{\psi_Z}^{AC_ZBE}$ is only a device used in the proof; it does not need to show up in the protocol directly. 

To achieve this simultaneous decoupling, we again begin by specifying two sets of stabilizers, $n_Z$ $Z$-type stabilizers to decouple the amplitude and $n_X$ $X$-type stabilizers to decouple the phase. As before, Alice's $n$ qubits can be grouped into three sets of virtual qubits, the $n-n_Z-n_X$ encoded qubits in $\overline{A}$, $n_X$ qubits in $\widehat{A}$, and $n_Z$ qubits in $\widetilde{A}$. If $n_X$ and $n_Z$ are chosen appropriately, we can be sure that both $p_{\rm secure}(\overline{X}^A|C_ZE)_{\psi_Z}$ and $p_{\rm secure}(\overline{Z}^A|E)_\psi$ are nearly one. Therefore system $\overline{A}$ is implicitly in a maximially-entangled state with the joint system $\widehat{A}\,\widetilde{A}\,B$, and the remaining task is to classically transfer $\widehat{A}\,\widetilde{A}$ to Bob without violating the privacy conditions. 

Following the method of the previous construction, suppose Alice makes \emph{amplitude} measurements on $\widehat{A}$ and \emph{phase} measurements on $\widetilde{A}$, which she then broadcasts this information publicly. While $E$ now recieves extra information about the original amplitude and phase, no information about the encoded amplitude and phase has been leaked for the same reason as in the previous construction. The marginal states in $E$ conditioned on the encoded amplitude (phase) value are still averaged over enough $\bz$ ($\bx$) values to make them essentially identical.

Formally, the situation is very similar to the previous case as well. In fact, for the observable $\overline{Z}$, the state of $\ket{\psi}^{ABE}$ after the measurements described above is precisely that of Equation~(\ref{eq:ampphasmeas}), and so we can immediately conclude that $p_{\rm secure}(\overline{Z}|EE'E'')\approx 1$. The state relevant to privacy amplification of the phase can be expressed as, following Equation~(\ref{eq:psizphase}),
\begin{align}
\ket{\psi_Z}^{AC_ZBE}&=\tfrac{1}{\sqrt{2^n}}\sum_{\overline{\bx},\widehat{\bx},\widetilde{\bx}}\ket{\overline{\bx}}^{\overline{A}}\ket{\widehat{\bx}}^{\widehat{A}}\ket{\widetilde{\bx}}^{\widetilde{A}}(Z^\bx)^{C_Z}\ket{\psi}^{C_ZBE},
\end{align}
and after the measurement it becomes 
\begin{align}
\ket{\Psi_Z''}^{\overline{A}BB'B''EE'E''}&= \tfrac{1}{\sqrt{2^{n+n_X}}}\sum_{\widehat{\bz},\overline{\bx},\widehat{\bx},\widetilde{\bx}}(-1)^{\widehat{\bx}\cdot\widehat{\bz}}\ket{\overline{\bx}}^{\overline{A}}\ket{\widehat{\bz}}^{B'}\ket{\widehat{\bz}}^{E'}\ket{\widetilde{\bx}}^{B''}\ket{\widetilde{\bx}}^{E''}(Z^\bx)^{C_Z}\ket{\psi}^{C_ZBE}\\
&= \tfrac{1}{\sqrt{2^{n+n_X}}}\sum_{\widehat{\bz},\overline{\bx},\widehat{\bx},\widetilde{\bx}}\ket{\overline{\bx}}^{\overline{A}}\ket{\widehat{\bz}}^{B'}\ket{\widehat{\bz}}^{E'}\ket{\widetilde{\bx}}^{B''}\ket{\widetilde{\bx}}^{E''}(Z^{\overline{\bx}})^{\overline{C}_Z}(Z^{\widetilde{\bx}})^{\widetilde{C}_Z}\ket{\psi}^{C_ZBE}.
\end{align}
Now the phase $(-1)^{\widehat{\bx}\cdot\widehat{\bz}}$ cancels the similar phase inherent in the operator $(Z^\bx)^C.$ Again this enforces an average over $\widehat{\bx}$ for the states in system $E$, ensuring that they are completely uncorrelated with $\bx$ and therefore $\overline{\bx}$. Just as before, $\widetilde{\bx}$ does not add any additional information about $\overline{\bx}$, so we can conclude that $p_{\rm secure}(\overline{X}|EE'E'')_{\psi_Z}\approx 1$ and therefore Theorem~\ref{thm:2xdecoupling} is applicable.
For $n_X$ and $n_Z$ we can pick $n-nH(X^A|C_ZE)_{\psi_Z}$ and $n-nH(Z^A|E)_\psi$, respectively, yielding an overall rate of $H(X^A|C_ZE)_{\psi_Z}+H(Z^A|E)_\psi-1$. This works out to be $H(A|E)=-H(A|B)$, which is the hashing bound once again. 

\section{Classical Channel Coding}
In Section~\ref{subsec:channelcoding} we described how a protocol for entanglement distillation using one-way communication can be used to reliably send quantum information over a noisy channel, and that protocols achieving the optimal rate of entanglement distillation lead to optimal channel coding. A similar result holds for classical information, as demonstrated in~\citeme{renes_noisy_2010}, albeit using information reconciliation and randomness extraction or privacy amplification. This leads not only to a new proof of Shannon's original noisy channel coding theorem in the case the channel is classical, but also to one-shot results for both public and private communication of classical information over noisy quantum channels. Moreover, using the results of Section~\ref{subsec:duality}, we can exchange the use of information reconciliation with privacy amplification of a complementary observable, and thereby construct a channel coding scheme which is entirely based on decoupling-type arguments. That is, we can construct a means for noisy channel communication not by directly ensuring that the receiver can properly decode the transmissions, but rather by ensuring that complementary information does not leak to the environment.  

On a heuristic level, the approach itself is quite similar to that of Section~\ref{subsec:channelcoding}, not just the result. We can make the same sort of modification to an appropriate information reconciliation protocol as we did to entanglement distillation in order to create a coding scheme for the channel scenario. Suppose that Alice can send classical messages $z\in\{0,1\}$ to Bob over a quantum channel such that he receives the corresponding state $\varphi_z$. If they are in possession of an information reconciliation protocol for the state $\psi^{AB}=\frac12\sum_z \ket{z}\bra{z}^A\otimes\varphi_z^B$, then they can use this to communicate reliably over the channel. In the information reconciliation scheme Alice would compute a hash function of $n$ instances of the random variable $Z^A$, and with this information $f(\bz)$ Bob could determine the actual $\bz$ from his state $\varphi_\bz^B$. 

In the channel scenario this can be used to specify a code by the set of all possible inputs $\bz$ (codewords) which hash to a specified value, say $\widehat{\bz}$. Ordering the elements of this set in some way, Alice can then map her actual message to the corresponding codeword. This defines an encoder. Presumably they have chosen an $\widehat{\bz}$ for which the information reconciliation decoder has a small probability of error, and thus Bob can use that decoder to determine $\bz$ and therefore Alice's intended message. 

In fact, when the original inputs $\bz$ are uniformly distributed as above, one can easily show that not only will Bob have a small average probability of decoding error, but also a low error probability for every message. 
To determine the number of messages Alice can send, it is simplest to consider the case of linear hash functions, where every output has the same number of preimages, namely the ratio of input to output size.\footnote{The general case can be handled by probabilistic arguments~\citeme{renes_noisy_2010}.} Thus, if information reconciliation requires an $m$-bit hash for an $n$-bit string $\bz$, the resulting code can be used to transmit $n-m$ bits, remembering that an $n$-bit input corresponds to $2^n$ possible input strings. 
As we are working in the asymptotic i.i.d.\ scenario, we can apply the result mentioned in Section~\ref{subsec:ir} that information reconciliation is possible at the rate $r=H(Z^A|B)_\psi$, so that $m\approx nr$. Therefore Alice can reliably send messages at rate $1-H(Z^A|B)_\psi$.

There is still room for improvement, however, as the Holevo-Schumacher-Westmoreland (HSW) theorem (the quantum version of Shannon's noisy channel coding theorem) assures us that rates up to at least the Holevo quantity $\chi=\max_{P_Z}I(Z^A{:}B)_\psi=\max_{P_Z}H(Z^A)_\psi-H(Z^A|B)_\psi$ are possible~\cite{holevo_capacity_1998,schumacher_sending_1997}.\footnote{As with quantum communication and private classical communication over quantum channels, regularization can increase the rate further. Indeed, as discussed at the end of Chapter~\ref{chap:qkd}, regularization is necessary to reach the capacity.} Clearly something is missing in the above, unless it happens that the optimal distribution is uniform. We have restricted attention to the uniform distribution for convenience in the proof, in particular to easily determine the number of messages which can be send with small worst-case probability of error. If we only cared about average probability of error, any distribution could be used for the purposes of converting an information reconciliation protocol to a channel code. The difficulty is then to exploit this freedom without requiring a substantially new proof. 

Fortunately, there is a simple way to deal with this problem by making use of the randomness extractors described in Section~\ref{sec:pa}, though here the privacy properties will only be relevant to the case of private channel communication. Alice can use the extractor in reverse as a \emph{distribution shaper} to simulate a random variable $Z$ with arbitrary distribution $P_Z$ using a uniformly-distributed random variable $U$. To do so, Alice chooses an extractor output $u$ at random and then maps it to a possible preimage $\bz$ using the conditional distribution $P_{Z|U=u}$. This requires an additional source of randomness, as the extractor function is not one-to-one. 

When $Z$ is destined to be the input to the communication channel, we can instead think of $U$ as the input to the ``superchannel'' composed of the shaper and the original channel. This is depicted in Figure~\ref{fig:channelcoding}, taken from~\citeme{renes_noisy_2010}. Note that for this step we must rely on the recently-established one-shot results on information reconciliation, as mentioned in Section~\ref{subsec:oneshot}, because the joint state shared by Alice and Bob which describes the input and output is generally not i.i.d.
However, in the one-shot framework all the previous results linking information reconciliation to channel coding can be applied to the superchannel. Alice encodes messages into the outputs $\mathbf{u}$ of the extractor and then sends these first through the shaper and then through the communication channel to Bob. Information reconciliation of $U$ with $B$ enables Bob to recover the original message.

\def\boxsep{2.25}
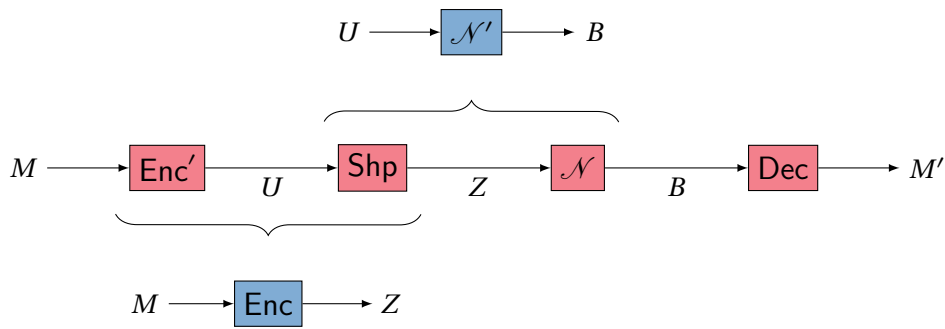
\begin{figure}[t]
\begin{center}
\begin{tikzpicture}[scale=1.2,box/.style={
rectangle,
minimum size=6mm,
draw,
fill=tud9b!50,
font=\itshape
},box2/.style={rectangle, minimum size=6mm, draw, 
fill=tud1b!50,
font=\itshape
}
]
\node (in) at (-0.7*\boxsep,0) {$M$};
\node[box] (encp) at (0,0) {\large $\mathsf{Enc}'$}
  edge [latex-] (in);
\node[box] (shp) at (\boxsep,0) {\large $\mathsf{Shp}$}
  edge [latex-] node[below] {$U$} (encp);
\node[box] (chan) at (2*\boxsep,0) {\large $\mathcal{N}$}
  edge [latex-] node[below] {$Z$} (shp);
\node[box] (dec) at (3*\boxsep,0) {\large $\mathsf{Dec}$}
  edge [latex-] node[below] {$B$} (chan);
\node (out) at (3*\boxsep+0.7*\boxsep,0) {$M'$}
  edge [latex-] (dec);
\draw[decoration={brace,amplitude=8},decorate] ($(shp.south east) + (.15,-.25)$) -- ($(encp.south west) + (-.15,-.25)$);
\draw[decoration={brace,amplitude=8},decorate] ($(shp.north west) + (-.15,.25)$) -- ($(chan.north east) + (.15,.25)$);
\node[box2] (chanp) at ($0.5*(shp.west) + 0.5*(chan.east) + (0,1.5)$) {\large $\mathcal{N}'$};
\node (intop) at ($(chanp)-(0.6*\boxsep,0)$) {$U$}
  edge [-latex] (chanp);
\node (outtop) at ($(chanp)+(0.6*\boxsep,0)$) {$B$}
  edge [latex-] (chanp);
\node[box2] (enc) at ($0.5*(shp.east) + 0.5*(encp.west) + (0,-1.5)$) {\large $\mathsf{Enc}$};
\node (inbot) at ($(enc)-(0.6*\boxsep,0)$) {$M$}
  edge [-latex] (enc);
\node (outbot) at ($(enc)+(0.6*\boxsep,0)$) {$Z$}
  edge [latex-] (enc);
\end{tikzpicture}
\end{center}
\caption{\label{fig:channelcoding}Schematic of using randomness extraction and information reconciliation to perform noisy channel communication. 
Messages $m\in M$ are input to the encoder $\mathsf{Enc}'$ and subsequently to the shaper $\mathsf{Shp}$, which is a randomness extractor run in reverse. Then they are then transmitted over the channel ${\mathcal{N}}$ to the receiver, who uses the decoder $\mathsf{Dec}$ to construct a guess $m'\in M'$ of the original input. Concatenating the shaper and channel gives a new effective channel ${\mathcal{N}}'$, for which an encoder/decoder pair $(\mathsf{Enc}',\mathsf{Dec})$ can be constructed by repurposing an information reconciliation scheme that operates on the joint input-output $UB$ of the channel. Ultimately, the shaper can instead be regarded as part of the encoder $\mathsf{Enc}$, which is formed by concatenating $\mathsf{Enc}'$ and $\mathsf{Shp}$.
}
\end{figure}

Using the smooth entropy results on structureless resources we can determine the (logarithm of the) raw number of messages Alice can reliably send to Bob, instead of the rate as appropriate to the i.i.d.\ setting. The details of the derivation are given in~\citeme{renes_noisy_2010}, and the result is that Alice can reliably transmit $N$ bits to Bob, for 
\begin{align}
\label{eq:classcap}
N_{\rm class}\approx \max_{P_Z}\left[H_{\rm min}^{\epsilon}(Z)_\psi-H_{\rm max}^\epsilon(Z|B)_\psi-O(\log\tfrac1\epsilon)\right]. 
\end{align}
Here $\epsilon$ characterizes the worst-case error probability of the coding scheme, and this expression 
agrees with a result found for classical channels found by Renner\etalsp\cite{renner_single-serving_2006}.\footnote{Wang and Renner have recently derived a one-shot result for classical communication over quantum channels via a different method~\cite{wang_one-shot_2010}.} This result applies to completely arbitrary channels, but when Alice and Bob would like to communicate using $n$ uses of a memoryless channel we can appeal to the asymptotic equipartition property (AEP) of the smooth min- and max-entropies, proven by Tomamichel\etalsp\cite{tomamichel_fully_2009}. Roughly speaking, it states that $H_{\rm min}^\epsilon(Z|B)_{\psi^{\otimes n}}\approx nH(Z|B)_\psi$ and similarly for the max-entropy. We then recover the rate given by the HSW theorem; for channels with purely classical outputs, i.e.\ quantum states which all pairwise commute, we recover Shannon's noisy channel coding theorem~\cite{shannon_mathematical_1948}. 

Besides an appealing modular proof of the noisy channel coding problem based on the simpler primitives of randomness extraction and information reconciliation, another appeal of this approach is that by using privacy amplification instead of just randomness extraction for the distribution shaper, we automatically obtain a construction suitable for private communication of classical information over a noisy quantum channel. In that case we find that the (logarithm of the) number of private messages which can be reliably sent is given by 
\begin{align}
\label{eq:privcap}
N_{\rm priv}\approx  \max_{P_Z}\left[H_{\rm min}^{\epsilon}(Z|E)_\psi-H_{\rm max}^\epsilon(Z|B)_\psi-O(\log\tfrac1\epsilon)\right], 
\end{align}
where system $E$ is the ``other half'' of the channel output. That is, upon input of $z$ the channel produces the pure state $\ket{\varphi_z}^{BE}$ shared between Bob and the environment or eavesdropper. As before, an application of the AEP recovers the rate relevant in the asymptotic i.i.d.\ setting, namely $\max_{P_Z}\left[H(Z|E)_\psi-H(Z|B)_\psi\right]$. This agrees with the findings of Devetak~\cite{devetak_private_2005} for quantum channels, and those of Wyner~\cite{wyner_wire-tap_1975}, Ahlswede and Csiszar~\cite{ahlswede_common_1993}, and Maurer and Wolf~\cite{maurer_information-theoretic_2000} for classical channels. 

Finally, we note that combining this proof technique with the duality between information reconciliation and privacy amplification it is possible to prove that reliable communication is possible by ensuring that not too much information leaks to the environment. This decoupling approach was heretofore unknown to work for channel coding of classical information, and in fact this was the one major protocol not known to be amenable to a decoupling analysis. The encoding and decoding procedure is precisely the same as before, using a distribution shaper and information reconciliation to create an encoder and decoder. But instead of relying on constructions of information reconciliation protocols, we use privacy amplification and duality. Thus, the size of the code is fixed by how much privacy amplification is needed for the observable conjugate to the uniform input $U$, and is therefore given by a smooth min-entropy. Using the uncertainty principle for smooth entropies formulated in~\cite{tomamichel_uncertainty_2011} we can relate this to the smooth max entropy of $U$ conditioned on $B$, and obtain again Equations~(\ref{eq:classcap}) and (\ref{eq:privcap}).

\chapter{Security of Quantum Key Distribution}
\label{chap:qkd}
Quantum key distribution is one of the major current applications of quantum information processing, requiring only minimal ability to coherently manipulate quantum information. Devices implementing QKD protocols such as BB84 are even currently available commercially. But where does the security of QKD come from? That is to say, how can we prove that a given protocol is truly secure and no would-be eavesdropper has any information about the key? 

There have been three main approaches to answering this question, each with its own advantages and disadvantages, which we briefly describe in the first of three sections in this chapter. In the second section we follow one of these methods, treating QKD as a means for virtual creation of entanglement as described in Section~\ref{sec:skd}, and recount the results of~\citeme{renes_generalized_2006} showing that it applies to a wide class of protocols, not just the original BB84 scheme. 

From Section~\ref{sec:secretkey} we know that entanglement is not strictly necessary for generating secret keys, and that in general private states suffice. In the third section of this chapter we describe how alterations to the BB84 protocol which improve the maximum tolerable error rates can be understood as part of a virtual private state distillation scheme, and that combining this additional step with similar enhancements to quantum error-correction lead to still better tolerable error rates. This work was first reported in~\citeme{renes_noisy_2007, smith_structured_2008,kern_improved_2008}. 

\section{Notions of Security}
The first proofs of \emph{unconditional} security of the BB84 protocol---that is, security of the protocol under arbitrary attacks on the public quantum channel by the eavesdropper Eve---were given by Biham\etalsp\cite{biham_proof_2000,biham_proof_2006} and Mayers~\cite{mayers_quantum_1996,mayers_unconditional_2001}. 
Their methods are similar, and essentially rest on an implicit use of the uncertainty principle to bound Eve's information about Alice's key by Bob's information about the conjugate basis to the key.\footnote{Both of their formal statements make use of a related result by Yao~\cite{yao_security_1995}.} Biham\etal characterized the security as due to an information-disturbance tradeoff, the fact that  eavesdropper cannot acquire information about Alice's signals without disturbing them. Such a tradeoff follows immediately from Equation~(\ref{eq:jcbjmr}), as to be able to gain information about e.g.\ the phase without disturbing the amplitude information would imply a violation of the entropic bound. 

At the same time, efforts to base the security of QKD on virtual entanglement distribution as described in Section~\ref{sec:qkd} were underway, culminating in Shor and Preskill's proof for BB84 shortly after the two mentioned above. Their proof was a good deal simpler than the earlier versions, and achieved a higher error 
threshold, the maximum error rate at which the protocol can still safely generate secret keys (albeit at vanishingly small rates). The new proof established a threshold of 11\%, the previous proofs 7.56\%. The simplicity also enabled the method to be extended to other protocols. Lo~\cite{lo_proof_2001} established the unconditional security of the six-state protocol proposed by Bruss~\cite{bru_optimal_1998} which uses the eigenstates of the $XZ$ operator as signals in addition to those of $X$ and $Z$. Tamaki, Koashi, and Imoto~\cite{tamaki_unconditionally_2003} extended the method to a proof of Bennett's two-state protocol (B92)~\cite{bennett_quantum_1992-1}, while Gottesman and Lo showed that it could also treat information reconciliation steps involving two-way communication~\cite{gottesman_proof_2003}, greatly increasing the error rate tolerable by BB84 to 18.9\%. Boileau\etal (including the present author)~\cite{boileau_unconditional_2005} proved the security of a B92-like protocol involving three states which was originally proposed by Phoenix\etalsp\cite{phoenix_three-state_2000}.

The original approach of Biham\etal and Mayers has its own advantages within the realm of the BB84 protocol, however, as it is not actually concerned with the details of Bob's 
measurement apparatus, only Alice's preparation device. This can be anticipated from the implicit use of the uncertainty principle: From Equation~(\ref{eq:jcbjmr}) it is clear that to bound Eve's knowledge of the key it suffices to have a bound on Bob's knowledge of the conjugate observable to the key. It is not necessary to have an accurate physical description of how he comes by such knowledge, which greatly extends the practicality of the proof. Koashi and Preskill combined techniques from both  methods to treat the problem of an uncharacterized source~\cite{koashi_secure_2003} (but characterized detector), and later Koashi gave an even simpler proof which was the first to quantitatively appeal to the uncertainty principle~\cite{koashi_unconditional_2006,koashi_simple_2009}. Although the proof itself is constructed via other means, Koashi used the Maassen and Uffink relation, Equation~(\ref{eq:maassen}), as a guide to determine the size of the secret key. Very recently, Tomamichel\etalsp\cite{tomamichel_tight_2011} have directly used the smooth entropy uncertainty relation of~\cite{tomamichel_uncertainty_2011} to give a simple security proof of BB84 with uncharacterized detectors.

Meanwhile, a third general approach focused on showing that privacy amplification produces secure keys even when the adversary holds quantum instead of classical information. 
 To make use of privacy amplification one then needs to characterize the quantum states held by the eavesdropper, or at least give a bound on the size of their overall support. Ben-Or showed that a result from quantum communication complexity implies the efficacy of privacy amplification and that the knowledge gained by Alice and Bob in the BB84 protocol can be used to bound the effective size of Eve's system~\cite{ben-or_security_2002}. K\"onig\etal demonstrated that privacy amplification works against quantum adversaries generally~\cite{koenig_power_2005}, and Christandl\etal developed this into a generic security proof which replicated the one-way results above, even improving the threshold for the B92 protocol~\cite{christandl_generic_2004}. Kraus, Gisin, and Renner~\cite{kraus_lower_2005,renner_information-theoretic_2005} extended this to establish that many protocols are not only unconditional secure, but also safely \emph{composable} with other cryptographic primitives to create larger cryptographic schemes which are themselves secure, following composability results by Renner and K\"onig~\cite{renner_universally_2005} and Ben-Or\etalsp\cite{ben-or_universal_2005}. Renner provided another method also suitable for two-way protocols in his thesis~\cite{renner_security_2005}.

It should be noted that the task of key \emph{distribution} is considerably more involved than the task of key \emph{distillation} as discussed in Section~\ref{sec:skd}, and the security issue all the more complex. There the input state shared by Alice and Bob is known in advance, and moreover it is assumed to consist of $n$ copies of some state $\psi$. Neither of these statements hold in general in the present context, for although Alice sends $n$ quantum systems to Bob, these travel over an insecure communication channel which could in principle be under the control of the would-be eavesdropper Eve. The difficulty lies in the fact that the eavesdropper could in principle attack all the signals jointly, what is termed a \emph{coherent attack}. If Eve attacked each signal separately, a \emph{collective attack}, then Alice's and Bob's state would have the aforementioned i.i.d.\ form, and could be handled by those methods. 

Unsurprisingly, then, one widely-used method of handling coherent attacks is to reduce them in some way to collective attacks. Originally this was done on a more \emph{ad hoc} basis for particular protocols, but has been made more systematic by Renner~\cite{renner_security_2005,renner_symmetry_2007}, culminating in a very general statement by Christandl\etalsp\cite{christandl_postselection_2009}. This states that as long as the key distribution protocol is unconcerned with the order in which Alice transmits the signals, which can be enforced by arbitrarily permuting them, then security against collective attacks implies security against coherent attacks.  

\section{Entanglement in Prepare and Measure QKD}
\label{sec:entqkd}
Quantum key distribution can be formulated as a virtual entanglement distribution scheme for a wide class of protocols and the Shor-Preskill approach used to prove the their security. In this section we briefly sketch out how this can be done, following~\citeme{renes_generalized_2006} and simplifying some issues in light of intervening research advances. The main conceptual difficulty in considering protocols other than BB84 in the Shor-Preskill framework is that it appears as if the CSS structure of information reconciliation and privacy amplification are directly related to the use of amplitude and phase eigenstates as the signals and for measurement. However, this is not actually the case, and in fact these two parts of the protocol have nothing to do with each other. 
This was already noted in the proofs by Tamaki\etalsp\cite{tamaki_unconditionally_2003} and Boileau\etalsp\cite{boileau_unconditional_2005}, but~\citeme{renes_generalized_2006} show how it can be made to work more generally.

First let us settle on the general framework of prepare and measure protocol. A generic protocol consists of five main stages. First Alice prepares quantum states and transmits them over the insecure quantum channel to Bob, who measures them; this is the only step in which quantum operations are actually needed. Second, they transform their classical transmission and measurement records to a prospective \emph{raw key}. This step is usually called sifting, after the specific mapping used in BB84, and usually the transformation is chosen so that the raw key would be a truly secret key if the quantum channel were noiseless. 

As real channels are inevitably noisy, Alice and Bob need to distill a truly shared, secret key from the raw key. In stage three, parameter estimation, they compare some random subset of the raw key to determine the likely number of errors. This serves two purposes. In the fourth stage, information reconciliation, they use the knowledge from parameter estimation to agree on an identical refined key. Usually this involves Bob reconciling his raw key to Alice's, hence the name. Finally, they also use this knowledge to perform privacy amplification and thereby generate the final secure key.  

The trick to applying the Shor-Preskill framework more generally is to first formulate the prepare and measure process coherently, i.e.\ in quantum-mechanical language, and then regard Alice's and Bob's systems in this setting as being composed of two virtual subsystems. One subsystem (quantumly) records  the key value, while the other (quantumly) records the sifting information. The sifting stage can then be seen as a measurement of the latter subsystems, plus postselection by public communication to select appropriately matching sifting outcomes. The virtual key subsystems remain, and it is their entanglement which is at issue in the Shor-Preskill framework. The amount of entanglement, and thus secret key, which can be distilled may be estimated by making use of the symmetries of the signal states and measurement. 

We can illustrate this most easily using the BB84 protocol itself and then describe how it can be made to work more generally. As discussed in Section~\ref{sec:qkd}, the BB84 protocol can be described coherently by pretending that Alice first creates EPR pairs and then sends one subsystem of each pair to Bob. Here, however, it is more appropriate to describe each signal sent by Alice as her preparation of the state
\begin{align}
\label{eq:qkdstart}
\ket{\psi_0}=\tfrac 12\sum_{j,k} \ket{j}^{A_K}\ket{k}^{A_S}\ket{\xi_{jk}}^{B},
\end{align} 
and transmission of the $B$ subsystem to Bob. The indices $j$ and $k$ specify the eigenvalue and observable, respectively, of the state $\ket{\xi_{jk}}$ transmitted by Alice; $k=0$ denotes amplitude $Z$ and $k=1$ phase, while the eigenvalue is given by $(-1)^j$. Bob makes a random measurement of the two observables, which can be described by the isometry $U_{\mathcal{M}}^{B\rightarrow B_KB_S}=\tfrac{1}{\sqrt{2}}\sum_{jk}\ket{j}^{B_K}\ket{k}^{B_S}\bra{\eta_{jk}}^B$, where here $\ket{\eta_{jk}}=\ket{\xi_{jk}}$ but the distinction will be useful later. For a noiseless channel, his measurement process results in the state
\begin{align}
\ket{\psi_1}^{A_KA_SB_KB_S}=\tfrac{1}{\sqrt{8}}\sum_{jj'kk'} \ket{j}^{A_K}\ket{k}^{A_S} \ket{j'}^{B_K}\ket{k'}^{B_S}\braket{\eta_{j'k'}|\xi_{jk}}.
\end{align}

From the form of the inner products $\braket{\eta_{j'k'}|\xi_{jk}}$ one can easily work out that if Alice and Bob each measure their $S$-labeled subsystems and obtain the same result, the remaining $K$-labeled systems are in the state $\ket{\Phi}^{A_KB_K}$ and thus measurement produces a secret key. This mimics the sifting process of the actual protocol, as Alice and Bob perform the measurements separately and compare their results by public discussion. 
Also crucial is the fact that the overall probability distribution for signals and measurement outcomes found here is precisely the same as in the prepare and measure scheme. Thus, this state has the form claimed above: It provides a coherent description of the real protocol in which Alice and Bob each have key $K$ and sifting $S$ subsystems, and sifting is accomplished by local measurement of the latter subsystems and postselection. 

Noisy channels require the additional steps of parameter estimation, information reconciliation, and privacy amplification, but change the above picture only slightly. Describing the channel resulting from Eve's attack by its decomposition into Kraus operators, and assuming the attack is collective, the state $\ket{\psi_1}$ is altered by the noise to
\begin{align}
\label{eq:psi1p}
\ket{\psi_1'}^{A_KA_SB_KB_SE}=\tfrac{1}{\sqrt{8}}\sum_{jj'kk'\ell} \ket{j}^{A_K}\ket{k}^{A_S} \ket{j'}^{B_K}\ket{k'}^{B_S}\ket{\ell}^E\braket{\eta_{j'k'}|E_\ell|\xi_{jk}}.
\end{align}

In the sifting stage, Alice and Bob keep only the cases in which $k=k'$ and subsequently discard the information specifying which value of $k$ they observed. We can model this process as keeping only the $k=k'$ terms in (\ref{eq:psi1p}) and then giving the $A_S$ and $B_S$ systems to Eve. Alice and Bob keep only the raw key, and the state becomes (slightly redefining $E$)
\begin{align}
\label{eq:postsift}
\ket{\psi_2'}^{A_KB_KE}\propto \sum_{jj'k\ell} c_{jj'}^{k\ell}\ket{j}^{A_K}\ket{j'}^{B_K}\ket{k,\ell}^E,\qquad c_{jj'}^{k\ell}=\braket{\eta_{j'k}|E_\ell|\xi_{jk}}.
\end{align}

Following the Shor-Preskill idea, as generalized in Section~\ref{sec:skd}, Alice and Bob can construct the information reconciliation and privacy amplification protocols necessary to turn the raw key into a secret key once they are able to estimate $p_{\rm guess}(Z^{A_K}|B_K)_{\psi_2'}$ and $p_{\rm guess}(X^{A_K}|C_ZB_K)_{\psi_2'}$. A bound on the former is given directly by parameter estimation, but the latter is not so straightforward. The joint state of the key systems is determined via the coefficients $c_{jj'}^{k\ell}$, creating a connection between the two guessing probabilities, albeit in general a not at all straightforward one. The structure of the sifting and of the signals and measurements greatly simplifies the connection, and makes it possible to find useful bounds on the latter guessing probability as a function of the former. This enables Alice and Bob to construct the remainder of the protocol to be provably secure. 

For BB84, one finds by direct calculation that $p_{\rm guess}(Z^{A_K}|Z^{B_K})=p_{\rm guess}(X^{A_K}|X^{B_K})$ regardless of the value of $\ell$. That is, the correlation in the amplitude basis (which gives the key itself) is precisely the same as the correlation in the phase basis (conjugate to the key). This was to be expected from the original coherent description of BB84 which explicitly uses EPR pairs from the beginning, since half the time the key comes from the original amplitude basis, and half the time from the phase basis, so the correlations ought to be the same. Using this relationship in the formula for the rate of secret key distillation, Equation~(\ref{eq:k1}) (ignoring $\mathcal{Q}$ and $T$), we recover the rate $r_{\rm BB84}=1-2h_2(\delta)$, for $\delta$ the observed error rate in the raw key and $h_2(\delta)=-\delta\log_2\delta-(1-\delta)\log_2(1-\delta)$ the binary entropy, which leads to the threshold of 11\%. Security against general coherent attacks is then ensured by the result of Christandl\etalsp\cite{christandl_postselection_2009}.

A great advantage of the above approach is the modularity of the security proof. The  details of the signals, measurements, and sifting are logically completely separate from the details of information reconciliation and privacy amplification. The former enter only into the coefficients $c_{jj'}^{k\ell}$, which are used to select a CSS code for the latter. This approach is developed in \citeme{renes_generalized_2006} as a generalization of that used by Tamaki\etalsp\cite{tamaki_unconditionally_2003} and Boileau\etalsp\cite{boileau_unconditional_2005}, and it is shown that it applies to a wide class of protocols, 
particularly those based on so-called \emph{equiangular spherical codes}. These are are constellations of pure states $\ket{\xi_j}$ whose pairwise overlaps are all identical, as in the three-state protocol of Phoenix~\cite{phoenix_three-state_2000} mentioned above, and were adapted for use in QKD generally by the author~\cite{renes_frames_2004,renes_spherical-code_2004,renes_equiangular_2005}. The other main contribution of \citeme{renes_generalized_2006} is the development of a method of exploiting the symmetries of the sifting scheme and the signal and measurement states to simplify this task, relying on results from group representation theory.

To see how this works, consider the protocol in which Alice's signals are four qubit states for which $|\braket{\xi_j|\xi_k}|^2=\frac{1}{3}$, as described in~\cite{renes_spherical-code_2004}. These form a regular tetrahedron in the Bloch-sphere representation of a qubit, and Bob's measurement is comprised of appropriately-normalized projectors onto the states $\ket{\eta_k}$ for which $\braket{\eta_k|\xi_k}=0$, i.e.\ the inverse tetrahedron in the Bloch-sphere. Due to symmetry, Bob's measurement would randomly reveal one state which Alice did not send if the channel were noiseless, and so the information exchanged by Alice in the sifting stage consists of a random choice of two states she did not send. 

In one-third of cases these two pieces of information specify which state she did send, and Bob publicly announces that he has successfully decoded the transmission. From this they generate one secret bit corresponding to which of the two signals Alice did send, given the public exclusion of two of the initial possibilities. There are 12 possible announcements by Alice, since she must also specify how the two possible signal states are to be decoded into the raw key, and we may label the signal states by the combination of sifting announcement and raw key value. In this way each signal is counted six times, but this presents no difficulty as each is counted the same number of times.  Much the same holds for Bob, and so the state in Equation~(\ref{eq:qkdstart}) can be used to describe the protocol coherently. 

The remaining task is to use the $c_{jj'}^{k\ell}$ to bound $p_{\rm guess}(X^{A_K}|C_ZB_K)$ in terms of $p_{\rm guess}(Z^{A_K}|Z^{B_K})$. By exploiting symmetries of the QKD protocol as in~\citeme{renes_generalized_2006}, we can greatly simplify this task. Suppose that the sifting step of the protocol is such that there exist unitaries $U_k$ and $V_k$ for which $\ket{\xi_{jk}}=U_k\ket{\xi_{j0}}$ and $\ket{\eta_{jk}}=V_k\ket{\eta_{j0}}$. Then the $c_{jj'}^{k\ell}$ become $
c_{jj'}^{k\ell}= \braket{\eta_{j'0}|V^\dagger_k E_\ell U_k|\xi_{j0}}$. Now let us focus on a particular Kraus operator $E_\ell$ by fixing the value of $\ell$, but average over the value of $k$, which corresponds to Alice and Bob throwing away the information specifying which particular sifting map they applied. Their shared state given the value of $\ell$ has the form
\begin{align}
\psi_\ell^{A_KB_K}\propto\sum_{ii'jj'}\ket{ii'}\bra{jj'}^{A_KB_K}\sum_k \braket{\eta_{i'0}|V^\dagger_k E_\ell U_k|\xi_{i0}}\braket{\xi_{j'0}|U^\dagger_k E_\ell^\dagger V_k|\eta_{j0}}.
\end{align}

Examining the form of the matrix elements, we see that the sifting symmetries $U_k$ and $V_k$ create an effective channel having Kraus operators $V^\dagger_k E_\ell U_k$. Moreover, the group nature of these operators enables us to compute the action of the channel by appealing to representation theory.  In the particular case of the tetrahedral protocol, one finds that the 
effective channel is just a depolarizing channel, irrespective of the value $\ell$. The depolarizing rate can be determined by the noise rate observed in the parameter estimation phase.
 Computing the state after the sifting step reveals that Alice and Bob can describe their shared key state by a Bell-diagonal state $\psi^{A_KB_K}=\sum_{jk}p_{jk}\ket{\beta_{jk}}\bra{\beta_{jk}}^{A_KB_K}$, as in Equation~(\ref{eq:belldiagonal}), with the $p_{jk}$ satisfying $p_{01}=p_{11}=2p_{10}$. 

This implies $p_{\rm guess}(Z^{A_K}|Z^{B_K})=\delta$ and $p_{\rm guess}(X^{A_K}|X^{B_K},Z^{A_K}=Z^{B_K})=\frac 13$ while $p_{\rm guess}(X^{A_K}|X^{B_K},Z^{A_K}\neq Z^{B_K})=1-2\delta/3(1-\delta)$. The latter guessing probabilities are directly related to $p_{\rm guess}(X^{A_K}|C_ZB)$ since Bob's knowledge of $Z^{A_K}$ stored in $C_Z$ can be equivalently thought of as the information as to whether or not an amplitude error $Z^{A_K}\neq Z^{B_K}$ occurred or not. Using these guessing probabilities in Equation~(\ref{eq:k1}), we obtain the rate $r_{\rm tetra}=1-h_2(\delta)-\delta h_2(\frac 13)-(1-\delta)h_2(2\delta/3(1-\delta))$, which has a threshold of 11.56\%.  In~\citeme{renes_generalized_2006} the method is applied to several other spherical code protocols with signal and measurement states having Hilbert space dimension three. 

\section{Private States in Quantum Key Distribution}
\label{sec:psqkd}
By clever if perhaps unintuitive choice of \emph{preprocessing} operations $\mathcal{Q}$ in Equations~(\ref{eq:k1}) and (\ref{eq:regK}) the error thresholds of QKD can be pushed higher than those found by the Shor-Preskill method alone. Understanding how this can be the case requires interpreting QKD as a virtual means of private state distillation rather than just entanglement distillation, as first shown in~\citeme{renes_noisy_2007}. Furthermore, the private state distillation approach suggests that it would be beneficial to combine two types of preprocessing operations previously studied, and this was indeed shown to be the case for the BB84 protocol in~\citeme{smith_structured_2008}. Further improvements and an extension of the method to the six-state protocol were reported in~\citeme{kern_improved_2008}, and we describe both of these results here. 

That private state distillation is actually needed to give a fully quantum-mechanical description of QKD was necessitated by the work of Kraus, Gisin, and Renner~\cite{kraus_lower_2005,renner_information-theoretic_2005}. They established the seemingly-paradoxical result that the noise threshold of BB84 can be improved if Alice randomly flips some of her raw key bits before performing the final three steps of the protocol, and they reported a threshold improvement from 11\% to 12.4\%. From the viewpoint of QKD as a virtual scheme for entanglement distillation this additional step would seem to be counterproductive, as noise inflicted by Alice behaves the same as noise inflicted by Eve. However, we saw in Chapter~\ref{chap:char} that entanglement is not actually necessary for secret key creation, private states are. This raises the question of whether or not one can view the noisy preprocessing step as part of a virtual scheme for private state distillation, which \citeme{renes_noisy_2007} answers in the affirmative. 

The crux of understanding such noisy preprocessing in a private state picture is to include the system Alice uses to impart the noise to her raw key and observe that it functions as a shield system. The overhead in the protocol of additional information reconciliation needed due to the noisy preprocessing is then more than made up for by a reduction in the required amount of privacy amplification. The particular guessing probabilities found in the previous section imply that the state of Alice's and Bob's raw keys immediately after the sifting stage takes a Bell-diagonal form in which the probabilities of amplitude and phase error are independent and equal. That is, in the state $\psi^{A_KB_K}=\sum_{jk}p_{jk}\ket{\beta_{jk}}\bra{\beta_{jk}}^{A_KB_K}$, one has $p_{00}=(1-\delta)^2$, $p_{10}=p_{01}=\delta(1-\delta)$, and $p_{11}=\delta^2$ for $\delta$ the probability of amplitude (or phase) error. 

Now suppose that Alice randomly flips each raw key bit independently with some probability $q$. This process may be modelled as a \textsc{cnot} gate whose control is an ancillary system $A'$ prepared by Alice in the state $\ket{\varphi}=\sqrt{1-q}\ket{0}+\sqrt{q}\ket{1}$ and whose target is her raw key $A_K$. The error rate in Alice's and Bob's keys has jumped to $\delta'=\delta(1-q)+q(1-\delta)$, but the crucial difference from the entanglement distillation scenario is that for security it is relevant how well $A'$ and $B_K$ together can predict $X^{A_K}$, not merely how well $B_K$ could alone. 

The resulting state of $A_KA'B_K$ can be used to compute $H(X^{A_K}|A'B)$ for use in Equation~(\ref{eq:k1}); observe that we do not need to make use of the $C_Z$ system here because knowing if there is an amplitude error tells Bob nothing about the likelihood of a phase error. Using Equation~(\ref{eq:k1}) and optimizing over the choice of $q$ we recover the threshold of 12.4\%. A similar calculation (now requiring the use of $C_Z$) recovers the six-state threshold of 14.1\%. Actually, \citeme{renes_noisy_2007} follows a different approach than what we have outlined here, directly constructing the twisting operator, but this can be seen as a particular case of the general results on private states and secret key distillation presented in Sections~\ref{sec:secretkey} and~\ref{sec:skd}. 

In his security proof of the six-state protocol, Lo observed~\cite{lo_proof_2001} that the threshold can be improved from the nominal 12.6\% one would find following the Shor-Preskill method to 12.7\% by employing so-called \emph{degenerate} quantum error-correcting codes first discussed by DiVincenzo, Shor, and Smolin~\cite{divincenzo_quantum-channel_1998}. This code consists of a concatenation of an amplitude repetition code with a random CSS code and in the present context corresponds to a preprocessing operation $\mathcal{Q}$ on blocks of inputs, as in Equation~(\ref{eq:regK}). 

The degeneracy of the code refers to the fact that several different errors can share the same recovery operation and the syndrome need only reveal which recovery operation is required, a phenomenon which is not possible for classical error-correcting codes. For example, in the amplitude repetition code of Section~\ref{subsec:compqecc}, the three possible phase errors acting on single qubits all have the same effect on the encoded quantum information, namely as a phase flip. Thus, if we concatenate the repetition code with another code, we need not determine the precise location of a phase error on the physical qubits. Reducing the number of stabilizers needed to enable correction of phase errors implies a reduction in the necessary amount of privacy amplification in the context of QKD, and thus the threshold increases. 

Shor's nine-qubit code described in Section~\ref{subsec:shor9} provides a simple example. There we considered the effect of a single phase flip error on the fourth qubit and found that it would be detected by measuring certain stabilizer operators. But it is clear from the argument there that the same result is obtained for a phase error on either qubit five or six. This is reflected in the fact that associated with the code are are six amplitude stabilizers and only two phase stabilizers. The former determine the precise location of an amplitude error, but the latter only fix the location of the phase error up to the position in the block. This is all that is necessary. 

It is possible to combine the noisy preprocessing discussed above with degenerate codes to improve the threshold of BB84 still further, as described in~\citeme{smith_structured_2008}. The original protocol is modified as follows. After the raw key is created in the sifting phase, Alice performs a noisy preprocessing step in which she independently flips each raw key bit with some probability $q$. Then she computes the syndromes of an amplitude repetition code encoding one qubit into $m$ qubits, i.e.\ $z_1\oplus z_2, z_1\oplus z_3,\dots,z_1\oplus z_m$ and transmits these publicly to Bob. The first bit of each block she saves for further use as the key. Bob then computes the syndromes of his block, and attempts to correct his key bit so that the syndromes match Alice's, exactly as was done in the entanglement distillation protocol discussed in Section~\ref{sec:ed1}. Alice and Bob then repeat this process for many blocks, collecting one key bit per block. On these refined key bits they then perform information reconciliation and privacy amplification as needed. 

To determine the threshold, for which the main difficulty is, as usual, to determine the amount of privacy amplification needed, it is simpler to focus on Eve's states and compute $H(\overline{A}_K|ES)$, where $S$ denotes the syndrome information and $\overline{A}_K$ the key bit encoded in the repetition code. Again the symmetries of the problem enable the use of group representation theory to make the calculation numerically tractable, allowing thresholds for blocklengths in the hundreds to be determined. The best threshold found in~\citeme{smith_structured_2008} is 12.9\%, corresponding to $q\approx 0.32$ and $m=400$. A more elaborate analysis is required for the six-state protocol, and this is carried out in~\citeme{kern_improved_2008}, with the result that the threshold is at least 14.59\%. Additionally, the effects of iterating the entire noisy preprocessing plus repetition code procedure are investigated therein, and this is found to offer substantial increases in the key distribution rate of the protocol at high error rates, though the overall threshold is not as large.  

As mentioned previously, the use of repetition codes is a type of blockwise preprocessing, in contrast to the noisy preprocessing which is applied to single key bits. As  blockwise preprocessing is more complicated, and the expression for the optimal rate for secret key distillation, Equation~(\ref{eq:regK}), essentially impossible to evaluate, the question arises whether blockwise preprocessing, i.e.\ regularization are truly necessary. Unfortunately, the answer is yes, as observed in~\citeme{smith_structured_2008}. One can show that the threshold found by noisy preprocessing, 12.4\%, is the optimal threshold using single-bit, or single-letter preprocessing. Since the combination of noisy processing and repetition codes leads to a higher threshold, regularization must in general be necessary. This result then applies to the private capacity of a channel as well, since one way to communicate privately is to first generate secret keys and then encrypt the actual message to be sent. 

Thus, neither the secret key distillation rate nor the private capacity are \emph{single-letterizeable} quantities. This reveals a large distinction between classical and quantum information theory, as single-letter quantities are usual in the former, reflecting the fact that the random coding arguments of Shannon are optimal in a wide variety of situations. In quantum information theory this is no longer true. The degenerate codes found by DiVincenzo, Shor, and Smolin~\cite{divincenzo_quantum-channel_1998} show that the quantum capacity is also not single-letterizeable, while Hastings has recently established that the classical capacity of a quantum channel is not single-letterizeable either~\cite{hastings_superadditivity_2009}. Despite the apparent similarities with classical information theory, a full understanding of quantum information theory will require the development of tools beyond the usual random coding methods.

\chapter{Summary and Outlook}
The preceding six chapters demonstrate that, far from being just an abstract mathematical study, the study of quantum information theory is quite closely connected to core physical concepts, namely complementarity and the uncertainty principle. Indeed, although we have presented the topics of this thesis in a logical order, it was actually research into secret key distillation in~\citeme{renes_physical_2008} that led to the conjecture of the entropic uncertainty principle of Equation~(\ref{eq:jcbjmr}) in~\citeme{renes_conjectured_2009} and its eventual proof in~\citeme{berta_uncertainty_2010}. 

The results described in this thesis spring from trying to make sense of what it means to have ``quantum'' information, working within the formalism of quantum theory itself. Using the conditional entropy $H(Z^A|B)_\psi$ we can describe the information held by $B$ about the amplitude measurement $Z$ on system $A$, when $A$ and $B$ are jointly in the quantum state $\psi^{AB}$. Having quantum information then refers to the situation in which $B$ implicitly contains information about two complementary observables $X$ and $Z$, and the uncertainty principle in the form $H(X^A|B)+H(Z^A|C)\geq \frac{1}{c}$ 
constrains the extent to which information about both can be simultaneously explicitly realized. Quantum information processing protocols can then be constructed by mimicking related classical information processing protocols for the two complementary pieces of classical information, taking care not to violate the uncertainty principle. 

Although complementarity is at the heart of the results presented herein, to complete the proofs we have also relied heavily on certain algebraic properties both of the observables $X$ and $Z$ as defined in Equation~(\ref{eq:wh}) and of the attendant CSS stabilizer codes. In particular, the algebraic properties of the amplitude and phase observables play important roles in Theorems~\ref{thm:entdec}, \ref{thm:2xdecoupling}, \ref{thm:entropic}, and \ref{thm:compps}, while the algebraic structure of CSS codes is used extensively throughout Chapters~\ref{chap:proc}, \ref{chap:duality}, and \ref{chap:qkd}. 
Removing the algebraic requirement on the observables is precisely the difference between the uncertainty principle results of~\citeme{renes_conjectured_2009} and~\citeme{berta_uncertainty_2010}, and a major goal of future work is to remove this requirement from the aforementioned results as well. The situation is akin to difference between the heuristic use of the uncertainty principle in the early proofs of QKD, where the uncertainty principle provided guidance for the actual algebraic arguments, and the recently formulated BB84 security proof of Tomamichel\etalsp\cite{tomamichel_uncertainty_2011,tomamichel_tight_2011} based directly on the uncertainty principle formulated in terms of smooth-entropy.

This goal is likely to be fairly straightforward for the results of Chapter~\ref{chap:char}, but the use of CSS codes in the protocols of the subsequent chapters appears much more central to those results. The difficulty lies in the need to combine classical protocols for complementary observables in such a way that all the important quantities can actually simultaneously exist, i.e.\ the corresponding operators all commute. In the entanglement distillation scheme of Chapter~\ref{chap:proc} for instance, the use of CSS codes ensures that the syndrome information needed to establish strong phase correlations does not interfere with either the amplitude syndromes nor the final encoded amplitude. 

Another goal of future work will be to extend all the results beyond the realm of asymptotic i.i.d.\ resources and into the one-shot domain of structureless resources briefly described in Section~\ref{subsec:oneshot}. Here we have presented optimal protocols in the former scenario, but it is not clear whether this will be possible in the more general setting. One cause for hope is that the uncertainty principle already plays a fundamental role in the one-shot setting. Tomamichel\etalsp\cite{tomamichel_duality_2010} have shown that the smooth min- and max-entropies are not independent: One may be defined in terms of the other using a purification system. The smooth entropy uncertainty relation then follows 
from this duality~\cite{tomamichel_uncertainty_2011}. 

Finally, a much more ambitious goal is to extend the notion of quantum information as complementary classical information past the simple two-party communication scenarios studied here. Can this point of view shed some light into how quantum computers work?

\cleardoublepage
\backmatter

\bibliographystyle{habi.bst}
\bibliography{habi}
   

\cleardoublepage

\bibliographystyleme{alpha}
\bibliographyme{habi-me}
\cleardoublepage

\end{document}